\documentclass[11pt]{article}

\usepackage[utf8]{inputenc}
\usepackage[colorlinks=true
]{hyperref}
\hypersetup{urlcolor=blue, citecolor=red, linkcolor=blue}
\usepackage{mleftright}
\usepackage{stmaryrd}
\SetSymbolFont{stmry}{bold}{U}{stmry}{m}{n}
\usepackage{microtype}

\usepackage{bbm}
\usepackage{colortbl}

\usepackage{graphicx}
\usepackage[active]{srcltx}

\usepackage{amsmath,amssymb}

\numberwithin{equation}{section}

\usepackage{physics}
\usepackage{bm}
\usepackage{enumerate}
\usepackage{amsthm}
\usepackage[all,2cell]{xy}
\usepackage{mathrsfs}
\usepackage{mathtools}
\mathtoolsset{showonlyrefs}

\usepackage{tikz-cd}
\usetikzlibrary{cd}

\usepackage{xcolor}
\DeclareUnicodeCharacter{0301}{*************************************}

\hypersetup{urlcolor=blue, citecolor=red, linkcolor=blue}

\usepackage[a4paper, left=25mm, right=25mm, top=30mm, bottom=25mm]{geometry}

\usepackage{marginnote}

\usepackage{imakeidx}
\makeindex[intoc]

\usepackage[nottoc]{tocbibind}
\usepackage[colorinlistoftodos]{todonotes} 

\theoremstyle{plain}
\newtheorem{theorem}{Theorem}[section]
\newtheorem{proposition}[theorem]{Proposition}

\newtheorem{lemma}[theorem]{Lemma}

\theoremstyle{definition}
\newtheorem{definition}[theorem]{Definition}
\newtheorem{remark}[theorem]{Remark}
\newtheorem{example}[theorem]{Example}

\AtBeginEnvironment{remark}{%
  \pushQED{\qed}%
}
\AtEndEnvironment{remark}{\popQED\endexample}

\AtBeginEnvironment{example}{%
  \pushQED{\qed}%
}
\AtEndEnvironment{example}{\popQED\endexample}

\newcommand\restr[2]{{
  \left.\kern-\nulldelimiterspace #1 \right|_{#2} 
}}

\newcommand*{\transp}[2][-3mu]{\ensuremath{\mskip1mu\prescript{\smash{\mathrm t\mkern#1}}{}{\mathstrut#2}}}%

\newcommand*{\longhookrightarrow}{\ensuremath{\lhook\joinrel\relbar\joinrel\rightarrow}}
\newcommand{\R}{\mathbb{R}}
\newcommand{ \Rk}{\mathbb{R}^k}

\renewcommand{\d}{\mathrm{d}}

\newcommand{\df}{\Omega}
\newcommand{\Cinfty}{\mathscr{C}^\infty}
\newcommand{\T}{\mathrm{T}}
\newcommand{\Tan}{\mathrm{T}}
\newcommand{\cT}{\mathrm{T}^\ast}

\newcommand*{\inn}[1]{\iota_{#1}}
\newcommand*{\innp}[1]{\iota\left(#1\right)}

\newcommand{\Lie}{\mathscr{L}}
\renewcommand{\L}{\mathcal{L}}
\newcommand{\vf}{\mathfrak{X}}
\newcommand{\X}{\mathfrak{X}}
\newcommand{\Xh}{\mathfrak{X}_{\rm Ham}}

\newcommand{\F}{\mathcal{F}}
\renewcommand{\H}{\mathcal{H}}

\newcommand{\Reeb}{R}

\newcommand{\C}{\mathcal{C}}
\newcommand{\D}{\mathcal{D}}

\newcommand*{\bd}{\overline{\mathrm{d}}}

\newcommand{\rmC}{\mathrm{C}}
\newcommand{\rmR}{\mathrm{R}}
\newcommand{\rmS}{\mathrm{S}}

\newcommand{\bfX}{\mathbf{X}}

\newcommand{\parder}[2]{\frac{\partial #1}{\partial #2}}
\newcommand{\dparder}[2]{\dfrac{\partial #1}{\partial #2}}
\newcommand{\tparder}[2]{\partial #1/\partial #2}
\newcommand{\parderr}[3]{\frac{\partial^2 #1}{\partial #2\partial #3}}
\newcommand{\dparderr}[3]{\dfrac{\partial^2 #1}{\partial #2\partial #3}}

\DeclareMathOperator{\rk}{rank}

\DeclareMathAlphabet{\mathpzc}{OT1}{pzc}{m}{it}
\def\d{\mathrm{d}}

\let\ds\displaystyle

\usepackage{xcolor}
\usepackage{fancyhdr}

\begin{document}


\vspace{3em}

{\huge\sffamily\raggedright A survey on geometric frameworks \\[.5ex] for action-dependent classical field theories \\[1ex] and their relationship}
\vspace{2em}

{\large\raggedright
    \today
}

\vspace{3em}

{\Large\raggedright\sffamily
    Jordi Gaset-Rifà
}\vspace{1mm}\newline
{\raggedright
    Department of Mathematics, CUNEF Universidad\\
    Calle Almansa 101,
28040 Madrid, Spain\\
    e-mail: \href{mailto:jordi.gaset@cunef.edu}{jordi.gaset@cunef.edu} --- {\sc orcid}: \href{https://orcid.org/0000-0001-8796-3149}{0000-0001-8796-3149}
}

\medskip

{\Large\raggedright\sffamily
    Xavier Rivas
}\vspace{1mm}\newline
{\raggedright
    Department of Computer Engineering and Mathematics, Universitat Rovira i Virgili\\
    Avinguda Països Catalans 26, 43007 Tarragona, Spain\\
    e-mail: \href{mailto:xavier.rivas@urv.cat}{xavier.rivas@urv.cat} --- {\sc orcid}: \href{https://orcid.org/0000-0002-4175-5157}{0000-0002-4175-5157}
}

\medskip

{\Large\raggedright\sffamily
    Narciso Román-Roy
}\vspace{1mm}\newline
{\raggedright
    Department of Mathematics, Universitat Politècnica de Catalunya\\
    C. de Jordi Girona 31, 08034 Barcelona, Spain\\
    e-mail: \href{mailto:narciso.roman@upc.edu}{narciso.roman@upc.edu } --- {\sc orcid}: \href{https://orcid.org/0000-0003-3663-9861}{0000-0003-3663-9861}
}

\vspace{3em}

{\large\bf\raggedright
    Abstract
}\vspace{1mm}\newline
{\raggedright
This work presents a comprehensive overview of three recently developed geometric frameworks for the study of classical action-dependent field theories. Specifically, the three underlying geometric structures 
— namely, $k$-contact, $k$-cocontact, and multicontact — are first introduced, and then used to develop the Lagrangian and Hamiltonian formalisms of the aforementioned theories. Finally, the relationship among these three types of structures is analyzed in the case of trivial bundles;
as well as the comparison with other alternative definitions of multicontact structure presented in the literature.
}

\vspace{3em}


\medskip

{\sl Secondary:} 35Q99, 53C15, 53Z05, 58A10, 70G45.

\medskip






\noindent {\bf Authors' contributions:} All authors contributed to the study conception and design. The manuscript was written and revised by all authors. All authors read and approved the final version.
\medskip

\noindent {\bf Competing Interests:} The authors have no competing interests to declare. 

\newpage

{\setcounter{tocdepth}{2}
\def\baselinestretch{1}
\small
\def\addvspace#1{\vskip 1pt}
\parskip 0pt plus 0.1mm
\tableofcontents
}

\pagestyle{fancy}

\fancyhead[L]{Geometry of action-dependent classical field theories}    
\fancyhead[C]{}                  
\fancyhead[R]{J. Gaset, X. Rivas, and N. Román-Roy}       

\fancyfoot[L]{}     
\fancyfoot[C]{\thepage}                  
\fancyfoot[R]{}            

\setlength{\headheight}{17pt}

\renewcommand{\headrulewidth}{0.1pt}  
\renewcommand{\footrulewidth}{0pt}    

\renewcommand{\headrule}{%
    \vspace{3pt}                
    \hrule width\headwidth height 0.4pt 
    \vspace{0pt}                
}

\setlength{\headsep}{30pt}  

\section{Introduction}

In the second half of the 20th century, {\sl symplectic geometry} proved to be a highly effective geometric framework for formalizing analytical mechanics. This success is well documented in numerous classical treatises on the subject (see, for instance, \cite{AM_78,Arn_89,LR_89,Holm_11,LM_87,MR_25}, and the references therein).

More recently, particularly during the early 21st century, there has been growing interest in the use of {\sl contact geometry} \cite{BH_16,Gei_08,Kho_13}
to describe a specific class of mechanical systems: those exhibiting dissipation or, equivalently, non-conservative behavior (see \cite{Bra_17,CG_19} for a motivated introduction). 
Beyond this, contact geometry has found broader applications in modeling various physical theories, such as thermodynamics, quantum mechanics, electric circuits, and control theory \cite{Bra_18,CCM_18,Kho_13}.
This renewed interest has led to a significant body of literature. 
For example, contact Hamiltonian systems have been studied in \cite{BCT_17,BLMP_20,LL_19,GG_22}, 
while their Lagrangian counterparts are addressed in \cite{CCM_18,LL_19a,GGMRR_20a}. 
Other relevant developments include non-autonomous systems \cite{LGGMR_23,GLR_23,RT_23}, quantization \cite{CCM_18}, and variational formulations dating back to the original work of G. Herglotz \cite{Her_30,Her_85}. 
(This list of references is by no means exhaustive.)

Just as mechanical systems can exhibit non-conservative dynamics, so too can classical field theories. 
These are generally known in physics as {\sl action-dependent theories}, which are extensions of standard models in which the corresponding Lagrangian and Hamiltonian functions incorporate additional variables related to the action. 
This results in extra terms in the dynamical or field equations, which can be interpreted as encoding dissipative effects
(although the applications of these theories extend well beyond dissipation).

In the conservative case, several geometric frameworks, generalizing symplectic geometry, have been developed to describe classical field theories. 
Among them are the so-called {\sl $k$-(co)symplectic} and {\sl polysymplectic} formalisms, introduced in \cite{Awa_92,Awa_94,AG_00} and later expanded and applied to describe Lagrangian and Hamiltonian field theories (see \cite{LSV_15,GMS_97,Kan_98} and references therein). 
However, the most general framework is {\sl multisymplectic geometry} \cite{Kij_73,KT_79}, for which there exists extensive literature. 
We refer to \cite{GIMM_98,Rom_09,RW_19} as general sources for its application to field theories and,
in particular, \cite{AA_80,Gar_74,Sau_89} for the Lagrangian formalism, 
\cite{CCI_91,FPR_05,GS_73,HK_02} for the Hamiltonian setting, and \cite{LMM_96a} for the singular case.

In the context of action-dependent field theories, recent efforts have been made to construct geometric frameworks analogous to the $k$-symplectic, $k$-cosymplectic, and multisymplectic formalisms. 
In particular, in \cite{LRS_24,GGMRR_20,GGMRR_21,GRR_22,Riv_23} the so-called {\sl $k$-contact} and {\sl $k$-cocontact} structures are introduced as natural extensions of contact geometry, built upon the $k$-symplectic and $k$-cosymplectic foundations.
Additionally, the fusion of contact and multisymplectic frameworks has recently led to the definition of the {\sl multicontact} structure, proposed in \cite{LGMRR_23,LGMRR_25,LIR_25,RRZ_25}. 
Other less general approaches to similar geometric frameworks appear in \cite{Bo_96,Fi_24,Mo_08,TV_08}; in particular, we highlight the different version of {\sl multicontact manifolds} presented in \cite{Vit_15}.
As in the case of contact mechanics, the field equations for action-dependent field theories can be derived from a variational principle \cite{GGB_03,LPAF_17}, and for a precise and
general variational formulation in this multicontact framework, see \cite{LGMRR_23,GLMR_24,GM_23a}.

The aim of this paper is twofold. First, we review the geometric formulations previously introduced for first-order action-dependent field theories, namely the $k$-contact, $k$-cocontact, and multicontact frameworks, as developed in \cite{LGMRR_23,LGMRR_25,LIR_25,GGMRR_20,GGMRR_21,Riv_23}. 
In this review, we restrict our attention to the case of regular theories, that is, those defined by regular Lagrangian functions. 
For a detailed analysis of singular cases, we refer the reader to the aforementioned references.
Second, we establish a correspondence between these geometric structures in the particular setting where the phase bundles associated with the field theories are assumed to be trivial.

The paper is organized as follows: Sections \ref{kcontactf}, \ref{kcocontactf}, and \ref{multicontactf} provide a review of the $k$-contact, $k$-cocontact, and multicontact formulations, respectively. 
In each case, we first introduce the underlying geometric structure and then develop the corresponding Lagrangian and Hamiltonian formalisms for action-dependent field theories.
Sections \ref{relation} and \ref{6} contain the main original contributions of this work. In particular, Section \ref{relation} is devoted to studying the relationship between the geometric frameworks mentioned above under the assumption of trivial phase bundles and, in Section \ref{6}, our notion of a multicontact structure is compared to the one given in
\cite{Vit_15}.

All manifolds are real, second-countable, and of class $\Cinfty$, and the mappings are assumed to be smooth.
Sum over crossed repeated indices is understood.

The following notation will be used throughout, adhering to standard conventions:
\vspace{-10pt}

\begin{itemize}\itemsep1pt
    \item $\Cinfty({\cal M})$: Smooth functions in a manifold ${\cal M}$.

    \item $\df^k({\cal M})$: Module of differential forms of degree $k$ in a manifold ${\cal M}$.

    \item $\vf({\cal M})$: Module of vector fields in a manifold ${\cal M}$.

    \item $\vf^k({\cal M})$: Module of $k$-multivector fields in a manifold ${\cal M}$.
    
    \item $\innp{X}\Omega$ or $\inn{X}\Omega$: Inner contraction of a vector field $X\in\vf({\cal M})$ and a $k$-form $\Omega\in\df^k({\cal M})$.

    \item $\Lie(X)$ or $\Lie_X$: Lie derivative by a vector field $X\in\vf({\cal M})$.

    \item $\d$: Exterior differential of differential forms.

\end{itemize}

\section{\texorpdfstring{$k$}{}-Contact field theories}
\label{kcontactf}

In this section, we present the most simple geometric description of action-dependent field theories, using a new framework that is an evolution of the $k$-symplectic formulation of classical field theories and contact mechanics. These kinds of formulation are specific for field theories that have the peculiarity that the Lagrangian or Hamiltonian functions describing them are independent of the space-time coordinates (or those analogous to these).

\subsection{\texorpdfstring{$k$}{}-Contact structures and \texorpdfstring{$k$}{}-contact Hamiltonian systems}

The $k$-contact (and $k$-precontact) manifolds $({\cal M},\eta^\alpha)$
and Hamiltonian systems defined on them were introduced in \cite{GGMRR_20}, where you can find more details on their definitions and properties.

\subsubsection{\texorpdfstring{$k$}{}-Contact structures}

Given a $N$-dimensional manifold $\mathcal{M}$, recall that,
for every non-vanishing, differential 1-form $\eta \in \df^1(\mathcal{M})$, its annihilator is a distribution of corank $1$, denoted $\langle \eta \rangle^\circ \subset \Tan \mathcal{M}$, which can be described as the kernel of the vector bundle morphism $\widehat \eta \colon \Tan \mathcal{M} \to \mathcal{M} \times \R$ defined by~$\eta$. Furthermore, $\eta$ generates a regular codistribution of rank $1$, denoted by
$\langle \eta \rangle \subset \Tan^*\mathcal{M}$.

The following definition generalizes the concept of {\sl contact structure} (which is recovered as a particular case, when $k=1$):

\begin{definition}\label{def:kcontact}
Given $k$ differential 1-forms
$\eta^1, \ldots, \eta^k \in \df^1(\mathcal{M})$,
consider the following associated distributions and codistributions:
\begin{align}
\mathcal{C}^{\mathrm{C}} &=
\langle \eta^1, \ldots, \eta^k \rangle \subset
\Tan^*\mathcal{M} \,,\\
\mathcal{D}^{\mathrm{C}} &=
\left( \mathcal{C}^{\mathrm{C}} \right)^\circ =
\ker \widehat{\eta^1} \cap \dotsb \cap \ker \widehat{\eta^k} \subset
\Tan \mathcal{M} \,, \\
\mathcal{D}^{\mathrm{R}} &=
\ker \widehat{\d \eta^1} \cap \dotsb \cap \ker \widehat{\d \eta^k} \subset
\Tan \mathcal{M} \,, \\
\mathcal{C}^{\mathrm{R}} &= 
\left( \mathcal{D}^{\mathrm{R}} \right)^\circ\subset
\Tan^*\mathcal{M} \,.
\end{align}
\label{kconman}
The family $\{\eta^\alpha\}$ is said to be a
{\sl\textbf{$k$-contact structure}} on $\mathcal{M}$ if:
\begin{enumerate}[(i)]
\item 
$\mathcal{D}^{\mathrm{C}} \subset \Tan \mathcal{M}$
is a regular distribution of corank~$k$;
or, what is equivalent,
$\eta^1 \wedge \dotsb \wedge \eta^k \neq 0$, at every point.
\item 
$\mathcal{D}^{\mathrm{R}} \subset \Tan \mathcal{M}$
is a regular distribution of rank~$k$.
\item
$\mathcal{D}^{\mathrm{C}} \cap \mathcal{D}^{\mathrm{R}} = \{0\}$  or, what is equivalent,
$\displaystyle
\bigcap_{\alpha=1}^{k} \left(
\ker \widehat{\eta^\alpha} \cap \ker \widehat{\d \eta^\alpha}\right) =\{0\}$.
\end{enumerate}
A {\sl\textbf{$k$-contact manifold}} is a manifold $\mathcal{M}$ endowed with a $k$-contact structure
and is denoted $(\mathcal{M},\eta^\alpha)$, $1\leq\alpha\leq k$.
We call
$\mathcal{C}^{\mathrm{C}}$
the {\sl\textbf{contact codistribution}},
$\mathcal{D}^{\mathrm{C}}$ 
the {\sl\textbf{contact distribution}},
$\mathcal{D}^{\mathrm{R}}$ 
the {\sl\textbf{Reeb distribution}},
and
$\mathcal{C}^{\mathrm{R}}$
the {\sl\textbf{Reeb codistribution}}
of the $k$-contact structure.
\end{definition}

\begin{remark}
If conditions (i) and (ii) hold, then (iii) is equivalent to 
$$
{\it (iii\,')}\ 
\Tan \mathcal{M} = \mathcal{D}^{\mathrm{C}} \oplus \mathcal{D}^{\mathrm{R}}  .
$$
\end{remark}

\begin{theorem}
\label{reebvf}
Let  $(\mathcal{M},\eta^\alpha)$ be a $k$-contact manifold.
\begin{enumerate}[{\rm(1)}]
\item
The Reeb distribution $\mathcal{D}^{\mathrm{R}}$ is involutive and therefore integrable.
\item
There exist $k$ vector fields 
$\Reeb_\alpha \in \vf(\mathcal{M})$,
called  {\sl\textbf{Reeb vector fields}},
which are uniquely defined by the relations
\begin{equation}
\label{reebcontact}
\inn{\Reeb_\beta} \eta^\alpha = \delta^\alpha_{\,\beta}
\,,\qquad 
\inn{\Reeb_\beta} \d\eta^\alpha = 0
\,.
\end{equation}
\item
The Reeb vector fields commute,
$\displaystyle [\Reeb_\alpha,\Reeb_\beta] = 0$,
and they generate the  Reeb distribution 
$\mathcal{D}^{\mathrm{R}}$.
\end{enumerate}
\end{theorem}

\begin{proposition}
Let $(\mathcal{M},\eta^\alpha)$ be a $k$-contact manifold.
Around every point of $\mathcal{M}$, there is a local chart  of coordinates $(U;z^I;s^\alpha)$, $U\subset \mathcal{M}$, such that
$$
    \Reeb_\alpha\vert_U = \frac{\partial}{\partial s^\alpha}  \,, \qquad
    \eta^\alpha\vert_U = \d s^\alpha - f_I^\alpha(z^J) \,\d z^I  \,,
$$
which are called  {\sl\textbf{adapted coordinates}} (to the $k$-contact structure).
\end{proposition}

The existence of {\sl canonical coordinates} is only assured
for a particular kind of $k$-contact manifolds:

\begin{theorem}[Darboux theorem for $k$-contact manifolds]
\label{Darboux k-contact}
Let  $(\mathcal{M},\eta^\alpha)$ be a $k$-contact manifold of dimension $n+kn+k$
such that there exists an integrable subdistribution ${\cal V}$ of ${\cal D}^{\rm C}$
with ${\rm rank}\,{\cal V}=nk$.
Then, around every point of $\mathcal{M}$, there exists a local chart of coordinates 
$(U;y^a,p^\alpha_a,s^\alpha)$, $1\leq\alpha\leq k \,,\ 1\leq a \leq n$,
 such that
$$
\eta^\alpha\vert_U=\d s^\alpha-p_a^\alpha\,\d y^a \,,\qquad
{\cal D}^{\rm R}\vert_U=\left\langle\Reeb_\alpha=\frac{\partial}{\partial s^\alpha}\right\rangle 
\,, \qquad
{\cal V}\vert_{U}=\left\langle\frac{\partial}{\partial p_a^\alpha}\right\rangle
\,.
$$
They are called the {\sl\textbf{canonical}} or {\sl\textbf{Darboux coordinates}} of the $k$-contact manifold.
\end{theorem}

The following example constitutes the canonical model for these kinds of $k$-contact manifolds.

\begin{remark}
\label{example-canmodel}
Given $k \geq 1$, let $Q$ be a $n$-dimensional differentiable manifold;
consider the vector bundle 
$\oplus^k \Tan^*Q:= \Tan^*Q \oplus \stackrel{k}{\dotsb} \oplus \Tan^*Q$,
which is called the {\sl\textbf {$k$-cotangent bundle}} or {\sl\textbf{bundle of $k^1$-momenta}} of $Q$.
Then, the manifold ${\cal M}=(\oplus^k \Tan^*Q) \times \Rk$
has a canonical $k$-contact structure defined by the 1-forms
$$
\eta^\alpha = \d s^\alpha - \theta^\alpha
\,,
$$
where $s^\alpha$ is the $\alpha$-th cartesian coordinate of $\Rk$,
and $\theta^\alpha$ is the pull-back of the canonical 1-form of $\Tan^*Q$ to $(\oplus^k \Tan^*Q) \times \Rk$
by the corresponding projection $(\oplus^k \Tan^*Q) \times   \Rk\to\Tan^*Q$.
Using coordinates $y^a$ on $Q$ and natural coordinates
$(y^a,p_a^\alpha)$ on each $\Tan^*Q$,
their local expressions are
\begin{equation}
\label{caneta}
\eta^\alpha = \d s^\alpha - p^\alpha_a \,\d y^a
\,,
\end{equation}
and the Reeb vector fields are
$$
\Reeb_\alpha = \frac{\partial}{\partial s^\alpha}
\,.
$$
\end{remark}

\subsubsection{$k$-Contact Hamiltonian systems}
\label{secHamsys}

First, let $\oplus^k\Tan\mathcal{M}:=\Tan\mathcal{M}\oplus \stackrel{(k)}{\dotsb}\oplus\Tan\mathcal{M}$ be
the so-called {\sl \textbf{$k$-tangent bundle}} or {\sl \textbf{bundle of $k^1$-velocities}} of $\mathcal{M}$.
It is endowed with natural projections to each direct summand and to the base manifold:
$$
\tau_\alpha\colon\oplus^k\Tan\mathcal{M}\longrightarrow\Tan\mathcal{M}
\,, \qquad \tau^1_\mathcal{M}\colon \oplus^k\Tan\mathcal{M}\longrightarrow\mathcal{M} \,.
$$
Then, a {\sl\textbf{$k$-vector field}} on $\mathcal{M}$ is a section
${\bf X}\colon\mathcal{M}\longrightarrow\oplus^k\Tan\mathcal{M}$ of the projection $\tau^1_\mathcal{M}$.
It is specified by giving $k$ vector fields $X_1,\dots,X_k\in\vf(\mathcal{M})$, 
obtained as $X_\alpha=\tau_\alpha\circ{\bf X}$.
Then, the $k$-vector field is specified as ${\bf X}=(X_1, \ldots, X_k)$.
Every $k$-vector field ${\bf X}=(X_1, \ldots, X_k)$ 
induces a decomposable, contravariant, skew-symmetric tensor field,
$X_1 \wedge \dotsb \wedge X_k$,
which is a section of the bundle $\Lambda^k\Tan\mathcal{M}\to\mathcal{M}$, and hence
this also induces a tangent distribution on~$\mathcal{M}$.
The sections of this bundle are generically called {\sl\textbf{$k$-multivector fields}} in $\mathcal{M}$
and, when they are of the form $X_1 \wedge \dotsb \wedge X_k$
(at least locally),
are called {\sl\textbf{(locally) decomposable $k$-multivector fields}}
(see the Appendix \ref{append}).

Let $\bm\psi\colon D\subset  \Rk\rightarrow\mathcal{M}$ be an immersion. 
If $t =(1,\ldots,k)$ denote the canonical coordinates in $\Rk$,
let $\bm\psi(x)=(\psi^I(x))$, $1\leq I\leq N$.
Then, the {\sl \textbf{first prolongation}} of $\bm\psi$ to $\oplus^k \Tan\mathcal{M}$
is the map $\bm\psi'\colon D\subset  \Rk\to\oplus^k \Tan\mathcal{M}$ defined by
$$
\bm\psi'(x) =
\left(\psi^I(x),\Tan\bm\psi \Big(\frac{\partial}{\partial x^1}\Big\vert_{x}\Big),
\ldots,
\Tan\bm\psi\Big(\frac{\partial}{\partial x^k}\Big\vert_{x}\Big)\right)
\equiv(\bm\psi(x);\bm\psi_\alpha'(x)) \,.
$$
We say that $\bm\psi$ is an {\sl\textbf{integral map}} of a $k$-vector field
${\bf X}=(X_{1},\dots, X_{k})$ if
\begin{equation}
\label{integsec}
\bm\psi'={\bf X}\circ\bm\psi
\,,
\end{equation}
or, equivalently, if
$\ds\Tan\bm\psi\circ\frac{\partial}{\partial x^\alpha}= X_\alpha\circ\bm\psi$,
for every~$\alpha$.
A $k$-vector field ${\bf X}$ is {\sl\textbf{integrable}} if
every point of~$Q$ is in the image of an integral map of~${\bf X}$.
In coordinates, if
$$
X_\alpha=X_\alpha^a\frac{\partial}{\partial y^a}+ X_{\alpha\beta}^a\frac{\partial}{\partial y^a_\beta}+X_\alpha^\beta\frac{\partial}{\partial s^\beta} \,,
$$
then $\bm\psi(x)=\big(y^a(x),y^a_\alpha(x),s^\alpha(x) \big)$ is an integral map of $\mathbf{X}$ 
if, and only if, it is a solution to the system of partial differential equations,
\begin{equation}
\label{inteqcoor}
\frac{\partial y^a}{\partial x^\alpha}=X_\alpha^a(\bm\psi) \,, \
\frac{\partial y^a_\beta}{\partial x^\alpha}=X_{\alpha\beta}^a(\bm\psi(x)) \,, \
\frac{\partial s^\beta}{\partial x^\alpha}=X_\alpha^\beta(\bm\psi(x)) \,,
\end{equation}

Now, we define:

\begin{definition}
A {\sl\textbf{$k$-contact Hamiltonian system}} is a family $(\mathcal{M},\eta^\alpha,\H)$,
where $(\mathcal{M},\eta^\alpha)$ is a $k$-contact manifold,
and $\H\in\Cinfty(\mathcal{M})$ is called a {\sl\textbf{Hamiltonian function}}.
\end{definition}

The field equations of a contact Hamiltonian system can be expressed in geometric form in two alternative ways:

\begin{definition}
\label{fieldeqs}
The \textbf{$k$-contact Hamilton--de Donder--Weyl equations} for a map
$\bm\psi\colon D\subset\Rk\to \mathcal{M}$ are
\begin{equation}
\label{hec}
\begin{cases}
\inn{\bm\psi'_\alpha}\d\eta^\alpha = \big(\d\H - (\Lie_{\Reeb_\alpha}\H)\eta^\alpha\big)\circ\bm\psi \,,\\
\inn{\bm\psi'_\alpha}\eta^\alpha = - \H\circ\bm\psi \,.
\end{cases}
\end{equation}
The {\sl\textbf{$k$-contact Hamilton--de Donder--Weyl equations}} 
for a $k$-vector field ${\bf X}=(X_1,\dots,X_k)$ in $\mathcal{M}$ are
\begin{equation} 
\label{fieldcontact}
\begin{cases}
    \inn{{X_\alpha}}\d\eta^\alpha=\d\H-(\Lie_{\Reeb_\alpha}\H)\eta^\alpha \,,\\
    \inn{X_\alpha}\eta^\alpha=-\H \,.
    \end{cases}
\end{equation}
Their solutions are called {\sl\textbf{Hamiltonian $k$-vector fields}}.
\end{definition}

Bearing in mind the definition of the integral maps of an integrable $k$-vector field (see equations \eqref{integsec} and \eqref{inteqcoor}),
it is immediate to prove the following:

\begin{proposition}
If ${\bf X}$ is an integrable $k$-vector field in $\mathcal{M}$, then
every integral map $\bm\psi\colon D\subset\Rk\to\mathcal{M}$ of ${\bf X}$
satisfies the $k$-contact equation \eqref{hec}
if, and only if,
${\bf X}$ is a solution to \eqref{fieldcontact}.
\end{proposition}

If $(\mathcal{M},\eta^\alpha,\H)$ is a contact Hamiltonian system; 
using canonical coordinates for the contact structure $(\mathcal{M},\eta^\alpha)$,
if $\bm\psi(x)=(y^a(x),p^\alpha_a(x),s^\alpha(x))$ is a solution to the equations \eqref{hec},
then $\displaystyle\bm\psi'_\alpha=
\Big(y^a,p^\alpha_a,s^\alpha,\parder{y^a}{\beta},\parder{p^\alpha_a}{\beta},\parder{s^\alpha}{\beta}\Big)$,
and equations \eqref{hec} read,
$$
\begin{cases}
\displaystyle \frac{\partial y^a}{\partial x^\alpha} = \frac{\partial\H}{\partial p^\alpha_a}\circ\psi \,,\\[15pt]
\displaystyle \frac{\partial p^\alpha_a}{\partial x^\alpha} = 
-\left(\frac{\partial\H}{\partial y^a}+ p_a^\alpha\frac{\partial\H}{\partial s^\alpha}\right)\circ\psi \,,\\[15pt]
\displaystyle \frac{\partial s^\alpha}{\partial x^\alpha} = 
\left(p_a^\alpha\frac{\partial\H}{\partial p^\alpha_a}-\H\right)\circ\psi\,.
\end{cases}
$$
Furthermore, if ${\bf X}=(X_\alpha)$, with
$\displaystyle
X_\alpha= X_\alpha^\beta\parder{}{s^\beta}+ X_\alpha^a\frac{\partial}{\partial y^a}+
X_{\alpha a}^\beta\frac{\partial}{\partial p_a^\beta}$,
is a $k$-vector field solution to \eqref{fieldcontact},
then these equations lead to
\begin{equation}
\label{coor0}
\begin{cases}
\displaystyle X_\alpha^a= \frac{\partial\H}{\partial p^\alpha_a} \,,\\[15pt]
\displaystyle X_{\alpha a}^\alpha = 
-\left(\frac{\partial\H}{\partial y^a}+ p_a^\alpha\frac{\partial\H}{\partial s^\alpha}\right) \,,\\[15pt]
\displaystyle X_\alpha^\alpha = 
p_a^\alpha\frac{\partial\H}{\partial p^\alpha_a}-\H\,.
\end{cases}
\end{equation}
And, from these last equations \eqref{coor0}
we obtain that (see also \cite{GGMRR_20}):

\begin{proposition}
If $(\mathcal{M},\eta^\alpha,\H)$ is a contact Hamiltonian system, then there exist
solutions to the equations \eqref{fieldcontact}, although
they are neither unique, nor necessarily integrable.    
\end{proposition}

\begin{remark}
An equivalent way to write equations \eqref{fieldcontact} is:
$$
\begin{cases}
\Lie_{X_\alpha}\eta^\alpha = - (\Lie_{\Reeb_\alpha}\H)\eta^\alpha \,, \\
\inn{X_\alpha}\eta^\alpha = - \H \,.
\end{cases}
$$
Another alternative and
partially equivalent, expression for the Hamilton--De Donder--Weyl equations,
without using the Reeb vector fields $\Reeb_\alpha$, is as follows (see \cite{GGMRR_20}):
Consider the 2-forms $\Omega^\alpha= -\H\,\d\eta^\alpha+\d\H\wedge\eta^\alpha$.
On the open set ${\cal O}=\{ p\in\mathcal{M}\mid \H(p)\not=0\}$,
if a $k$-vector field ${\bf X}=(X_\alpha)$ satisfies,
\begin{equation}
\label{hamilton-eqs-no-reeb}
    \begin{cases}
        \inn{X_\alpha}\Omega^\alpha = 0\,,\\
        \inn{X_\alpha}\eta^\alpha = -\H\,,
    \end{cases}
\end{equation}
then ${\bf X}$ is a solution of the Hamilton--De Donder--Weyl equations \eqref{fieldcontact}).
Then, the integral maps $\psi$ of such a $k$-vector fields are solutions to
\begin{equation}
 \begin{cases}
\inn{\psi'_\alpha}\Omega^\alpha = 0 
\,,\\
\inn{\psi'_\alpha}\eta^\alpha = -\H\circ\psi 
\,.
\end{cases}
\label{hec2}
\end{equation}
\end{remark}

\subsection{\texorpdfstring{$k$}{}-Contact Lagrangian formalism}

Now we describe the Lagrangian formalism of action-dependent field theories,
using $k$-contact structures.

\subsubsection{Geometry of the phase bundle}

Let $Q$ be a $n$-dimensional differentiable manifold, and
consider its $k$-tangent bundle $\oplus^k \Tan Q = \Tan Q \oplus \stackrel{k}{\dots} \oplus \Tan Q$.
The natural coordinates in $\oplus^k\Tan Q$ are denoted $(y^a ,y_\alpha^a)$,
with $1\leq i\leq n$ and $1\leq\alpha\leq k$.

The $k$-tangent bundle has some canonical structures which are induced on it by the canonical structures of the tangent bundle $\Tan Q$.
In particular, first we have the so-called {\sl\textbf{ canonical $k$-tangent structure}}, which is the set
$(J^1,\ldots,J^k)$ of tensor fields of type $(1,1)$ in $\oplus^k \Tan Q$
whose local expression in natural coordinates are
$\displaystyle J^\alpha=\frac{\partial}{\partial y^a_\alpha} \otimes \d y^a$.
Second, we have the {\sl\textbf{Liouville vector field}}
$\Delta\in\vf(\oplus^k \Tan Q)$, which is the infinitesimal generator of dilations in the fibers of the bundle $\oplus^k \Tan Q\rightarrow\Tan Q$; that is, whose flow
$\psi\colon\R\times \oplus^k \Tan Q\longrightarrow \oplus^k \Tan Q$ is given by
$\psi(t;v_{1q},\ldots,v_{kq}) = (e^t v_{1q},\ldots,e^t v_{kq})$.
In coordinates, 
$\displaystyle\Delta = y^a_\alpha\parder{}{y_\alpha^a}$.

A map $\varphi\colon D\subset\Rk\to \oplus^k\Tan Q$
is said to be {\sl\textbf{holonomic}} if it is the first prolongation of a map $\phi\colon D\subset\Rk\rightarrow Q$.
In coordinates, if $\phi(x)=(\phi^a(x))$, then
$\ds\phi'(x)=\Big(\phi^a(x),\parder{\phi^a}{x^\alpha}(x)\Big)$.
(See \cite{LSV_15} for more details on all the above topics).

Action-dependent Lagrangian field theories are developed in a bundle that is built by enlarging the above $k$-tangent bundle to include the dissipation variables.
Thus, consider the bundle $P\equiv\oplus^k \Tan Q\times\Rk$;
whose natural coordinates are $(y^a,y^a_\alpha,s^\alpha)$.
We have the canonical projections
\begin{align}
& \bar\tau_1\colon P\equiv\oplus^k \Tan Q\times\Rk\to\oplus^k \Tan Q\,, &&
 \bar\tau_2\colon P\equiv\oplus^k \Tan Q\times\Rk\to\Rk
 \,, \\
& s^\alpha\colon P\equiv\oplus^k\Tan Q\times\Rk\to\R\,, &&
\bar\tau_{Q\times\Rk}\colon P\equiv\oplus^k\Tan Q\times\Rk\to Q\times \Rk \,.
\end{align}
The manifolds $P$ and $Q\times\Rk$ are called the {\sl\textbf{ $k^1$-velocity phase space}} and  the {\sl\textbf{configuration space}} of the $k$-contact field theory, respectively.

\begin{definition}
\label{de652}
Let $\psi\colon\Rk\rightarrow Q\times\Rk$ and $\phi\colon\Rk\to Q$ be immersions,
such that $\psi(x)=(\phi^a(x),s^\alpha(x))$.
The {\sl\textbf{first prolongation}} 
 of $\psi$ to $ P=\oplus^k \Tan Q\times\Rk$ is the map
$\bm\psi\colon\Rk\to\oplus^k\Tan Q\times\Rk$
given by
$\bm\psi=(\phi',s^\alpha)$;
where $\phi'\colon\Rk\to\oplus^k\Tan Q$
is the first prolongation of $\phi$ to $\oplus^k\Tan Q$.
The map $\bm\psi$ is said to be {\sl\textbf{holonomic}} in $ P$.
\end{definition}

In coordinates, the expression of a holonomic map in $ P$ is
\begin{equation}
\label{holsecP}
\bm\psi(x)=\Big(\phi^a(x),\parder{\phi^a}{\alpha}(x),s^\alpha(x)\Big) \,.
\end{equation}

\begin{definition}
\label{lem0}
A $k$-vector field $\bm\Gamma$ in $ P$ 
is said to be {\sl\textbf{holonomic}} or a {\sl\textbf{second order partial differential equation}}  
(\textsc{sopde})
if it is integrable and its integral maps are holonomic in $ P$.
\end{definition}

If $\bm\psi$ is locally given by
\eqref{holsecP} and it is an integral map of
a {\sc sopde} $\bm\Gamma$,
whose vector fields components have local expressions as
\begin{equation}
\label{locsopde}
\Gamma_\alpha= 
\Gamma^a_\alpha\frac{\displaystyle\partial} {\displaystyle
\partial y^a}+\Gamma_{\alpha\beta}^a\frac{\displaystyle\partial}{\displaystyle \partial y^a_\beta}+\Gamma^\beta_\alpha\,\frac{\partial}{\partial s^\beta}\,.
\end{equation}
Then, from \eqref{integsec} we have that the components of $\bm\psi(x)$ are the solution to the system of second order partial
differential equations,
\begin{equation}
\label{2ordeq}
\frac{\partial\phi^a} {\partial \alpha}=\Gamma^a_\alpha(\bm\psi(x)) 
\quad , \quad
\frac{\partial^2 \phi^a}{\partial \alpha\partial \beta}=
\Gamma_{\alpha\beta}^a(\bm\psi(x)) \,.
\end{equation}
Therefore, the local expressions of the vector fields components of a {\sc sopde} are
\begin{equation}
\label{localsode1}
\Gamma_\alpha= 
y^a_\alpha\frac{\displaystyle\partial} {\displaystyle
\partial y^a}+\Gamma_{\alpha\beta}^a\frac{\displaystyle\partial}{\displaystyle \partial y^a_\beta}+\Gamma^\beta_\alpha\,\frac{\partial}{\partial s^\beta}\,,
\end{equation}
and observe that, from the second equation of \eqref{2ordeq}, we obtain that
$\Gamma_{\alpha\beta}^a=\Gamma_{\beta\alpha}^a$.

\begin{remark}
Since $\oplus^k\Tan Q\times\Rk\to\oplus^k\Tan Q$ is a trivial bundle, 
the canonical structures in $\oplus^k\Tan Q$; i.e., the canonical $k$-tangent structure and the Liouville vector field,
can be extended to $P\equiv\oplus^k\Tan Q\times\Rk$ in a natural way.
They are denoted with the same notation, $(J^\alpha)$ and $\Delta$,
and have the same coordinate expressions as above.
Then, using these structures, we have the following alternative geometric characterizations for {\sc sopde} $k$-vector fields in $ P$.
\end{remark}

Then, a simple calculation in coordinates leads to the following result:

\begin{proposition}
An integrable $k$-vector field $\bm\Gamma=(\Gamma_\alpha)$ in $P$ is a {\sc sopde}
if, and only if,
\ $J^\alpha(\Gamma_\alpha)=\Delta$.
\end{proposition}

\begin{remark}
The $k$-vector fields that satisfy the above condition, $J^\alpha(\Gamma_\alpha)=\Delta$, whose local expression is \eqref{localsode1},
are called {\sl\textbf{semi-holonomic $k$-vector fields}}.
\end{remark}

\subsubsection{$k$-Contact Lagrangian systems}

Now, we can state the Lagrangian formalism for action-dependent  field theories (see \cite{GGMRR_21}).

\begin{definition}
\label{lagrangean}
A {\sl \textbf{Lagrangian function}} is a function $\L\in\Cinfty( P)$. The {\sl \textbf{Lagrangian energy}} associated to $\L$ is the function $$E_\L:=\Delta(\L)-\L\in\Cinfty( P) \ . $$
The {\sl\textbf{Cartan forms}} associated the Lagrangian function $\L$ are
    \begin{equation*}
\theta_\L^\alpha={}^t(J^\alpha)\circ\d\L \in\df^1( P)
\quad , \quad
\omega_\L^\alpha=-\d\theta_\L^\alpha\in\df^2( P)\,.
    \end{equation*}
Finally, we can define the forms
$$
\eta_\L^\alpha=\d s^\alpha-\theta_\L^\alpha\in\Omega^1( P) \quad ,\quad \d\eta_\L^\alpha=\omega_\L^\alpha\in\Omega^2( P) \,.
$$
\end{definition}

In natural coordinates $(y^a,y^a_\alpha,s^\alpha)$ of $P$,
the local expressions of these elements are
\begin{equation}
\label{EetaL}
E_L=y^a_\alpha\frac{\partial\L}{\partial y^a_\alpha}-\L
\quad ,\quad
\eta_\L^\alpha=\d s^\alpha-\frac{\partial\L}{\partial y^a_\alpha}\d y^a\,.
\end{equation}

\begin{definition}
The {\sl\textbf{Legendre map}} associated with a Lagrangian $\L\in\Cinfty( P)$
is the fiber derivative of~$\L$, considered as a function on the vector bundle $\bar\tau_{Q\times\Rk}\colon P \to Q\times\Rk$;
that is, the map
${\cal F}\L\colon P\equiv\oplus^k\Tan Q\times\Rk\to  P^*\equiv\oplus^k\Tan^*Q\times\Rk$,
given by
$$
{\cal F}\L({v_1}_q,\ldots,{v_k}_q;s^\alpha)=\left({\cal F}\L_s({v_1}_q,\ldots,{v_k}_q),s^\alpha\right)
\ ; \quad
({v_1}_q,\ldots,{v_k}_q)\in\oplus^k\Tan Q\ ;
$$
where $\L_s$ denotes the restriction of the Lagrangian function to the fibers of the projection
$\bar\tau_2\colon\oplus^k\Tan Q\times\Rk\to\Rk$ (i.e.; with $s^\alpha$ ``freezed''),
and ${\cal F}\L_s\colon\oplus^k\Tan Q\to\oplus^k\Tan^*Q$ is the corresponding fiber derivative.
\end{definition}

The local expression of this map is\ 
$\displaystyle{\cal F}\L(y^a,y^a_\alpha,s^\alpha)=\Big(y^a,\frac{\partial\L}{\partial y^a_\alpha},s^\alpha\Big)$.

\begin{proposition}
\label{Prop-regLag}
For a Lagrangian function $\L$ the following conditions are equivalent:
\begin{enumerate}[{\rm(1)}]
\item
$( P,\eta_\L^\alpha)$ is a $k$-contact manifold.
\item
The Legendre map
${\cal FL}$ is a local diffeomorphism.
\item
The Hessian matrix $\ds \left(\frac{\partial^2{\cal L}}{\partial y^a_\alpha \partial y^b_\beta} \right)$
is nondegenerate everywhere.
\end{enumerate}
\end{proposition}

\begin{definition}
A Lagrangian function $\L$ is said to be {\sl\textbf{regular}} if the equivalent
conditions in Proposition \ref{Prop-regLag} hold.
Otherwise, $\L$ is a {\sl\textbf{singular}} Lagrangian.
In particular, 
$\L$ is said to be {\sl\textbf{hyperregular}} 
if ${\cal FL}$ is a global diffeomorphism.
\end{definition}

\begin{definition}
The pair  $( P,\L)$ is called a {\sl\textbf{$k$-contact Lagrangian system}}.
It defines a
$k$-contact Hamiltonian system $(P,\eta^\alpha_\L,E_\L)$.
\end{definition}

For a $k$-contact Lagrangian system $(P,\L)$; i.e., when $\L$ is regular,
the Reeb vector fields $(\Reeb_\L)_\alpha\in\X(P)$ 
for this system are the unique solution to the equations \eqref{reebcontact}, which now read as
$$
\inn{(\Reeb_\L)_\alpha}\d\eta^\beta_\L=0\quad ,\quad
\inn{(\Reeb_\L)_\alpha}\eta^\beta_\L=\delta_\alpha^\beta \,.
$$
In this case, there exists the inverse 
$W^{ij}_{\alpha\beta}$ of the Hessian matrix,
namely $\displaystyle W^{ab}_{\alpha\beta}\frac{\partial^2\L}{\partial y^b_\beta \partial y^c_\gamma}=\delta^a_c\delta^\gamma_\alpha$,
and then we obtain that
$$
(\Reeb_\L)_\alpha=\frac{\partial}{\partial s^\alpha}-W^{ba}_{\gamma\beta}\frac{\partial^2\L}{\partial s^\alpha\partial y^b_\gamma}\,\frac{\partial}{\partial y^a_\beta} \,.
$$

\subsubsection{The $k$-contact Lagrangian equations}

The field equations for the Lagrangian formalism of action-dependent field theories can be expressed in the two alternative ways stated in Definition \ref{fieldeqs}:

\begin{definition}
Let $( P,\L)$ be a $k$-contact Lagrangian system.
\begin{enumerate}[{\rm(1)}]
\item 
The {\sl\textbf{$k$-contact Euler--Lagrange equations}} for holonomic maps 
$\bm\psi\colon\Rk\to P$ are:
\begin{equation}
\label{ELkcontact}
 \begin{cases}
\inn{\bm\psi'_\alpha}\d\eta_\L^\alpha=
 \Big(\d E_\L-(\Lie_{(\Reeb_\L)_\alpha}E_\L\,\eta_\L^\alpha\Big)\circ\bm\psi \,,\\
\inn{\bm\psi'_\alpha}\eta_\L^\alpha=
-E_\L\circ\mbox{\boldmath$\psi$} \,.
\end{cases}
\end{equation}
\item
The {\sl\textbf{$k$-contact Lagrangian equations}} for holonomic $k$-vector fields
${\bf X}_\L=((X_\L)_\alpha)$ in $ P$ are
\begin{equation} \begin{cases}
\inn{(X_\L)_\alpha}\d\eta_\L^\alpha=\d E_\L- (\Lie_{(\Reeb_\L)_\alpha}E_\L)\eta_\L^\alpha \,,\\
\inn{(X_\L)_\alpha}\eta_\L^\alpha=-E_\L \,.
    \end{cases}
    \label{fieldLcontact}
\end{equation}
A $k$-vector field which is a solution to these equations is called a
{\sl\textbf{Lagrangian $k$-vector field}}.
These holonomic $k$-vector fields are called {\sl\textbf{Euler--Lagrange $k$-vector fields}}.
\end{enumerate}
\end{definition}

\begin{proposition}
Let $( P,\L)$ be a $k$-contact Lagrangian system.
If $\bfX_\L$ is a holonomic $k$-vector field (that is, a {\sc sopde}) solution to the Lagrangian equations \eqref{fieldLcontact}, then its integral sections are the solutions to the multicontact Euler--Lagrange field equations for holonomic sections \eqref{ELkcontact} associated with $\mathcal{L}$.

In addition, if the Lagrangian system is regular (that is, $k$-contact) then:
\begin{enumerate}[\rm (1)]
\item
The $k$-contact Lagrangian field equations for $k$-vector fields \eqref{fieldLcontact} admit solutions on $P$.
(The solutions are not unique if $m>1$).
\item
Every $k$-vector field $\bfX_\L$ that is
solution to equations \eqref{fieldLcontact} is semi-holonomic.
\end{enumerate}

\end{proposition}
\begin{proof}
In a natural chart of coordinates of $P$, equations \eqref{ELkcontact} read
\begin{equation}
\label{ELeqs1}
\frac{\partial}{\partial \alpha}
\left(\frac{\displaystyle\partial \L}{\partial
y^a_\alpha}\circ{\bm\psi}\right)=
\left(\frac{\partial \L}{\partial y^a}+
\displaystyle\frac{\partial\L}{\partial s^\alpha}\displaystyle\frac{\partial\L}{\partial y^a_\alpha}\right)\circ{\bm\psi}
 \quad  , \quad
\parder{s^\alpha}{\alpha}=\L\circ{\bm\psi} \ ;
\end{equation}
meanwhile, for a $k$-vector field ${\bf X}_\L=((X_\L)_\alpha)$, with
$$\displaystyle (X_\L)_\alpha= 
(X_\L)_\alpha^a\frac{\displaystyle\partial}{\displaystyle
\partial y^a}+(X_\L)_{\alpha\beta}^a\frac{\displaystyle\partial}{\displaystyle\partial y^a_\beta}+(X_\L)_\alpha^\beta\,\frac{\partial}{\partial s^\beta}\ ,$$
the Lagrangian equations \eqref{fieldLcontact} are
\begin{align}
0 &=
\displaystyle \L + 
\frac{\partial\L}{\partial y^a_\alpha}\Big((X_\L)_\alpha^a-y^a_\alpha\Big)-(X_\L)_\alpha^\alpha\,,
\label{A-E-L-eqs40}
\\
0 &=
\displaystyle \Big((X_\L)_\alpha^a-y^a_\alpha\Big)
\frac{\partial^2\L}{\partial y^a_\alpha\partial s^\beta} \,,
\label{A-E-L-eqs20}
\\
0 &=
\displaystyle \Big((X_\L)_\alpha^a-y^a_\alpha\Big)
\frac{\partial^2\L}{\partial y^b_\beta\partial y^a_\alpha}
\label{A-E-L-eqs10} \,,
\\
0 &=
\displaystyle
\Big((X_\L)_\alpha^a-y^a_\alpha\Big)
\frac{\partial^2\L}{\partial y^b\partial y^a_\alpha}
+\frac{\partial\L}{\partial y^b}
-\frac{\partial^2\L}{\partial s^\beta\partial y^b_\alpha}(X_\L)_\alpha^\beta
\nonumber
\\ & \quad
-\frac{\partial^2\L}{\partial y^a \partial y^b_\alpha}(X_\L)_\alpha^a
-\frac{\partial^2\L}{\partial y^a_\beta\partial y^b_\alpha}(X_\L)_{\alpha\beta}^a
+\frac{\partial\L}{\partial s^\alpha}
\frac{\partial\L}{\partial y^b_\alpha}\,.
\label{A-E-L-eqs30}
\end{align}
If $\bfX_\L$ is a {\sc sopde}, then $y^a_\alpha=(X_\L)_\alpha^a$;
therefore, equations \eqref{A-E-L-eqs20} hold identically, and
\eqref{A-E-L-eqs40} and \eqref{A-E-L-eqs30} give 
\begin{align}
\label{eq1}
(X_\L)_\alpha^\alpha&= \L\,,
\\
\displaystyle
\frac{\partial\L}{\partial y^b}
-\frac{\partial^2\L}{\partial s^\beta\partial y^b_\alpha}(X_\L)_\alpha^\beta
-\frac{\partial^2\L}{\partial y^a \partial y^b_\alpha}y_\alpha^a
-\frac{\partial^2\L}{\partial y^a_\beta\partial y^b_\alpha}(X_\L)_{\alpha\beta}^a
&=-\frac{\partial\L}{\partial s^\alpha}
\frac{\partial\L}{\partial y^b_\alpha}\,.
\label{eq2}
\end{align}
Finally, for the holonomic integrable maps of ${\bf X}_\L$,
these last equations lead to the Euler--Lagrange equations
\eqref{ELeqs1} for its integral maps.
In addition, the first equation \eqref{eq1} relates the variation 
of the ``dissipation coordinates'' $s^\alpha$ to the Lagrangian function.

If $\L$ is a regular Lagrangian, equations \eqref{A-E-L-eqs10}
lead to $y^a_\alpha=(X_\L)_\alpha^a$, which is the {\sc sopde} condition for ${\bf X}_\L$.
Furthermore, equations \eqref{eq2} have always solution for coefficients $(X_\L)_{\alpha\beta}^b$
(not necessarily unique, unless $k=1$), since the Hessian matrix 
$\ds \left(\frac{\partial^2{\cal L}}{\partial y^a_\alpha \partial y^b_\beta} \right)$
is regular everywhere.
\end{proof}

\subsection{\texorpdfstring{$k$}{}-Contact Hamiltonian formalism}

Next, we use the developments stated in Section \ref{secHamsys} to develop the Hamiltonian formalism for action-dependent field theories.

In the $k$-contact ambient,
action-dependent Hamiltonian field theories are developed in a manifold which is built enlarging the $k$-cotangent bundle of a manifold $Q$,
as in the Lagrangian setting.
Thus, we consider the bundle $ P^*\equiv\oplus^k \Tan^*Q\times\Rk$;
whose natural coordinates are $(y^a,p_a^\alpha,s^\alpha)$.
We have the canonical projections
\begin{align}
& \widetilde\tau_1\colon \oplus^k\Tan^*Q\times\Rk\to\oplus^k\Tan^*Q
\,, \ &&
\widetilde\tau_2\colon\oplus^k\Tan^*Q\times\Rk\to\Rk
 \,, \\
& s^\alpha\colon\oplus^k\Tan^*Q\times\Rk\to\R 
\,, \ &&
\widetilde\tau_{Q\times\Rk}\colon \oplus^k\Tan^*Q\times\Rk\to Q\times \Rk \,.
\end{align}
Regular or $k$-contact Hamiltonian field theories take place in the canonical $k$-contact manifold
$(\oplus^k\Tan^*Q\times\Rk,\theta^\alpha)$,
giving a {\sl Hamiltonian function} $\mathcal{H}\in\Cinfty(\oplus^k\Tan^*Q\times\Rk)$.

\begin{remark}
[{\sl The canonical $k$-contact Hamiltonian system associated with a $k$-contact Lagrangian system\/}]
In particular, if $(P=\oplus^k\Tan Q\times\Rk,\L)$ is a $k$-contact Lagrangian system,
we have that ${\cal FL}$ is a local or global diffeomorphism between $P$ and $ P^*$,
depending on $\L$ to be a regular or hyper-regular  Lagrangian.
Then, bearing in mind the coordinate expressions \eqref{caneta}
and \eqref{EetaL} of $\eta^\alpha$ $\eta^\alpha_\L$,
and of the Legendre map, we have that
$$
\theta_\L^\alpha={\cal FL}^{\;*}\theta^\alpha
\,,\qquad
\omega_\L^\alpha={\cal FL}^{\;*}\omega^\alpha
\, ,
$$
where $\omega^\alpha=-\d\theta^\alpha$.
Furthermore, there exists (maybe locally) a function 
$\H\in\Cinfty( P^*)$  such that 
$$\H=E_\L\circ{\cal F}\L^{-1}\,.$$
Then, $( P^*,\eta^\alpha,\H)$ is the {\sl\textbf{canonical $k$-contact Hamiltonian system}}
associated with the $k$-contact Lagrangian system
$(P,\L)$ and,
for it, ${\cal FL}_*({\Reeb}_\L)_\alpha={\Reeb}_\alpha$.
Therefore, if $\mathbf{X}_\L$ is an Euler--Lagrange $k$-vector field
associated with $\L$ in $ P$, then 
${\cal FL}_*\mathbf{X}_\L={\bf X}_\H$ is a
contact Hamiltonian $k$-vector field associated with $\H$ in
$ P^*$, and conversely.    
\end{remark}

\section{\texorpdfstring{$k$}{}-Cocontact field theories}
\label{kcocontactf}

This section reviews the basics of $k$-cocontact manifolds and their applications in modeling non-autonomous action-dependent field theories (see \cite{Riv_23} for details).

\subsection{\texorpdfstring{$k$}{}-Cocontact structures and \texorpdfstring{$k$}{}-cocontact Hamiltonian systems}

First, we summarize the fundamental concepts and properties about 
$k$-cocontact manifolds and $k$-cocontact Hamiltonian systems.

\subsubsection{\texorpdfstring{$k$}{}-Cocontact structures}\label{sec:kcocostruc}

Given a $N$-dimensional manifold $\mathcal{M}$,
let $\tau^1,\dotsc,\tau^k\in\Omega^1(\mathcal{M})$ be a family of closed one-forms on $\mathcal{M}$ and let $\eta^1,\dotsc,\eta^k\in\Omega^1(\mathcal{M})$ be a family of one-forms on $\mathcal{M}$. We will use the following notations:
\begin{itemize}
    \item $\C^\rmC = \langle\eta^1,\dotsc,\eta^k\rangle\subset\cT \mathcal{M}$,
    \item $\D^\rmC = \left(\C^\rmC\right)^\circ = \ker\widehat{\eta^1}\cap\dotsb\cap\ker\widehat{\eta^k}\subset\T \mathcal{M}$,
    \item $\D^\rmR = \ker\widehat{\d\eta^1}\cap\dotsb\cap\ker\widehat{\d\eta^k}\subset\T \mathcal{M}$,
    \item $\C^\rmR = \left(\D^\rmR\right)^\circ\subset\cT \mathcal{M}$,
    \item $\C^\rmS = \langle\tau^1,\dotsc,\tau^k\rangle\subset\cT \mathcal{M}$,
    \item $\D^\rmS = \left(\C^\rmS\right)^\circ = \ker\widehat{\tau^1}\cap\dotsb\cap\ker\widehat{\tau^k}\subset\T \mathcal{M}$.
\end{itemize}
With these notations, we can define the notion of $k$-cocontact structure:
\begin{definition}\label{dfn:k-cocontact-manifold}
    A {\sl\textbf{$k$-cocontact structure}} on a manifold $\mathcal{M}$ is a family of $k$ closed differential one-forms $\tau^1,\dotsc,\tau^k\in\Omega^1(\mathcal{M})$ and a family of $k$ differential one-forms $\eta^1,\dotsc,\eta^k\in\Omega^1(\mathcal{M})$ such that, with the preceding notations,
    \begin{enumerate}[{\rm(1)}]
        \item $\D^\rmC\subset\T \mathcal{M}$ is a regular distribution of corank $k$,
        \item $\D^\rmS\subset\T \mathcal{M}$ is a regular distribution of corank $k$,
        \item $\D^\rmR\subset\T \mathcal{M}$ is a regular distribution of rank $2k$,
        \item $\D^\rmC \cap\D^\rmS$ is a regular distribution of corank $2k$, $\D^\rmC \cap\D^\rmR$ is a regular distribution of rank $k$, and $\D^\rmS \cap\D^\rmR$ is a regular distribution of rank $k$,
        \item $\D^\rmC\cap\D^\rmR\cap\D^\rmS = \{0\}$.
    \end{enumerate}
    We call $\C^\rmC$ the {\sl\textbf{contact codistribution}}, $\D^\rmC$ the {\sl\textbf{contact distribution}}, $\D^\rmR$ the {\sl\textbf{Reeb distribution}}, $\C^\rmR$ the {\sl\textbf{Reeb codistribution}}, $\C^\rmS$ the {\sl\textbf{space-time codistribution}} and $\D^\rmS$ the {\sl\textbf{space-time distribution}}.
    
    A manifold $\mathcal{M}$ endowed with a $k$-cocontact structure $\tau^1,\dotsc,\tau^k,\eta^1,\dotsc,\eta^k\in\Omega^1(\mathcal{M})$ is a {\sl\textbf{$k$-cocontact manifold}}.
\end{definition}

Notice that the condition $\D^\rmC\cap\D^\rmR\cap\D^\rmS = \{0\}$ implies that
$$ \cT \mathcal{M} = \C^\rmC \oplus \C^\rmR \oplus \C^\rmS \,. $$

\begin{remark}
    In the particular case $k=1$, a 1-cocontact structure is given by two one-forms $\tau,\eta$, with $\d\tau=0$. The conditions in Definition \ref{dfn:k-cocontact-manifold} mean the following: (1) $\eta\neq 0$ everywhere, (2) $\tau\neq 0$ everywhere, (4) $\tau\wedge\eta\neq 0$, (5) $\ker\widehat{\tau}\cap\ker\widehat{\eta}\cap\ker\widehat{\d\eta} = \{0\}$, which implies that $\ker\widehat{\d\eta}$ has rank 0, 1 or 2, and (3) implies that $\ker\widehat{\d\eta}$ has rank 2. Thus, a 1-cocontact structure coincides with the cocontact structure introduced in \cite{LGGMR_23} to describe time-dependent contact mechanics. 
\end{remark}

\begin{lemma}
    The Reeb distribution $\D^\rmR$ and the space-time distribution $\D^\rmS$ are involutive, and therefore integrable.
\end{lemma}

Thus, the distribution $\D^\rmR\cap\D^\rmS$ is also involutive, and therefore integrable. Moreover, the distribution $\D^\rmR\cap\D^\rmC$ is also involutive and integrable. The following theorem characterizes a family of vector fields spanning the Reeb distribution $\D^\rmR$.

\begin{theorem}\label{thm:Reeb-vector-fields}
    Let $(\mathcal{M},\tau^\alpha,\eta^\alpha)$ be a $k$-cocontact manifold. Then, there exist a unique family $R^x_1,\dotsc,R^x_k,R^s_1,\dotsc,R^s_k\in\X(\mathcal{M})$ such that
    \begin{gather}
        \inn{R^x_\alpha}\d\eta^\beta = 0\,,\qquad \inn{R^x_\alpha}\eta^\beta = 0\,,\qquad \inn{R^x_\alpha}\tau^\beta = \delta_\alpha^\beta\,,\\
        \inn{R^s_\alpha}\d\eta^\beta = 0\,,\qquad \inn{R^s_\alpha}\eta^\beta = \delta_\alpha^\beta\,,\qquad \inn{R^s_\alpha}\tau^\beta = 0\,.
    \end{gather}
    The vector fields $R^x_\alpha$ are called {\sl\textbf{space-time Reeb vector fields}}. The vector fields $R^s_\alpha$ are called {\sl\textbf{contact Reeb vector fields}}.

    Moreover, the Reeb vector fields commute and span the Reeb distribution introduced in Definition \ref{dfn:k-cocontact-manifold}:
    $$ \D^\rmR = \langle R^x_1,\dotsc,R^x_k,R^s_1,\dotsc,R^s_k\rangle\,, $$
    motivating its name.
\end{theorem}

The following proposition proves the existence of a special set of coordinates, the so-called adapted coordinates.

\begin{proposition}\label{prop:adapted-coordinates}
    Consider a $k$-cocontact manifold $(\mathcal{M},\tau^\alpha,\eta^\alpha)$. Then, around every point in $\mathcal{M}$, there exist local coordinates $(\alpha,z^I,s^\alpha)$ such that
    $$ R_\alpha^x = \parder{}{x^\alpha}\,,\qquad \tau^\alpha = \d x^\alpha\,,\qquad R_\alpha^s = \parder{}{s^\alpha}\,,\qquad \eta^\alpha = \d s^\alpha - f_I^\alpha(z^J)\d z^I\,, $$
    where the functions $f_I^\alpha$ only depend on the coordinates $z^I$. These coordinates are called {\sl\textbf{adapted coordinates}}.
\end{proposition}

\begin{example}[Canonical $k$-cocontact structure]\label{ex:canonical-k-cocontact-structure}
    Let $Q$ be a smooth $n$-dimensional manifold with coordinates $(y^a)$ and let $k\geq 1$. Consider the product manifold $\mathbf{P}^*=\Rk\times\bigoplus^k\cT Q\times\Rk$ endowed with natural coordinates $(x^\alpha; y^a, p_a^\alpha; s^\alpha)$. We have the canonical projections
    \begin{center}
        \begin{tikzcd}
            \R & \Rk\times\bigoplus^k\cT Q\times\Rk \arrow[l, swap, "\pi_1^\alpha", ] \arrow[r, "\pi_3^\alpha"] \arrow[dd, "\pi_2"] \arrow[ddd, bend right=50, swap, "\pi_\circ"] \arrow[ddr, "\pi_2^\alpha"] & \R \\
            &  && \\
            & \bigoplus^k\cT Q \arrow[r, swap, "\pi^\alpha"] & \cT Q \\
            & \Rk\times Q\times\Rk 
        \end{tikzcd}
    \end{center}
    Let $\theta$ be the Liouville one-form on $\cT Q$ with local expression in natural coordinates $\theta = p_a \d y^a$. Then, the family $(\tau^\alpha,\eta^\alpha)$ where $\tau^\alpha = \pi_1^{\alpha\,*}\,\d x$ with $x$ the canonical coordinate of $\R$ and $\eta^\alpha = \d s^\alpha - \pi_2^{\alpha\,*}\theta$, is a $k$-cocontact structure on $\mathcal{M}$. In natural coordinates,
    $$ \tau^\alpha = \d x^\alpha\,,\qquad \eta^\alpha = \d s^\alpha - p_a^\alpha\d y^a\,. $$
    Thus, the Reeb vector fields are $R_\alpha^x = \tparder{}{x^\alpha}$ and $R_\alpha^s = \tparder{}{s^\alpha}$.
\end{example}

The following theorem is an upgrade of Proposition 
\ref{prop:adapted-coordinates} and states the existence of Darboux-like coordinates in a $k$-cocontact manifold provided the existence of a certain subdistribution $\mathcal{V}\subset\D^\rmC$.

\begin{theorem}[Darboux theorem for $k$-cocontact manifolds]
    Let $(\mathcal{M},\tau^\alpha,\eta^\alpha)$ be a $k$-cocontact manifold with dimension $\dim \mathcal{M} = k + n + kn + k$ such that there exists an integrable subdistribution $\mathcal{V}\subset\D^C$ with $ \rk \mathcal{V} = nk$. Then, around every point of $\mathcal{M}$ there exist local coordinates $(x^\alpha,y^a,p_a^\alpha,s^\alpha)$, where $1\leq\alpha\leq k$ and $1\leq a\leq n$, such that, locally,
    $$ \tau^\alpha = \d x^\alpha\,,\qquad \eta^\alpha = \d s^\alpha - p_a^\alpha\d y^a\,. $$
    Using these coordinates,
    $$ \D^\rmR = \left \langle R_\alpha^x = \parder{}{x^\alpha}\,,\ R_\alpha^s = \parder{}{s^\alpha} \right\rangle\,,\qquad \mathcal{V} = \left\langle \parder{}{p_a^\alpha} \right\rangle\,. $$
    These coordinates are called {\sl\textbf{Darboux coordinates}} of the $k$-cocontact manifold $(\mathcal{M},\tau^\alpha,\eta^\alpha)$.
\end{theorem}

Taking into account the previous theorem, we can consider the manifold introduced in Example \ref{ex:canonical-k-cocontact-structure} as the canonical model for $k$-cocontact structures.

\subsubsection{\texorpdfstring{$k$}{}-Cocontact Hamiltonian systems}
\label{secHamsys2}

This section introduces the notion of $k$-cocontact Hamiltonian system and its Hamilton--De Donder--Weyl equations. The existence of solutions to these equations is proved. We provide local expressions of the Hamilton--De Donder--Weyl equations for maps and $k$-vector fields in both adapted and Darboux coordinates.

\begin{definition}
    A {\sl\textbf{$k$-cocontact Hamiltonian system}} is a tuple $(\mathcal{M},\tau^\alpha,\eta^\alpha,h)$, where $(\tau^\alpha,\eta^\alpha)$ is a $k$-cocontact structure on the manifold $\mathcal{M}$ and $h\colon \mathcal{M}\to\R$ is a {\sl\textbf{Hamiltonian function}}. Given a map $\psi\colon D\subset\Rk\to \mathcal{M}$, the {\sl\textbf{$k$-cocontact Hamilton--De Donder--Weyl equations for the map $\psi$}} are
    \begin{equation}\label{eq:HDW-map}
        \begin{dcases}
            \inn{\psi_\alpha'}\d\eta^\alpha = \left(\d h - (\Lie_{R_\alpha^x}h)\tau^\alpha - (\Lie_{R_\alpha^s}h)\eta^\alpha\right)\circ\psi\,,\\
            \inn{\psi_\alpha'}\eta^\alpha = -h\circ\psi\,,\\
            \inn{\psi_\alpha'}\tau^\beta = \delta_\alpha^\beta\,.
        \end{dcases}
    \end{equation}
\end{definition}

Now we are going to look at the expression in coordinates of the Hamilton--De Donder--Weyl equations \eqref{eq:HDW-map}. 
Consider first the adapted coordinates $(x^\alpha,z^I,s^\alpha)$,
 where $t = (t^1,\dotsc,t^k)\in\Rk$. In these coordinates,
$$ R_\alpha^x = \parder{}{x^\alpha}\,,\quad \tau^\alpha = \d x^\alpha\,,\quad R_\alpha^s = \parder{}{s^\alpha}\,,\quad \eta^\alpha = \d s^\alpha - f_I^\alpha(z^J)\d z^I\,,\quad \d\eta^\alpha = \frac{1}{2}\omega^\alpha_{IJ}\d z^I\wedge\d z^J\,, $$
where $\omega^\alpha_{IJ} = \dparder{f_I^\alpha}{z^J} - \dparder{f_J^\alpha}{z^I}$. Consider a map $\psi\colon D\subset\Rk\to \mathcal{M}$ with local expression $\psi(t) = (x^\alpha(t),x^I(t),s^\alpha(t))$. Then,
$$ \psi'_\alpha = \left(x^\beta, z^I, s^\beta; \parder{x^\beta}{t^\alpha}, \parder{z^I}{t^\alpha}, \parder{s^\beta}{t^\alpha}\right)\,. $$
Then, the Hamilton--De Donder--Weyl equations in adapted coordinates read
\begin{equation}
    \begin{dcases}
        \parder{x^J}{t^\alpha}\omega_{JI}^\alpha = \left(\parder{h}{x^I} + \parder{h}{t^\alpha}f_I^\alpha\right)\circ\psi\,,\\
        \parder{s^\alpha}{t^\alpha} - f_I^\alpha\parder{x^I}{t^\alpha} = -h\circ\psi\,,\\
        \parder{x^\alpha}{t^\beta} = \delta^\alpha_\beta\,.
    \end{dcases}
\end{equation}

Furthermore, if the local expression in Darboux coordinates of a map $\psi\colon D\subset\Rk\to \mathcal{M}$ is $\psi(t) = (x^\alpha(t), y^a(t), p_a^\alpha(t), s^\alpha(t))$. Then, the Hamilton--De Donder--Weyl equations in Darboux coordinates read
\begin{equation}\label{eq:HDW-map-Darboux}
    \begin{dcases}
        \parder{x\beta}{t^\alpha} = \delta_\alpha^\beta\,,\\
        \parder{y^a}{t^\alpha} = \parder{h}{p_a^\alpha}\circ\psi\,,\\
        \parder{p_a^\alpha}{t^\alpha} = -\left( \parder{h}{y^a} + p_a^\alpha\parder{h}{s^\alpha} \right)\circ\psi\,,\\
        \parder{s^\alpha}{t^\alpha} = \left( p_a^\alpha\parder{h}{p_a^\alpha} - h \right)\circ\psi\,.
    \end{dcases}
\end{equation}

\begin{definition}
    Consider a $k$-cocontact Hamiltonian system $(\mathcal{M},\tau^\alpha,\eta^\alpha, h)$. The {\sl\textbf{$k$-cocontact Hamilton--De Donder--Weyl equations for a $k$-vector field}} $\bfX = (X_\alpha)\in\X^k(\mathcal{M})$ are
    \begin{equation}\label{eq:HDW-field}
        \begin{dcases}
            \inn{X_\alpha}\d\eta^\alpha = \d h - (\Lie_{R_\alpha^x}h)\tau^\alpha - (\Lie_{R_\alpha^s}h)\eta^\alpha\,,\\
            \inn{X_\alpha}\eta^\alpha = -h\,,\\
            \inn{X_\alpha}\tau^\beta = \delta_\alpha^\beta\,.
        \end{dcases}
    \end{equation}
    A $k$-vector field solution to these equations is a {\sl\textbf{$k$-cocontact Hamiltonian $k$-vector field}}. We will denote this set of $k$-vector fields by $\Xh^k(\mathcal{M})$.
\end{definition}

\begin{proposition}\label{prop:k-cocontact-HDW-have-solutions}
    The $k$-cocontact Hamilton--De Donder--Weyl equations \eqref{eq:HDW-field} admit solutions. They are not unique if $k > 1$.
\end{proposition}

Consider a $k$-vector field $\bfX = (X_1,\dotsc,X_k)\in\X^k(\mathcal{M})$ with local expression in adapted coordinates
$$ X_\alpha = A_\alpha^\beta\parder{}{x^\beta} + B_\alpha^I\parder{}{z^I} + D_\alpha^\beta\parder{}{s^\beta}\,. $$
Thus, equations \eqref{eq:HDW-field} in adapted coordinates read
\begin{equation}
    \begin{dcases}
        A_\alpha^\beta = \delta_\alpha^\beta\,,\\
        B_\alpha^J\omega_{JI}^\alpha = \parder{h}{z^I} + \parder{h}{s^\alpha}f_I^\alpha\,,\\
        D_\alpha^\alpha - f_I^\alpha B_\alpha^I = -h\,.
    \end{dcases}
\end{equation}

On the other hand, consider a $k$-vector field $\bfX = (X_1,\dotsc,X_k)\in\X^k(\mathcal{M})$ with local expression in Darboux coordinates
$$ X_\alpha = A_\alpha^\beta\parder{}{x^\beta} + B_\alpha^a\parder{}{y^a} + C_{\alpha i}^\beta\parder{}{p_a^\beta} + D_\alpha^\beta\parder{}{s^\beta}\,. $$
Imposing equations \eqref{eq:HDW-field}, we get the conditions
\begin{equation}\label{eq:HDW-field-Darboux}
    \begin{dcases}
        A_\alpha^\beta = \delta_\alpha^\beta\,,\\
        B_\alpha^a = \parder{h}{p_a^\alpha}\,,\\
        C_{\alpha i}^\alpha = -\left( \parder{h}{y^a} + p_a^\alpha\parder{h}{s^\alpha} \right)\,,\\
        D_\alpha^\alpha = p_a^\alpha \parder{h}{p_a^\alpha} - h\,.
    \end{dcases}
\end{equation}

\begin{proposition}
    Let $\bfX \in\X^k(\mathcal{M})$ be an integrable $k$-vector field. Then $\bfX$ is a solution to \eqref{eq:HDW-field} if and only if every integral section of $\bfX$ satisfies the $k$-cocontact Hamilton--De Donder--Weyl equations \eqref{eq:HDW-map}.
\end{proposition}

It is worth noting that, as in the $k$-symplectic and $k$-contact cases, equations \eqref{eq:HDW-map} and \eqref{eq:HDW-field} are not completely equivalent since a solution to \eqref{eq:HDW-map} may not be an integral section of an integrable $k$-vector field $\bfX$ solution to equations \eqref{eq:HDW-field}.

The following proposition provides an alternative way of writing the $k$-cocontact Hamilton--De Donder--Weyl equations for $k$-vector fields.

\begin{proposition}
    The $k$-cocontact Hamilton--De Donder--Weyl equations \eqref{eq:HDW-field} are equivalent to
    \begin{equation}
        \begin{dcases}
            \Lie_{X_\alpha}\eta^\alpha = -(\Lie_{R_\alpha^x}h)\tau^\alpha - (\Lie_{R_\alpha^s}h)\eta^\alpha\,,\\
            \inn{X_\alpha}\eta^\alpha = -h\,,\\
            \inn{X_\alpha}\tau^\beta = \delta_\alpha^\beta\,.
        \end{dcases}
    \end{equation}
\end{proposition}

\subsection{\texorpdfstring{$k$}{}-Cocontact Lagrangian formalism}

In this section we devise the Lagrangian counterpart of the formulations introduced in the previous section. We begin by introducing the geometric structures of the phase bundle and defining the notion of second-order partial differential equation. In second place, we develop the Lagrangian formalism and introduce the $k$-cocontact Euler--Lagrange equations as the Hamilton--De Donder--Weyl of a $k$-cocontact Lagrangian system.

\subsubsection{Geometry of the phase bundle}

The phase space for the Lagrangian counterpart of the $k$-cocontact formalism will be the product bundle $\mathbf{P}= \Rk\times\bigoplus^k\T Q\times\Rk$ endowed with natural coordinates $(x^\alpha, y^a, y^a_\alpha, s^\alpha)$. We have the natural projections
\begin{align*}
    \tau_1^\alpha&\colon \mathbf{P}\to\R\ ,&& \tau_1^\alpha(x^1,\dotsc,x^k, {v_q}_1, \dotsc, {v_q}_k, s^1,\dotsc,s^k) = x^\alpha\,,\\
    \tau_2&\colon \mathbf{P}\to\textstyle\bigoplus\nolimits^k\T Q\ ,&& \tau_2(x^1,\dotsc,x^k, {v_q}_1, \dotsc, {v_q}_k, s^1,\dotsc,s^k) = ({v_q}_1, \dotsc, {v_q}_k)\,,\\
    \tau_2^\alpha&\colon \mathbf{P}\to\T Q\ ,&& \tau_2^\alpha(x^1,\dotsc,x^k, {v_q}_1, \dotsc, {v_q}_k, s^1,\dotsc,s^k) = {v_q}_\alpha\,,\\
    \tau^\alpha&\colon \textstyle\bigoplus\nolimits^k\T Q\to\T Q\ ,&& \tau^\alpha(x^1,\dotsc,x^k, {v_q}_1, \dotsc, {v_q}_k, s^1,\dotsc,s^k) = {v_q}_\alpha\,,\\
    \tau_3^\alpha&\colon \mathbf{P}\to\R\ ,&& \tau_3^\alpha(x^1,\dotsc,x^k, {v_q}_1, \dotsc, {v_q}_k, s^1,\dotsc,s^k) = s^\alpha\,,\\
    \tau_\circ&\colon \mathbf{P}\to \Rk\times Q\times\Rk\ ,&& \tau_\circ(x^1,\dotsc,x^k, {v_q}_1, \dotsc, {v_q}_k, s^1,\dotsc,s^k) = (x^1,\dotsc,x^k, q, s^1,\dotsc,s^k)\,, 
\end{align*}
which can be summarized in the following diagram:
\begin{center}
    \begin{tikzcd}
        \R & \Rk\times\bigoplus^k\T Q\times\Rk \arrow[l, swap, "\tau_1^\alpha", ] \arrow[r, "\tau_3^\alpha"] \arrow[dd, "\tau_2"] \arrow[ddd, bend right=50, swap, "\tau_\circ"] \arrow[ddr, "\tau_2^\alpha"] & \R \\
        &  && \\
        & \bigoplus^k\T Q \arrow[r, swap, "\tau^\alpha"] & \T Q \\
        & \Rk\times Q\times\Rk 
    \end{tikzcd}
\end{center}

Since the bundle $\tau_2 \colon \Rk\times\bigoplus^k\T Q\times\Rk\to\bigoplus^k\T Q$ is trivial, the canonical structures in $\bigoplus^k\T Q$, namely the canonical $k$-tangent structure $(J^\alpha)$ and the Liouville vector field $\Delta$, can be extended to $\Rk\times\bigoplus^k\T Q\times\Rk$ in a natural way. Their local expression remain the same:
$$ J^\alpha = \parder{}{y^a_\alpha}\otimes\d y^a\,,\qquad \Delta = y^a_\alpha\parder{}{y^a_\alpha}\,. $$
These canonical structures can be used to extend the notion of \textsc{sopde} (second-order partial differential equation) to the bundle $\Rk\times\bigoplus^k\T Q\times\Rk$:

\begin{definition}
    A $k$-vector field $\mathbf{\Gamma} = (\Gamma_\alpha)\in\X^k(\Rk\times\bigoplus^k\T Q\times\Rk)$ is a {\sl\textbf{second-order partial differential equation}} or \textsc{sopde} if $J^\alpha(\Gamma_\alpha) = \Delta$.
\end{definition}

A straightforward computations shows that the local expression of a \textsc{sopde} reads
$$ \Gamma_\alpha = A_\alpha^\beta\parder{}{x^\beta} + y^a_\alpha\parder{}{y^a} + C_{\alpha\beta}^a\parder{}{y_\beta^a} + D_\alpha^\beta\parder{}{s^\beta} \,. $$

\begin{definition}
    Consider a map $\psi\colon\Rk\to\Rk\times Q\times\Rk$ with $\psi = (x^\alpha, \phi, s^\alpha)$, where $\phi\colon\Rk\to Q$. The {\sl\textbf{first prolongation}} of $\psi$ to $\Rk\times\bigoplus^k\T Q\times\Rk$ is the map $\psi'\colon\Rk\to\Rk\times\bigoplus^k\T Q\times\Rk$ given by $\psi' = (x^\alpha,\phi',s^\alpha)$, where $\phi'$ is the first prolongation of $\phi$ to $\bigoplus^k\T Q$. The map $\psi'$ is said to be {\sl\textbf{holonomic}}.
\end{definition}
Let $\psi\colon\Rk\to\Rk\times Q\times\Rk$ be a map with local expression $\psi(r) = (x^\alpha(r), y^a(r), s^\alpha(r))$, where $r\in\Rk$. Then, its first prolongation has local expression
$$ \psi'(t) = \left( x^\alpha(t), y^a(t), \parder{y^a}{t^\alpha}(t), s^\alpha(r) \right)\,. $$

\begin{proposition}
    An integrable $k$-vector field $\mathbf{\Gamma} \in\X^k(\Rk\times\bigoplus^k\T Q\times\Rk)$ is a \textsc{sopde} if and only if its integral sections are holonomic.
\end{proposition}

It is important to point out that the product manifold $\Rk\times\bigoplus^k\T Q\times\Rk$ does not have a canonical $k$-cocontact structure, in contrast to what happens to the manifold $\Rk\times\bigoplus^k\cT Q\times\Rk$, where we do have a natural $k$-cocontact structure as seen in Example \ref{ex:canonical-k-cocontact-structure}. In what follows we will show that, in favourable cases, given a Lagrangian function $L$ defined on $\Rk\times\bigoplus^k\T Q\times\Rk$ one can build up a $k$-cocontact structure.

\subsubsection{\texorpdfstring{$k$}--cocontact Lagrangian systems}

\begin{definition}
    A {\sl\textbf{Lagrangian function}} on $\Rk\times\bigoplus^k\T Q\times\Rk$ is a function $L\colon\Rk\times\bigoplus^k\T Q\times\Rk\to\R$.
    \begin{itemize}
        \item The {\sl\textbf{Lagrangian energy}} associated to the Lagrangian function $L$ is the function $E_L\in\Cinfty(\Rk\times\bigoplus^k\T Q\times\Rk)$ given by $E_L = \Delta(L) - L$.
        \item The {\sl\textbf{Cartan forms}} associated to the Lagrangian $L$ are
        $$ \theta_L^\alpha = \transp{J^\alpha}\circ\d L\in\Omega^1(\Rk\times\textstyle\bigoplus\nolimits^k\T Q\times\Rk)\,,\qquad \omega_L^\alpha = -\d\theta_L^\alpha\in\Omega^2(\Rk\times\textstyle\bigoplus\nolimits^k\T Q\times\Rk)\,, $$
        where $\transp{J^\alpha}$ denotes the transpose of $J^\alpha$.
        \item The {\sl\textbf{contact forms}} associated to the Lagrangian $L$ are
        $$ \eta_L^\alpha = \d s^\alpha - \theta_L^\alpha\in\Omega^1(\Rk\times\textstyle\bigoplus\nolimits^k\T Q\times\Rk)\,. $$
        \item The couple $(\Rk\times\bigoplus^k\T Q\times\Rk, L)$ is a $k$-{\sl\textbf{cocontact Lagrangian system}}.
    \end{itemize}
\end{definition}

It is clear that $\d\eta_L^\alpha = \omega_L^\alpha$. The local expressions in natural coordinates $(x^\alpha, y^a, y_\alpha^a, s^\alpha)$ of the objects introduced in the previous definition are
\begin{align}
    E_L &= y_\alpha^a\parder{L}{y_\alpha^a} - L\,,\\
    \theta^\alpha_L &= \parder{L}{y_\alpha^a}\d y^a\,,\\
    \eta_L^\alpha &= \d s^\alpha - \parder{L}{y_\alpha^a}\d y^a\,,\\
    \d\eta_L^\alpha &= \parderr{L}{x^\beta}{y_\alpha^a}\d y^a\wedge\d x^\beta + \parderr{L}{y^b}{y_\alpha^a}\d y^a\wedge\d y^b + \parderr{L}{y_\beta^b}{y_\alpha^a}\d y^a\wedge\d y_\beta^b + \parderr{L}{s^\beta}{y_\alpha^a}\d y^a\wedge\d s^\beta\,.
\end{align}

\begin{definition}
    Given a Lagrangian function $L\colon\Rk\times\bigoplus^k\T Q\times\Rk\to\R$, the {\sl\textbf{Legendre map}} of $L$ is its fibre derivative as a function on the vector bundle $\tau_\circ\colon\Rk\times\bigoplus^k\T Q\times\Rk\to \Rk\times Q\times\Rk$. Namely, the Legendre map of a Lagrangian function $L\colon\Rk\times\bigoplus^k\T Q\times\Rk\to\R$ is the map
    $$ \F L\colon \Rk\times\textstyle\bigoplus\nolimits^k\T Q\times\Rk\longrightarrow\Rk\times\textstyle\bigoplus\nolimits^k\cT Q\times\Rk $$
    given by
    $$ \F L(t, {v_q}_1,\dotsc, {v_q}_k, z) = (t, \F L(t, \cdot, z)({v_q}_1,\dotsc, {v_q}_k), z)\,, $$
    where $\F L(t, \cdot, z)$ denotes the Legendre map of the Lagrangian function with $t$ and $z$ freezed.
\end{definition}

In natural coordinates $(x^\alpha, y^a, y_\alpha^a,s^\alpha)$, the Legendre map has local expression
$$
    \F L(x^\alpha, y^a, y^a_\alpha, s^\alpha) = \left( x^\alpha, y^a, \parder{L}{y^a_\alpha}, s^\alpha \right)\,.
$$
\begin{proposition}
    The Cartan forms satisfy
    $$ \theta_L^\alpha = (\pi_2^\alpha\circ\F L)^\ast\theta\,,\qquad \omega_L^\alpha = (\pi_2^\alpha\circ\F L)^\ast\omega\,, $$
    where $\theta\in\Omega^1(\cT Q)$ and $\omega = -\d\theta\in\Omega^2(\cT Q)$ are the Liouville and symplectic canonical forms of the cotangent bundle $\cT Q$.
\end{proposition}

The regularity of the Legendre map characterizes the Lagrangian functions which yield $k$-cocontact structures on the phase bundle $\Rk\times\bigoplus^k\T Q\times\Rk$.

\begin{proposition}\label{prop:regular-lagrangian}
Consider a Lagrangian function $L\colon \Rk\times\bigoplus^k\T Q\times\Rk\to\R$. The following are equivalent:
    \begin{enumerate}[{\rm(1)}]
\item The Legendre map $\F L$ is a local diffeomorphism.
\item The family $(\tau^\alpha = \d x^\alpha, \eta^\alpha_L)$ is a $k$-cocontact structure on $\Rk\times\bigoplus^k\T Q\times\Rk$.
    \end{enumerate}
\end{proposition}
\begin{proof}
    Taking natural coordinates $(x^\alpha, y^a, y^a_\alpha, s^\alpha)$, We have
    \begin{align}
        \F^2 L(x^\alpha, y^a, y^a_\alpha, s^\alpha) &= \left( x^\alpha, y^a, W_{ij}^{\alpha\beta}, s^\alpha \right)\,,\quad \text{where } W_{ij}^{\alpha\beta} = \left( \parderr{L}{y^a_\alpha}{y^b_\beta} \right)\,.
    \end{align}
    The conditions in the proposition mean that the matrix $W = (W_{ij}^{\alpha\beta})$ is everywhere nonsingular.
\end{proof}

\begin{definition}
    A Lagrangian function $L\colon \Rk\times\bigoplus^k\T Q\times\Rk\to\R$ is said to be {\sl\textbf{regular}} if the equivalent statements in Proposition \ref{prop:regular-lagrangian} hold. Otherwise $L$ is said to be {\sl\textbf{singular}}. In addition, if the Legendre map $\F L$ is a global diffeomorphism, $L$ is a {\sl\textbf{hyperregular}} Lagrangian.
\end{definition}

Let $(\Rk\times\bigoplus^k\T Q\times\Rk, L)$ be a regular $k$-cocontact Lagrangian system. By Theorem \ref{thm:Reeb-vector-fields}, the Reeb vector fields $(R^x_L)_\alpha, (R^s_L)_\alpha\in\X(\Rk\times\bigoplus^k\T Q\times\Rk)$ are uniquely given by the relations
\begin{gather}
    \inn{(R^x_L)_\alpha}\d\eta_L^\beta = 0\,,\qquad \inn{(R^x_L)_\alpha}\eta_L^\beta = 0\,,\qquad \inn{(R^x_L)_\alpha}\d x^\beta = \delta_\alpha^\beta\,,\\
    \inn{(R^s_L)_\alpha}\d\eta_L^\beta = 0\,,\qquad \inn{(R^s_L)_\alpha}\eta_L^\beta = \delta_\alpha^\beta\,,\qquad \inn{(R^s_L)_\alpha}\d x^\beta = 0\,.
\end{gather}
The local expressions of the Reeb vector fields are
\begin{gather}
    (R^x_L)_\alpha = \parder{}{x^\alpha} - W_{\gamma\beta}^{ji}\parderr{L}{x^\alpha}{y^b_\gamma}\parder{}{y^a_\beta}\,,\quad
    (R^s_L)_\alpha = \parder{}{s^\alpha} - W_{\gamma\beta}^{ji}\parderr{L}{s^\alpha}{y^b_\gamma}\parder{}{y^a_\beta}\,,
\end{gather}
where $W_{\alpha\beta}^{ij}$ is inverse of the Hessian matrix $W_{ij}^{\alpha\beta} = \left( \dparderr{L}{y^a_\alpha}{y^b_\beta} \right)$, namely
$$ W_{\alpha\beta}^{ij}\parderr{L}{y^b_\beta}{y^k_\gamma} = \delta^a_k\delta^\gamma_\alpha\,. $$

\subsubsection{\texorpdfstring{$k$}--cocontact Euler--Lagrange equations}

We have proved in the previous section that every regular $k$-cocontact Lagrangian system $(\Rk\times\bigoplus^k\T Q\times\Rk,L)$ yields the $k$-cocontact Hamiltonian system $(\Rk\times\bigoplus^k\T Q\times\Rk,\tau^\alpha = \d x^\alpha, \eta^\alpha, E_L)$. Taking this into account, we can define:

\begin{definition}
    Let $(\Rk\times\bigoplus^k\T Q\times\Rk,L)$ be a $k$-cocontact Lagrangian system. The {\sl\textbf{$k$-cocontact Euler--Lagrange equations for a holonomic map}} $\psi\colon\Rk\to \Rk\times\bigoplus^k\T Q\times\Rk$ are
    \begin{equation}\label{eq:EL-map}
        \begin{dcases}
            \inn{\psi_\alpha'}\d\eta_L^\alpha = \left(\d E_L - (\Lie_{(R^x_L)_\alpha}E_L)\d x^\alpha - (\Lie_{(R^s_L)_\alpha}E_L)\eta_L^\alpha\right)\circ\psi\,,\\
            \inn{\psi_\alpha'}\eta_L^\alpha = -E_L\circ\psi\,,\\
            \inn{\psi_\alpha'}\d x^\beta = \delta_\alpha^\beta\,.
        \end{dcases}
    \end{equation}
    The {\sl\textbf{$k$-cocontact Lagrangian equations for a $k$-vector field}} $\ \bfX = (X_\alpha)\in\X^k(\Rk\times\bigoplus^k\T Q\times\Rk)$ are
    \begin{equation}\label{eq:EL-field}
        \begin{dcases}
            \inn{X_\alpha}\d\eta_L^\alpha = \d E_L - (\Lie_{(R^x_L)_\alpha}E_L)\d x^\alpha - (\Lie_{(R^s_L)_\alpha}E_L)\eta_L^\alpha\,,\\
            \inn{X_\alpha}\eta_L^\alpha = -E_L\,,\\
            \inn{X_\alpha}\d x^\beta = \delta_\alpha^\beta\,.
        \end{dcases}
    \end{equation}
    A $k$-vector field $\bfX$ solution to equations \eqref{eq:EL-field} is said to be a {\sl\textbf{$k$-cocontact Lagrangian vector field}}.
\end{definition}

The next proposition states that, if the Lagrangian $L$ is regular, the Lagrangian equations \eqref{eq:EL-field} always have solutions, although they are not unique in general.

\begin{proposition}
    Consider a regular $k$-cocontact Lagrangian system $(\Rk\times\bigoplus^k\T Q\times\Rk, L)$. Then, the $k$-cocontact Lagrangian equations \eqref{eq:EL-field} admit solutions. They are not unique if $k>1$.
\end{proposition}

Consider a map $\psi\colon\Rk\to\Rk\times\bigoplus^k\T Q\times\Rk$ with local expression in natural coordinates $\psi(r) = (x^\alpha(r),y^a(r),y^a_\alpha(r),s^\alpha(r))$, where $r = (r^1,\dotsc,r^k)\in\Rk$. Then, equations \eqref{eq:EL-map} for the map $\psi$ read
\begin{equation}\label{eq:EL-map-coordinates}
    \begin{dcases}
        \parder{x^\beta}{r^\alpha} = \delta_\alpha^\beta\,,\\
        \parder{}{r^\alpha}\left( \parder{L}{y^a_\alpha} \circ\psi\right) = \left( \parder{L}{y^a} + \parder{L}{s^\alpha}\parder{L}{y^a_\alpha} \right)\circ\psi\,,\\
        \parder{(s^\alpha)}{r^\alpha} = L\circ\psi\,.
    \end{dcases}
\end{equation}
For a $k$-vector field $\bfX = (X_\alpha)\in\X^k(\Rk\times\bigoplus^k\T Q\times\Rk)$, with local expression in natural coordinates
$$ X_\alpha = A_\alpha^\beta\parder{}{x^\beta} + B_\alpha^a\parder{}{y^a} + C_{\alpha \beta}^a\parder{}{y_\beta^a} + D_\alpha^\beta\parder{}{s^\beta}\,, $$
equations \eqref{eq:EL-field} read
\begin{align}
    0 &= A_\alpha^\beta - \delta_\alpha^\beta\,,\label{eq:k-contact-Lagrangian-1}\\
    0 &= \left( B_\alpha^b - y_\alpha^b \right)\parderr{L}{y^b_\alpha}{s^\beta}\,,\label{eq:k-contact-Lagrangian-2}\\
    0 &= \left( B_\alpha^b - y_\alpha^b \right)\parderr{L}{y^b_\alpha}{x^\beta}\,,\label{eq:k-contact-Lagrangian-2.5}\\
    0 &= \left( B_\alpha^b - y_\alpha^b \right)\parderr{L}{y^a_\beta}{y^b_\alpha}\,,\label{eq:k-contact-Lagrangian-3}\\
    0 &= \left( B_\alpha^b - y_\alpha^b \right)\parderr{L}{y^a}{y^b_\alpha} + \parder{L}{y^a} - \parderr{L}{x^\alpha}{y_\alpha^a} - \parderr{L}{y^b}{y^a_\alpha}B_\alpha^b \nonumber\\
    & \qquad\qquad\qquad\qquad\qquad\qquad - \parderr{L}{y^b_\beta}{y^a_\alpha}C_{\alpha\beta}^b - \parderr{L}{s^\beta}{y^a_\alpha}D_\alpha^\beta + \parder{L}{s^\alpha}\parder{L}{y^a_\alpha} \,,\label{eq:k-contact-Lagrangian-4}\\
    0 &= L + \parder{L}{y^a_\alpha}\left( B_\alpha^a - y^a_\alpha \right) - D_\alpha^\alpha\,.\label{eq:k-contact-Lagrangian-5}
\end{align}
If the Lagrangian function $L$ is regular, equations \eqref{eq:k-contact-Lagrangian-3} yield the conditions $B_\alpha^a = y_\alpha^a$, namely the $k$-vector field $\bfX$ has to be a \textsc{sopde}. In this case, equations \eqref{eq:k-contact-Lagrangian-2} and \eqref{eq:k-contact-Lagrangian-2.5} hold identically and equations \eqref{eq:k-contact-Lagrangian-1}, \eqref{eq:k-contact-Lagrangian-4} and \eqref{eq:k-contact-Lagrangian-5} yield
\begin{align}
    A_\alpha^\beta &= \delta_\alpha^\beta\,,\label{eq:k-contact-Lagrangian-regular-1}\\
    \parder{L}{y^a} + \parder{L}{s^\alpha}\parder{L}{y^a_\alpha} &= \parderr{L}{x^\alpha}{y_\alpha^a} + \parderr{L}{y^b}{y^a_\alpha}y^b_\alpha + \parderr{L}{y^b_\beta}{y^a_\alpha}C_{\alpha\beta}^b + \parderr{L}{s^\beta}{y^a_\alpha}D_\alpha^\beta\,,\label{eq:k-contact-Lagrangian-regular-2}\\
    D_\alpha^\alpha &= \L\,.\label{eq:k-contact-Lagrangian-regular-3}
\end{align}
If the \textsc{sopde} $\bfX$ is integrable, equations \eqref{eq:k-contact-Lagrangian-regular-1}, \eqref{eq:k-contact-Lagrangian-regular-2} and \eqref{eq:k-contact-Lagrangian-regular-3} are the Euler--Lagrange equations \eqref{eq:EL-map-coordinates} for its integral maps. Thus, we have proved the following:

\begin{proposition}\label{prop:Euler-Lagrange}
    Let $L\colon \Rk\times\bigoplus^k\T Q\times\Rk\to\R$ be a regular Lagrangian and consider a Lagrangian $k$-vector field $\bfX$, namely a solution to equations \eqref{eq:EL-field}. Then $\bfX$ is a \textsc{sopde} and if, in addition, $\bfX$ is integrable, its integral sections are solutions to the $k$-cocontact Euler--Lagrange equations \eqref{eq:EL-map}.

    The \textsc{sopde} $\bfX$ is called an {\sl\textbf{Euler--Lagrange $k$-vector field}} associated to the Lagrangian function $L$.
\end{proposition}

\begin{remark}
    In the case $k = 1$, we recover the cocontact Lagrangian formalism presented in \cite{LGGMR_23} for time-dependent contact Lagrangian systems.
\end{remark}

\subsection{\texorpdfstring{$k$}{}-Cocontact Hamiltonian formalism}

Now, the developments stated in Section \ref{secHamsys2} are used to develop the Hamiltonian formalism for action-dependent field theories in this formulation.

In the $k$-cocontact formulation,
action-dependent Hamiltonian field theories is developed in the product bundle $\mathbf{P}^*= \Rk\times\bigoplus^k\T^*Q\times\Rk$ endowed with natural coordinates $(x^\alpha, y^a, p^a_\alpha, s^\alpha)$.
Then, regular or $k$-cocontact Hamiltonian field theories take place in the canonical $k$-cocontact manifold
$(\Rk\times\oplus^k\Tan^*Q\times\Rk,\tau^\alpha,\theta^\alpha)$,
giving a {\sl Hamiltonian function} $\mathcal{H}\in\Cinfty(\Rk\times\oplus^k\Tan^*Q\times\Rk)$.

\begin{remark}
[{\sl The canonical $k$-cocontact Hamiltonian system associated with a $k$-cocontact Lagrangian system\/}]
Let $(\Rk\times\bigoplus^k\T Q\times\Rk, L)$ be a $k$-cocontact Lagrangian system.
    If the Lagrangian function $L$ is regular or hyperregular, the Legendre map $\F L$ is a (local) diffeomorphism between $\Rk\times\bigoplus^k\T Q\times\Rk$ and $\Rk\times\bigoplus^k\cT Q\times\Rk$ such that $\F L^\ast\eta^\alpha = \eta_L^\alpha$. In addition, there exists, at least locally, a function $h\in\Cinfty(\Rk\times\bigoplus^k\cT Q\times\Rk)$ such that $h \circ \F L = E_L$. Then, we have the $k$-cocontact Hamiltonian system $(\Rk\times\bigoplus^k\cT Q\times\Rk, \eta^\alpha, h)$, for which $\F L_\ast (R^x_L)_\alpha = R^x_\alpha$ and $\F L_\ast (R^s_L)_\alpha = R^s_\alpha$. If $\mathbf{\Gamma}$ is an Euler--Lagrange $k$-vector field associated to the Lagrangian function $L$ in $\Rk\times\bigoplus^k\T Q\times\Rk$, we have that the $k$-vector field $\bfX = \F L_\ast\mathbf{\Gamma}$ is a $k$-cocontact Hamiltonian $k$-vector field associated to $h$ in $\Rk\times\bigoplus^k\T Q\times\Rk$, and conversely.
\end{remark}

\section{Multicontact field theories}
\label{multicontactf}

The {\sl multicontact formulation} is the most general geometric framework
to describe action-dependent field theories.
It was first introduced in \cite{LGMRR_23,LGMRR_25},
although a more general definition of multicontact structure has recently been proposed in \cite{LIR_25}.
(You can find all the details on this structure and its applications in these references).

\subsection{Multicontact structures}

First, following \cite{LIR_25}, we define:

\begin{definition}
Let $\mathcal{M}$ be a differentiable manifold.
A form $\Theta\in\df^k(\mathcal{M})$, with $\dim \mathcal{M}>k$,
is a {\sl \textbf{multicontact form}} in $\mathcal{M}$ if:
\begin{enumerate}[{\rm(1)}]
    \item
$\ker \Theta\cap\ker\d \Theta=\{0\}$, and
    \item 
$\ker\d\Theta\neq\{0\}$.
    \end{enumerate}
    Then, the pair $(\mathcal{M},\Theta)$ is said to be a {\sl \textbf{multicontact manifold}}.
\end{definition}

The properties of this kind of structure have been studied in \cite{LIR_25}.
Nevertheless, as is already the case with multisymplectic structures \cite{LMS_03,Mar_88}, 
the existence of adapted or Darboux coordinates is not guaranteed for these multicontact forms, 
unless additional conditions are imposed \cite{LGMRR_23}.

Thus, let $\mathcal{M}$ with $\dim{\mathcal{M}}=k+N$ be a manifold, with $N\geq k\geq 1$;
and let $\Theta,\omega\in\df^k(\mathcal{M})$ be two $k$-forms with constant ranks.
Given a regular distribution $\mathcal{D}\subset\Tan \mathcal{M}$, consider the $\Cinfty(\mathcal{M})$-module of sections $\Gamma(\mathcal{D})$ and, for every $m\in\mathbb{N}$, define
the set of $m$-forms on $\mathcal{M}$ vanishing by the vector fields in $\Gamma(\mathcal{D})$; that is,
\[
\mathcal{A}^m(\mathcal{D}):=\big\{ \alpha\in\df^m(\mathcal{M})\mid
\inn{Z}\alpha=0 \,,\ \mbox{\rm for every}\ Z\in\Gamma({\cal D})\big\}=
\big\{ \alpha\in\df^m(\mathcal{M}) \mid 
\Gamma(\cal D)\subset\ker\alpha \big\} \,,
\]
where $\ker\alpha =\{ Z\in\vf(\mathcal{M})\mid \inn{Z}\alpha=0\}$
(it is the `$1$-$\ker$ of a form $\alpha\in\df^m(\mathcal{M})$, with $m>1$). Then,

\begin{definition}
\label{def:reeb_dist}
The {\sl\textbf{Reeb distribution}} associated with
the pair $(\Theta,\omega)$ is the distribution
${\cal D}^{\mathfrak{R}}\subset\Tan \mathcal{M}$ defined as
\[
{\cal D}^{\mathfrak{R}}=\big\{ Z\in\ker\omega\ \mid\ \inn{Z}\dd\Theta\in {\cal A}^k(\ker\omega)\big\} \ .
\]
The set of sections of the Reeb distribution is denoted by $\mathfrak{R}:=\Gamma({\cal D}^{\mathfrak{R}})$,
and its elements $R\in\mathfrak{R}$ are called {\sl\textbf{Reeb vector fields}}.
If $\ker\omega$ has a constant rank, then
\begin{equation}
\label{Reebdef}
\mathfrak{R}=\big\{ R\in\Gamma(\ker\omega)\,|\, \inn{R}\dd\Theta\in {\cal A}^k(\ker\omega)\big\} \,.
\end{equation}
\end{definition}

Note that $\ker\omega\cap\ker\dd\Theta\subset{\cal D}^\mathfrak{R}$.
Furthermore, if $\omega\in\df^k(\mathcal{M})$ is a closed form and has a constant rank, then $\mathfrak{R}$ is involutive. Therefore:

\begin{definition}
\label{multicontactdef}
The pair $(\Theta,\omega)$ is a {\sl\textbf{special multicontact structure}} on $\mathcal{M}$ if $\omega\in\df^k(\mathcal{M})$ is a closed form, and we have the following properties:
\begin{enumerate}[{\rm (1)}]
\item\label{prekeromega}
$ \rank\ker\omega=N$.
\item\label{prerankReeb}
$ \rank{\cal D}^{\mathfrak{R}}=k$.
\item\label{prerankcar}
$ \rank\left(\ker\omega\cap\ker\Theta\cap\ker\d\Theta\right)=0$.
\item \label{preReebComp}
${\cal A}^{k-1}(\ker\omega)=\{\inn{R}\Theta\mid R\in \mathfrak{R}\}$.
\end{enumerate}
Then, 
the triple $(\mathcal{M},\Theta,\omega)$ is said to be a {\sl\textbf{special multicontact manifold}},
$(\Theta,\omega)$ is a {\sl\textbf{special multicontact structure}},
and $\Theta\in\df^k(\mathcal{M})$ is a {\sl\textbf{special multicontact form}} on $\mathcal{M}$.
\end{definition}

The following proposition presents an essential characteristic of special multicontact structures.

\begin{proposition}
\label{sigma}
Let $(\mathcal{M},\Theta,\omega)$ be a multicontact manifold, then
there exists a unique $1$-form
$\sigma_{\Theta}\in\df^1(\mathcal{M})$, called the {\sl\textbf{dissipation form}}, satisfying
\[
\sigma_{\Theta}\wedge\inn{R}\Theta=\inn{R}\dd\Theta \,,\qquad \text{for every }\,R\in\mathfrak{R} \,.
\]
\end{proposition}

And, using this form, we introduce:

\begin{definition}\label{def:bard}
\label{bard}
Let $\sigma_{\Theta}\in\df^1(\mathcal{M})$ be the dissipation form. 
We define the operator
\begin{align*}
\bd:\df^m(\mathcal{M})&\longrightarrow\df^{m+1}(\mathcal{M})
\\
\beta&\longmapsto\bd\beta=\dd \beta+\sigma_{\Theta}\wedge\beta\,.
\end{align*}
\end{definition}

The multicontact structures corresponding to action-dependent field theories arising from the Herglotz variational principle (see \cite{GLMR_24}) satisfy the following additional requirement.

\begin{definition}\label{def:variational}
Let \((\mathcal{M},\Theta,\omega)\) be a multicontact manifold satisfying
\begin{equation}
\iota_X\iota_{Y}\Theta = 0 \,, \quad
\text{for every } X,Y\in\Gamma(\ker\omega) \,.
\end{equation}
Then $(\mathcal{M},\Theta,\omega)$ is a {\sl\textbf{variational multicontact manifold}} and $(\Theta,\omega)$ is said to be a {\sl\textbf{variational multicontact structure}}.
\end{definition}

The next theorem states the existence of canonical coordinates for these last kinds of multicontact structures:

\begin{theorem} 
\label{cor:coordenadas}
Given a special multicontact manifold $(\mathcal{M},\Theta,\omega)$;
around every point ${\rm p}\in\mathcal{M}$,  there exists a local chart of {\rm adapted coordinates} 
$(U;x^\mu,u^I,s^\mu,)$ ($1\leq\mu\leq k$, $1\leq I\leq N-k$) such that
$$
\ker\omega=\Big<\frac{\partial}{\partial u^I},\frac{\partial}{\partial s^\mu}\Big> \, , \qquad
{\cal D}^\mathfrak{R}=\Big<\frac{\partial}{\partial s^\mu}\Big> \, ;
$$
and the coordinates $(x^\mu)$ can be chosen such that 
$$\omega\vert_U=\dd x^1\wedge\dots\wedge\dd x^k\equiv\dd^k x\ , $$
(and so we shall do henceforth).
In addition, if $(P,\Theta,\omega)$is a variational multicontact manifold, then the local expression of the multicontact form is
\begin{equation}
\label{thetacoor}
\Theta\vert_U=H(x^\nu,u^J,s^\nu)\,\d^kx+f_I^\mu(x^\nu,u^J)\,\d u^I\wedge\d^{k-1}x_\mu+
\d s^\mu\wedge\d^{k-1}x_\mu\,;
\end{equation}
where $\displaystyle\d^{k-1}x_\mu=\inn{\frac{\partial}{\partial x^\mu}}\d^kx$.
Furthermore, in these coordinates,
\begin{equation}
\label{sigmacoord}
\sigma_{\Theta}\vert_U=\parder{H}{s^\mu}\,\d x^\mu\,,
\end{equation}
\end{theorem}

Moreover, we have the following local characterization of the Reeb vector fields:

\begin{proposition}
\label{reeblemma}
If $(\mathcal{M},\Theta,\omega)$ is a special multicontact manifold, in the above coordinate chart,
there exists a unique local basis
$\{ R_\mu\}$ of $\mathfrak{R}$ such that,
\begin{equation}
\label{reebcalcul}
(\inn{R_\mu}\Theta)\vert_U=\dd^{k-1}x_\mu\, .
\end{equation}
Moreover $[R_\mu,R_\nu]=0$.
\end{proposition}

Finally, to establish the relation between multicontact and special multicontact structures, we need following result:

\begin{lemma}\label{lem:multicoMit} 
Let $(\mathcal{M},\Theta,\omega)$ be a special multicontact manifold. Then,
\begin{enumerate}
    \item $\ker \Theta\subset\ker\omega$\,.
    \item Moreover, if $(\mathcal{M},\Theta,\omega)$ is a variational
    multicontact manifold (see Definition \ref{def:variational}), then
    \begin{align}
        \ker \Theta\oplus D^{\mathfrak{R}}=\ker \omega\,.
    \end{align}
\end{enumerate}
\end{lemma}
\begin{proof}
    First, notice that condition $4$ of a special multicontact structure in Definition \ref{multicontactdef} implies that, for every $(m-1)$-form $\alpha$ that vanishes by the contraction of any element of $\ker \omega$, there exists a Reeb vector field $R\in\Gamma(\mathcal{D}^\mathfrak{R})$ be such that $\inn{R}\Theta=\alpha$.
    \begin{enumerate}
        \item Assume that at a point $p\in\mathcal{M}$ there exists an element $Y_p\in \T_p\mathcal{M}$ such that $Y_p\in \ker\Theta_p$ but $Y_p\notin \ker\omega_p$. Then, there exists a $(m-1)$ form $\alpha$ at $p$ that vanishes by the contraction of any element of $\ker \omega_p$ such that $\inn{Y_p}\alpha\neq 0$. Let $R\in \mathcal{D^\mathfrak{R}}$ such that $\inn{R}\Theta=\alpha$, at $p$. Then,
        \begin{align}
            0=\inn{Y_p}\inn{R_p}\Theta_p=\inn{Y_p}\alpha_p\,,
        \end{align}
        which is a contradiction.
        \item Due to the previous item and the definition of the Reeb distribution, clearly 
        \begin{align}
        \ker \Theta+D^{\mathfrak{R}}\subset\ker \omega\,.
        \end{align}
        Fix a point $p\in\mathcal{M}$ and let $Y_p\in\ker\omega_p$. Then $\inn{Y_p}\Theta_p$ vanishes by the action of any element of $\ker \omega_p$, because it is variational. Then, there exists $R\in \mathcal{D^\mathfrak{R}}$ such that $\inn{R_p}\Theta_p=\inn{Y_p}\Theta_p$.
        In particular, $\inn{(Y_p-R_p)}\Theta_p=0$. Therefore, we can decompose $Y_p=(Y_p-R_p)+R_p$ as a sum of an element of $\ker\Theta_p$ and an element of $D^{\mathfrak{R}}_p$.
        Finally, $\ker \Theta_p\cap D^{\mathfrak{R}}_p=\{0\}$ due to Lemma $3.5$ in \cite{LGMRR_23}.
    \end{enumerate}
\end{proof}

Then, taking this lemma into account, it is clear that:

\begin{proposition}
    If $(\mathcal{M},\Theta,\omega)$ is a special multicontact manifold, then $(\mathcal{M},\Theta)$ is a multicontact manifold.
\end{proposition}

\subsection{Multicontact Lagrangian formalism}

Next, we describe the Lagrangian action-dependent field theories using multicontact structures.

\subsubsection{Geometry of the phase bundle}

Consider a bundle $\pi\colon E\to M$,
where $M$ is an orientable $k$-dimensional manifold with volume form $\omega_M\in\df^m(M)$, and let $J^1\pi\to E\to M$ be the corresponding first-order jet bundle.
If $\dim M=k$ and $\dim E=n+k$, then $\dim J^1\pi=nk+n+k$.
Natural coordinates in $J^1\pi$ adapted to the bundle structure
are $(x^\mu,y^a,y^a_\mu)$ ($\mu = 1,\ldots,k$; $a=1,\ldots,n$),
and are such that
$\omega_M=\d x^1\wedge\cdots\wedge\d x^k=:\d^kx$.

In the multicontact Lagrangian formalism for action-dependent field theories,
the {\sl configuration bundle} of the theory is $E\times_M\bigwedge^{k-1}(\Tan^*M)\to M$,
where $\bigwedge^{k-1}(\Tan^*M)$ denotes the bundle of $(k-1)$-forms on $M$.
The corresponding {\sl phase bundle} is
${\cal P}=J^1\pi\times_M\bigwedge^{k-1}(\Tan^*M)$.
Natural coordinates in ${\cal P}$ are $(x^\mu,y^a,y^a_\mu,s^\mu)$,
and $\dim\mathcal{P}=2k+n+nk$.
Moreover, we also have the natural projections depicted in the following diagram:
\begin{equation}\label{diagrammulticoLag}
\xymatrix{
&\ &  \  &{\cal P}=J^1\pi\times_M\bigwedge^{k-1}(\Tan^*M)  \ar[rrd]_{\tau_1}\ar[lld]^{\rho}\ar[ddd]_{\tau}\ &  \   &
\\
&J^1\pi \ar[ddrr]^{\bar{\pi}^1}\ar[d]_{\pi^1}\ & \ & \ & \ &\ar[ddll]_{\tau_o}\bigwedge^{k-1}(\Tan^*M) 
\\
&E\ar[drr]_{\pi}\ & \ & \ & \ &  
 \\
&\ & \ &M \ & \ & 
}\nonumber
\end{equation}

As $\Lambda^{k-1}(\Tan^*M)$ is a bundle of forms over $M$, it is endowed with a canonical structure, its ``tautological form''
$\theta\in \df^{k-1}(\Lambda^{k-1}(\Tan^*M))$,
which is defined as usual,
and whose local expression, in natural coordinates, is $\theta=s^\mu\,\d^{k-1}x_\mu$. Then:

\begin{definition}
The form $\overline{S}:=\tau_1^*\theta\in\df^{k-1}({\cal P})$
is called the {\sl\textbf{canonical action form}} of ${\cal P}$.
\end{definition}

Its expression in coordinates is also $\overline{S}=s^\mu\,\d^{k-1} x_\mu$.

\begin{definition}
A section $\bm{\psi}\colon M\rightarrow {\cal P}$ of the projection $\tau:\mathcal{P}\rightarrow M$ is said to be a {\sl\textbf{holonomic section}} in ${\cal P}$ if
the section $\psi:=\rho\circ\bm{\psi}\colon M\to J^1\pi$
is holonomic in $J^1\pi$; that is,
there is a section $\phi\colon M\to E$ of $\pi$ such that $\psi=j^1\phi$.
It is customary to write $\bm{\psi}=(\psi,s)=(j^1\phi,s)$, 
where $s\colon M\to\bigwedge^{k-1}(\Tan^*M)$  is a section of the projection $\tau_\circ:\bigwedge^{k-1}(\Tan^*M)\to M$;
then, we also say that $\bm{\psi}$ is the
{\sl\textbf{canonical prolongation}} of the section 
$\bm\phi:=(\phi,s)\colon M\to E\times_M\bigwedge^{k-1}(\Tan^*M)$ to ${\cal P}$.
\end{definition}

Now, consider the set $\vf^k({\cal P})$, of multivector fields in ${\cal P}$ (see the Appendix \ref{append} for details).

\begin{definition}
A $k$-multivector field
$\bm{\Gamma}\in\vf^k({\cal P})$ is a {\sl\textbf{holonomic $k$-multivector field}}
or a {\sl\textbf{second-order partial differential equation}} (\textsc{sopde}) in ${\cal P}$ if
it is $\tau$-transverse, integrable, and its integral sections are holonomic on ${\cal P}$.
\end{definition}

The local expression of a {\sc sopde} in ${\cal P}$, satisfying the transversality condition $\inn{\bfX}\omega=1$, is 
\begin{equation}
\label{localsode2}
\bfX = \bigwedge^k_{\mu=1}
\Big(\parder{}{x^\mu}+y^a_\mu\frac{\displaystyle\partial} {\displaystyle
\partial y^a}+F_{\mu\nu}^a\frac{\displaystyle\partial}{\displaystyle \partial y^a_\nu}+g^\nu_\mu\,\frac{\partial}{\partial s^\nu}\Big)\,,
\end{equation}
and its integral sections are solutions to the system of second-order partial differential equations:
\begin{equation} y^a_\mu=\parder{y^a}{x^\mu}\,,\qquad F^a_{\mu\nu}=\frac{\partial^2y^a}{\partial x^\mu \partial x^\nu}\,. \end{equation}

\begin{remark}
The first-order jet bundle $J^1\pi$ is endowed with a canonical structure
which is called the {\sl canonical endomorphism},
and is a $(1,2)$-tensor field in $J^1\pi$,
denoted ${\rm J}$. Its local expression in natural coordinates of $J^1\pi$ is
$${\rm J}=\left(\d y^a-y^a_\mu\d x^\mu\right)\otimes
\parder{}{y^a_\nu}\otimes\parder{}{x^\nu}\ ,
$$
(see \cite{EMR_96,Sau_89}).
As ${\cal P}=J^1\pi\times_M\Lambda^{k-1}(\Tan^*M)$
is a trivial bundle, this canonical structure 
can be extended to ${\cal P}$ in a natural way. This extension is denoted with the same notation ${\rm J}$, and has the same coordinate expression.
\end{remark}
Then, a direct calculation in coordinates leads to the following characterization of {\sc sopde} multivector fields:

\begin{proposition}
An integrable $k$-multivector field $\bfX\in\vf^m(\mathcal{P})$ is a {\sc sopde}
if, and only if,
\begin{equation}\label{Jsopde}
\inn{\bfX}{\rm J}=0\ ,
\end{equation}
where $\inn{\bfX}{\rm J}$ denotes the natural inner contraction between tensor fields.
\end{proposition}

\begin{remark}
The $\tau$-transverse decomposable $k$-multivector fields satisfying condition \eqref{Jsopde}, whose local expression is \eqref{localsode2},
are usually referred to as {\sl\textbf{semi-holonomic multivector fields}}.
\end{remark}

\subsubsection{Multicontact Lagrangian systems}

Physical information in field theories is introduced by means of {\sl Lagrangian densities}.
A {\sl\textbf{Lagrangian density}} is a $k$-form $\mathcal{L}\in\df^k({\cal P})$;
hence $\mathcal{L}=L\,\d^kx$, where $L\in\Cinfty({\cal P})$ is the
{\sl\textbf{Lagrangian function}} and $\d^kx$ is also the local expression of the form $\omega:=\tau^*\eta$.

\begin{definition}
\label{lagrangean2}
The {\sl\textbf{Lagrangian form}} associated with $\L$ is the form
\[
\Theta_{\L}=-\inn{\rm J}\d\mathcal{L}-\mathcal{L}+\d \overline{S}\in\df^k({\cal P}) \,.
\]
\end{definition}

In natural coordinates, the coordinate expression of this form reads,
\begin{equation}
\label{thetacoor1}
    \Theta_{\mathcal{L}} =-\frac{\partial L}{\partial y^a_\mu}\d y^a\wedge\d^{k-1}x_\mu +\left(\frac{\partial L}{\partial y^a_\mu}y^a_\mu-L\right)\d^k x+\d s^\mu\wedge \d^{k-1}x_\mu\,,
\end{equation}
where the local function
$\displaystyle E_\L:=\frac{\partial L}{\partial y^a_\mu}y^a_\mu-L$
is called the {\sl\textbf{Lagrangian energy}} associated with $L$.

Then, the following property holds \cite{LGMRR_23}:

\begin{proposition}
\label{Prop-regLag2}
For a Lagrangian function $L\in\Cinfty({\cal P})$, the Lagrangian form $\Theta_\L$ is a special (variational) multicontact form in ${\cal P}$
(and hence $(\Theta_\L,\omega)$ is a special (variational) multicontact structure)
if, and only if, the Hessian matrix
$\displaystyle (W_{ij}^{\mu\nu})= 
\bigg(\frac{\partial^2L}{\partial y^a_\mu\partial y^b_\nu}\bigg)$
is regular everywhere.
\end{proposition}

Thus, we define:

\begin{definition}
A Lagrangian function $L\in\Cinfty({\cal P})$ is said to be {\sl\textbf{regular}} if the equivalent
conditions in Proposition \ref{Prop-regLag2} hold.
Otherwise, $L$ is a {\sl\textbf{singular}} Lagrangian.
\end{definition}

\begin{definition}
If $L\in\Cinfty({\cal P})$ is a regular Lagrangian function,
the triad $({\cal P},\Theta_\L,\omega)$ is called a {\sl\textbf{multicontact Lagrangian system}}.
\end{definition}

For a multicontact Lagrangian system $({\cal P},\Theta_\L,\omega)$,
as $L$ is regular, there exists the inverse 
$(W^{ab}_{\mu\nu})$ of the Hessian matrix,
namely $\displaystyle W^{ab}_{\mu\nu}\frac{\partial^2L}{\partial y^b_\nu \partial y^c_\gamma}=\delta^a_c\delta^\gamma_\mu$.
Then, from Lemma \ref{reeblemma}, a simple calculation in coordinates leads to
the following expression for the Reeb vector fields 
$(R_\L)_\mu\in\mathfrak{R}_\L$,
$$
(R_\L)_\mu=\frac{\partial}{\partial s^\mu}-W^{ba}_{\gamma\nu}\frac{\partial^2\L}{\partial s^\mu\partial y^b_\gamma}\,\frac{\partial}{\partial y^a_\nu} \,.
$$
Furthermore, bearing in mind Proposition \ref{sigma} and equation \eqref{thetacoor1}, we obtain that
\begin{equation}
\displaystyle \sigma_{\Theta_\L}=-\parder{L}{s^\mu}\,\d x^\mu\,.
\label{sigmaL}
\end{equation}

Finally, we construct the form,
$$
\overline\d\Theta_\mathcal{L} =\d\Theta_\mathcal{L}+\sigma_{\Theta_\mathcal{L}}\wedge\Theta_\mathcal{L}
 = \d\Theta_\mathcal{L}-\parder{L}{s^\nu}\d x^\nu\wedge\Theta_\mathcal{L}\,.
$$

For a multicontact Lagrangian system $({\cal P},\Theta_\mathcal{L},\omega)$
the Lagrangian field equations are derived from the {\sl generalized Herglotz Variational Principle} \cite{GLMR_24},
and are stated alternatively as:

\begin{definition}
\label{mconteqH}
Let $({\cal P},\Theta_\mathcal{L},\omega)$ be a multicontact Lagrangian system.
\begin{enumerate}[{\rm(1)}]
\item  
The {\sl\textbf{multicontact Lagrangian equations for holonomic sections}}
$\bm{\psi}\colon M\to{\cal P}$ are
\begin{equation}
\label{sect1H}
\bm{\psi}^*\Theta_{\cal L}= 0  \,,\qquad
\bm{\psi}^*\inn{X}\bd\Theta_{\cal L}= 0 \,, \qquad \text{\rm for every }\ X\in\vf({\cal P}) \,.
\end{equation}
or, equivalently, for the canonical prolongation $\bm{\psi}^{(k)}$ (see the Appendix \ref{append})
\begin{equation}
\label{sect2H}
\inn{\bm{\psi}^{(k)}}(\Theta_\L\circ\bm{\psi})=0 \,,\qquad
\inn{\bm{\psi}^{(k)}}(\bd\Theta_\L\circ\bm{\psi}) = 0\,.
\end{equation}
\item 
The {\sl\textbf{multicontact Lagrangian equations for 
holonomic multivector fields}} $\bfX_\L \in\vf^k({\cal P})$ are
\begin{equation}\label{vfH}
\inn{\mathbf{X}_\L}\Theta_{\cal L}=0 \,,\qquad \inn{\bfX_\L}\bd\Theta_{\cal L}=0 \,.
\end{equation}
These holonomic multivector fields are called the {\sl\textbf{Euler--Lagrange multivector fields}} associated with $\mathcal{L}$.

Recall that holonomic multivector fields are $\tau$-transverse.
Note also that equations \eqref{vfH} and the $\tau$-transversality condition, $\inn{\bfX_\L}\omega\not=0$, hold for every multivector field of the equivalence class $\{\bfX_\L\}$ (that is, for every $\bfX_\L'=f\bfX_\L$, with $f$ nonvanishing; see the Appendix \ref{append}).
Then, the condition of $\tau$-transversality can be imposed simply by asking $\inn{\bfX_\L}\omega=1$.
\end{enumerate}
\end{definition}

\begin{theorem}
\label{solutions}
Let $({\cal P},\Theta_\mathcal{L},\omega)$ be a multicontact Lagrangian system.
\begin{enumerate}[\rm (1)]
\item
The multicontact Lagrangian field equations
for multivector fields \eqref{vfH} 
have solutions on ${\cal P}$.
(The solutions are not unique if $k>1$).
\item
Every $k$-multivector field $\bfX_\L$ that is
solution to equations \eqref{vfH} is semi-holonomic.
\item If $\bfX_\L$ is a holonomic multivector field (a {\sc sopde}) solution to the Lagrangian equations \eqref{vfH}, then its integral sections are the solutions to the multicontact Euler--Lagrange field equations for holonomic sections \eqref{sect1H} or \eqref{sect2H}. 
\end{enumerate}
\end{theorem}
\begin{proof}
In a natural chart of coordinates of ${\cal P}$, 
a $\tau$-transverse and locally decomposable $k$-multivector field
satisfying $\inn{\bfX_\L}\omega=1$, has the local expression
$$\displaystyle
{\bf X}_\L= \bigwedge_{\mu=1}^k
\bigg(\parder{}{x^\mu}+(X_\L)_\mu^a\frac{\displaystyle\partial}{\displaystyle
\partial y^a}+(X_\L)_{\mu\nu}^a\frac{\displaystyle\partial}{\displaystyle\partial y^a_\nu}+(X_\L)_\mu^\nu\,\frac{\partial}{\partial s^\nu}\bigg) \,;
$$
and, bearing in mind the local expressions \eqref{thetacoor1} and \eqref{sigmaL}, we have that
$$
\bd\Theta_\L=
\d\left(-\frac{\partial L}{\partial y^a_\mu}\d y^a\wedge\d^{k-1}x_\mu +\Big(\frac{\partial L}{\partial y^a_\mu}y^a_\mu-L\Big)\d^k x\right)
-\left(\parder{L}{s^\mu}\frac{\partial L}{\partial y^a_\mu}\d y^a
-\parder{L}{s^\mu}\d s^\mu\right)\wedge\d^kx
\,.
$$
Then, equations \eqref{vfH} lead to
\begin{align}
0 &=
\displaystyle L + 
\frac{\partial L}{\partial y^a_\mu}\Big((X_\L)_\mu^a-y^a_\mu\Big)-(X_\L)_\mu^\mu\,,
\label{A-E-L-eqs4}
\\
0 &=
\displaystyle \Big((X_\L)_\mu^a-y^a_\mu\Big)
\frac{\partial^2L}{\partial y^a_\mu\partial s^\nu} \,,
\label{A-E-L-eqs2}
\\
\displaystyle 0&=\Big((X_\L)_\mu^a-y^a_\mu\Big)
\frac{\partial^2L}{\partial y^a_\mu\partial x^\nu} \,,
\label{A-E-L-eqs0}
\\
0 &=
\displaystyle \Big((X_\L)_\mu^a-y^a_\mu\Big)
\frac{\partial^2L}{\partial y^b_\nu\partial y^a_\mu}
\label{A-E-L-eqs1} \,,
\\
0 &=
\displaystyle
\Big((X_\L)_\mu^a-y^a_\mu\Big)
\frac{\partial^2 L}{\partial y^b\partial y^a_\mu}
+\frac{\partial L}{\partial y^b}- \parderr{L}{x^\mu}{y_\mu^b}
-\frac{\partial^2L}{\partial s^\nu\partial y^b_\mu}(X_\L)_\mu^\nu
\nonumber
\\ &\quad
-\frac{\partial^2L}{\partial y^a \partial y^b_\mu}(X_\L)_\mu^a
-\frac{\partial^2L}{\partial y^a_\nu\partial y^b_\mu}(X_\L)_{\mu\nu}^a
+\frac{\partial L}{\partial s^\mu}
\frac{\partial L}{\partial y^b_\mu}\,,
\label{A-E-L-eqs3}
\end{align}
and a last group of equations which are identities when they are combined with the above ones.
If $\bfX_\L$ is a {\sc sopde}, then it is semi-holonomic and,
\begin{equation}
\label{semihol}
y^a_\mu=(X_\L)_\mu^a \,;
\end{equation}
therefore, \eqref{A-E-L-eqs2}, \eqref{A-E-L-eqs0} and \eqref{A-E-L-eqs1} hold identically, 
and \eqref{A-E-L-eqs4} and \eqref{A-E-L-eqs3} give 
\begin{align}
(X_\L)_\mu^\mu&= L \,, \nonumber
\\
\frac{\partial L}{\partial y^b}- \parderr{L}{x^\mu}{y_\mu^b}
-\frac{\partial^2L}{\partial y^a \partial y^b_\mu}y_\mu^a
-\frac{\partial^2L}{\partial s^\nu\partial y^b_\mu}(X_\L)_\mu^\nu
-\frac{\partial^2L}{\partial y^a_\nu\partial y^b_\mu}(X_\L)_{\mu\nu}^a
&=-\frac{\partial L}{\partial s^\mu}
\frac{\partial L}{\partial y^a_\mu} \,.
\label{ELeqmvf}
\end{align}
Finally, for the holonomic integral sections $\displaystyle\bm{\psi}(x^\nu)=\Big(x^\mu,y^a(x^\nu),\parder{y^a}{x^\mu}(x^\nu),s^\mu(x^\nu)\Big)$ of $\bfX_\L$, the last equations transform into
\begin{align}
 \parder{s^\mu}{x^\mu}&=L\circ{\bm{\psi}} \,, \nonumber
 \\
\label{ELeqs2}
\frac{\partial}{\partial x^\mu}
\left(\frac{\displaystyle\partial L}{\partial
y^b_\mu}\circ{\bm{\psi}}\right)&=
\left(\frac{\partial L}{\partial y^b}+
\displaystyle\frac{\partial L}{\partial s^\mu}\displaystyle\frac{\partial L}{\partial y^b_\mu}\right)\circ{\bm{\psi}} \,,
\end{align}
which are the coordinate expression of the Lagrangian equations
\eqref{sect1H} or \eqref{sect2H} for holonomic sections.
Equations \eqref{ELeqs2}
are called the {\sl\textbf{Herglotz--Euler--Lagrange field equations}}.

If $L$ is a regular Lagrangian, equations \eqref{A-E-L-eqs1} lead to the semi-holonomy condition \eqref{semihol} and,
in addition, equations \eqref{ELeqmvf} have always solution since the Hessian matrix 
$\displaystyle\bigg(\frac{\partial^2L}{\partial y^b_\nu\partial y^a_\mu}\bigg)$ is regular everywhere.
The solution is not unique unless $k=1$.
\end{proof}

Observe that all these equations are the same as those obtained in the $k$-cocontact Lagrangian formulation of non-conservative field theories \cite{Riv_23}
and also match those of the $k$-contact Lagrangian formalism when the Lagrangian function does not depend on the spacetime variables $x^\mu$ 
\cite{GGMRR_20,GGMRR_21}.

\subsection{Multicontact Hamiltonian formalism}

The Hamiltonian formulation of action-dependent first-order field theories is based on the Hamiltonian multisymplectic formalism of standard classical field theories.

\subsubsection{Geometry of the phase bundle}

Consider a bundle $\pi\colon E\to M$,
where $M$ is an orientable $k$-dimensional manifold with volume form $\eta\in\df^m(M)$. Let ${\cal M}\pi\equiv\Lambda_2^k\Tan^*E$ denote the bundle of $k$-forms on
$E$ vanishing by contraction with two $\pi$-vertical vector fields which, in field theories, is called the {\sl extended multimomentum bundle}.
It is endowed with natural coordinates $(x^\nu,y^a,p^\nu_a,p)$
adapted to the bundle structure ${\cal M}\pi\to E\to M$, and such that
$\eta=\d^kx$; so $\dim\, {\cal M}\pi=nk+n+k+1$.
Consider also the quotient manifold
$J^{1*}\pi={\cal M}\pi/\pi^*\Lambda^k\Tan^*M$
($\pi^*\Lambda^k\Tan^*M$ is the bundle of $\pi$-basic $k$-forms on $E$),
which is called the {\sl restricted multimomentum bundle}.
Its natural coordinates are $(x^\mu,y^a,p_a^\mu)$, and so 
$\dim\,J^{1*}\pi=nk+n+k$.

Then, for the Hamiltonian formalism of action-dependent field theories,
in the regular case, consider the bundles 
$$
\widetilde{\cal P}={\cal M}\pi\times_M\Lambda^{k-1}(\Tan^* M)
\ , \qquad
{\cal P}^* =J^{1*}\pi\times_M\Lambda^{k-1}(\Tan^*M) \,,
$$
which have natural coordinates $(x^\mu, y^a,p_a^\mu,p,s^\mu)$ and$(x^\mu,y^a,p_a^\mu,s^\mu)$, respectively. 
We have the natural projections depicted in the following diagram:
$$
\xymatrix{
& &\widetilde{\cal P}={\cal M}\pi\times_M\Lambda^{k-1}(\Tan^* M)  
\ar[rddd]^{\widetilde\tau_1} \ar[llddd]_{\widetilde{\varrho}} \ar[dd]_{\widetilde{\mathfrak{p}}} 
 \ar[lldd]_{\widetilde{\varrho}_1} 
\\ & &
\\
{\cal M}\pi\ 
\ar[d]_{\mathfrak{p}}& &{\cal P}^*=J^1\pi^*\times_M\Lambda^{k-1}(\Tan^* M)  
\ar@/_1pc/[uu]_{{\bf h}}
\ar[rd]_{\overline\tau_1}\ar[lld]^{\varrho}\ar[ddd]_{\overline\tau}
\\
J^1\pi^*
\ar[ddrr]^{\bar{\kappa}^1}\ar[d]_{\kappa^1}& & &\ar[ddl]_{\tau_\circ}\Lambda^{k-1}(\Tan^* M) 
\\
E\ar[drr]_{\pi}
 \\
& \ &M \ & \ & 
}
$$

Since ${\cal M}\pi$ and  $\Lambda^{k-1}(\Tan^* M)$ are bundles of forms, they have canonical structures, their ``tautological forms''
$\Theta\in\df^k({\cal M}\pi)$, 
called the {\sl\textbf{Liouville form}} of ${\cal M}\pi$
(see, for instance, \cite{CCI_91,EMR_2000} for its definition),
and $\theta\in \df^{k-1}(\Lambda^{k-1}(\Tan^* M))$
whose local expressions are
\begin{equation}
\label{canformmulticot}
\Theta=p_a^\mu\d y^a\wedge\d^{k-1}x_\mu+p\,\d^k x
\,,\qquad
\theta=s^\mu\,\d^{k-1}x_\mu \,.
\end{equation}

\begin{definition}
The {\sl\textbf{canonical (special) multicontact form}} of $\widetilde{\cal P}$ is
\begin{equation}
\label{canoncontact}
\widetilde\Theta:= -\widetilde\varrho_1^{\,*}\Theta+\d(\widetilde\tau_1^{\,*}\theta) \,. 
\end{equation}
\end{definition}

In natural coordinates of $\widetilde{\mathcal{P}}$, the expression of this form is
$$
\widetilde\Theta=
-p_a^\mu\d y^a\wedge\d^{k-1}x_\mu-p\,\d^k x+\d s^\mu\wedge \d^{k-1}x_\mu\ . 
$$

\subsubsection{Multicontact Hamiltonian systems}

\begin{definition}
Let ${\bf h}\colon {\cal P}^*\to\widetilde{\cal P}$ be a section of $\widetilde{\mathfrak{p}}$.
It is locally determined by a function $H\in\Cinfty(U)$, $U\subset{\cal P}^*$,
such that ${\bf h}(x^\mu,y^a,p^\mu_a,s^\mu)=(x^\mu,y^a,p^\mu_a,p=-H(x^\nu,y^b,p^\nu_b,s^\nu),s^\mu)$.
The elements ${\bf h}$ and $H$ are called
a {\sl\textbf{Hamiltonian section}} and its associated {\sl\textbf{Hamiltonian function}}.

\noindent Then, the {\sl\textbf{Hamiltonian form}} associated with ${\bf h}$ is defined by
\begin{equation}
\label{thetaHcoor}
\Theta_{\cal H}:=\mathbf{h}^*\tilde{\Theta}=
-(\widetilde\varrho_1\circ{\bf h})^*\Theta+\d(\overline\tau_1^*\theta) \,. 
\end{equation}
It is a variational multicontact form and the triad $({\cal P}^*,\Theta_{\cal H},\omega=(\bar{\tau}\circ\widetilde{\mathfrak{p}})^*\eta)$
is called a {\sl\textbf{multicontact Hamiltonian system}}.
\end{definition}

Obviously, $\Theta_{\cal H}$ is a special (variational) multicontact form.
In natural coordinates of $\mathcal{P}$, the expression of this form is
\begin{equation}
\label{thetacoor2}
\Theta_{\cal H}=
-p_a^\mu\d y^a\wedge\d^{k-1}x_\mu+H\,\d^k x+\d s^\mu\wedge \d^{k-1}x_\mu\ ,
\end{equation}
and the dissipation form is expressed as
\begin{equation}
\label{sigmaH}
\sigma_{\cal H}=\parder{H}{s^\mu}\,\d x^\mu\,.
\end{equation}

\begin{definition}
Given a multicontact Hamiltonian system $({\cal P}^*,\Theta_{\cal H},\omega)$, the field equations
can be stated alternatively as:
\begin{enumerate}[{\rm(1)}]
\item  
The {\sl\textbf{multicontact Hamilton--de Donder--Weyl equations for sections}}
$\bm{\psi}\colon M\to{\cal P}^*$:
\begin{equation}
\label{sect1H0}
\bm{\psi}^*\Theta_{\cal H}= 0  \,, \qquad
\bm{\psi}^*\inn{Y}\bd\Theta_{\cal H}= 0 \,, \qquad \text{for every }\ Y\in\vf({\cal P}^*) \, ,
\end{equation}
or, equivalently,
\begin{equation}
\label{sect2H0}
\inn{\bm{\psi}^{(k)}}(\Theta_{\cal H}\circ\bm{\psi})=0 \,,\qquad
\inn{\bm{\psi}^{(k)}}(\bd\Theta_{\cal H}\circ\bm{\psi}) = 0\,.
\end{equation}
\item 
The {\sl\textbf{multicontact Hamilton--de Donder--Weyl equations for $\overline\tau$-transverse and integrable multivector fields}} ${\bf X}_{\cal H}\in\vf^k({\cal P}^*)$:
\begin{equation}
\label{vfH2}
\inn{{\bf X}_{\cal H}}\Theta_{\cal H}=0 \,, \qquad \inn{\bfX_{\cal H}}\bd\Theta_{\cal H}=0 \,.
\end{equation}
Equations \eqref{vfH2} and the $\overline\tau$-transversality condition hold for every multivector field of the equivalence class $\{\bfX_\mathcal{H}\}$,
and the transversality condition can be imposed by asking $\inn{\bfX_\mathcal{H}}\omega=1$.
\end{enumerate}
\end{definition}

In natural coordinates, for a $\overline\tau$-transverse, locally decomposable multivector field ${\bf X}_{\cal H}\in~\vf^k({\cal P}^*)$, 
\begin{equation}
\label{Hammv}
{\bf X}_{\cal H}=
\bigwedge_{\mu=0}^{k-1}\Big(\parder{}{x^\mu}+ (X_{\cal H})^a_\mu\frac{\partial}{\partial y^a}+
(X_{\cal H})_{\mu a}^\nu\frac{\partial}{\partial p_a^\nu}+(X_{\cal H})_\mu^\nu\parder{}{s^\nu}\Big) \ ;
\end{equation} 
if it is a solution to equations \eqref{vfH2},
bearing in mind the local expression \eqref{thetacoor2},
these field equations lead to
\begin{equation*}
(X_{\cal H})_\mu^\mu = 
p_a^\mu\,\frac{\partial H}{\partial p^\mu_a}-H \,,\qquad
(X_{\cal H})^a_\mu=\frac{\partial H}{\partial p^\mu_a} \,,\qquad
(X_{\cal H})_{\mu a}^\mu= 
-\left(\frac{\partial H}{\partial y^a}+ p_a^\mu\,\frac{\partial H}{\partial s^\mu}\right) \,,   
\end{equation*}
together with a last group of equations which are identities when the above ones are taken into account.
Then, the integral sections $\bm{\psi}(x^\nu)=(x^\mu,y^a(x^\nu),p^\mu_a(x^\nu),s^\mu(x^\nu))$
of the integrable solutions ${\bf X}_{\cal H}$ of \eqref{vfH2} 
are the solutions to the equations \eqref{sect1H0} or \eqref{sect2H0} which read as
\begin{equation*}
\frac{\partial s^\mu}{\partial x^\mu} = \left(p_a^\mu\,\frac{\partial H}{\partial p^\mu_a}-H\right)\circ\bm{\psi}\,, \qquad 
\frac{\partial y^a}{\partial x^\mu}= \frac{\partial H}{\partial p^\mu_a}\circ\bm{\psi} \,,\qquad
\frac{\partial p^\mu_a}{\partial x^\mu} = 
-\left(\frac{\partial H}{\partial y^a}+ p_a^\mu\,\frac{\partial H}{\partial s^\mu}\right)\circ\bm{\psi} \,,
\end{equation*}
and are called the {\sl\textbf{Herglotz--Hamilton--de Donder--Weyl equations}}
for action-dependent field theories.
These equations are compatible in ${\cal P}^*$.

Observe that these equations are the same as those obtained in the $k$-cocontact Hamiltonian formulation of non-conservative field theories
and also lead to those of the $k$-contact Hamiltonian formalism when the Hamiltonian function does not depend on the spacetime variables $x^\mu$ 
\cite{GGMRR_20,GGMRR_21}.

\begin{remark}[{\sl The canonical multicontact Hamiltonian system associated with a multicontact Lagrangian system\/}]

Let $\L\in\df^k({\cal P})$ be a Lagrangian density with $\L=L\,\omega$. 

First, denote $FL_s\colon J^1\pi\to J^{1*}\pi$ the Legendre map associated with the restriction of the Lagrangian function $L\in\Cinfty({\cal P})$
to the fibers of the projection $\tau_1$ (recall diagram  \ref{diagrammulticoLag}). Informally, it is obtained considering $L$ with $s^\mu$ ``freezed'', which is denoted $L_s\in\Cinfty(J^1\pi)$.
Then, the {\sl\textbf{restricted Legendre map}} associated with the Lagrangian function $L\in\Cinfty({\cal P})$
is the map
${\cal FL}\colon {\cal P}\to {\cal P}^*$
given by ${\cal FL}:=(FL_s,{\rm Id}_{\Lambda^{k-1}(\Tan^*M)})$.
It is locally given by
$$\displaystyle
{\cal FL}(x^\mu,y^a,y^a_\mu,s^\mu)=\Big(x^\mu,y^a,\frac{\partial L}{\partial y^a_\mu},s^\mu\Big) \ .$$
Similarly, the {\sl\textbf{extended Legendre map}} associated with $L$ is the map 
$\widetilde{\mathcal{FL}}\colon{\cal P}\to\widetilde{\cal P}$ given by $\widetilde{\mathcal{FL}}:=(\widetilde{FL}_s,{\rm Id}_{\Lambda^{k-1}(\Tan^*M)})$.
Its local expression is
$$\widetilde{\mathcal{FL}}(x^\mu,y^a,y^a_\mu,s^\mu)=\Big(x^\mu,y^a,\frac{\partial L}{\partial y^a_\mu},L-y^a_\mu\parder{L}{y^a_\mu},s^\mu\Big)\ . $$
It is not difficult to prove that $L$ is a regular Lagrangian function if, and only if,
the Legendre map ${\cal FL}$ is a local diffeomorphism.
In particular, 
$L$ is said to be {\sl\textbf{hyperregular}}
if ${\cal FL}$ is a global diffeomorphism.

Therefore, if $L$ is a hyperregular Lagrangian function, we can define the Hamiltonian section ${\bf h}:=\widetilde{\mathcal{FL}}\circ\mathcal{FL}^{-1}$,
and construct the corresponding multicontact Hamiltonian system
$({\cal P}^*,\Theta_{\cal H},\omega)$,
which is called the {\sl\textbf{canonical multicontact Hamiltonian system}} associated with the multicontact Lagrangian system
$({\cal P},\Theta_\mathcal{L},\omega)$.
If $\L$ is regular, this construction is local.
Bearing in mind the coordinate expressions \eqref{thetacoor1}, \eqref{sigmaL}, \eqref{thetacoor2}, and \eqref{sigmaH},
we obtain that $\mathcal{FL}^*\Theta_\mathcal{H}=\Theta_\L$ and
$\mathcal{FL}^*\bd\Theta_\mathcal{H}=\bd\Theta_\L$.
\end{remark}

\section{Relationship between multicontact, \texorpdfstring{$k$}{}-cocontact, and \texorpdfstring{$k$}{}-contact structures}
\label{relation}

For the relation among multisymplectic, $k$-cosymplectic and $k$-symplectic structures in classical field theories, see \cite{LMNRS-2002,RRSV_11} (see also \cite{LMORS_98, LMOS_97})

The relation among the multicontact, the $k$-cocontact, and the $k$-contact structures
for action-dependent field theories is done for the particular situation when
$\pi\colon E\to M$ is the trivial bundle
$\Rk\times Q\rightarrow \Rk$.

\subsection{The Hamiltonian case}

A previous result needed to establish this relation is as follows:

\begin{proposition}
\label{lemaaux}
\begin{enumerate}
\item
The manifold ${\mathcal M}\pi\equiv\Lambda_2^k\T^*(\Rk\times Q)$
 is diffeomorphic to $  \Rk\times\R\times\oplus^k\Tan^*Q$.
\item
As a consequence,
$J^1\pi^*$ is diffeomorphic to $  \Rk\times\oplus^k\Tan^*Q$.
\end{enumerate}
\end{proposition}
\begin{proof}
\begin{enumerate}
\item
For $t\in\mathbb{R}^k$, let $i_t\colon Q\hookrightarrow
  \Rk\times Q$ be the canonical embedding given by $i_t(q)=(t,q)$, 
and 
$\rho_{_Q}\colon \Rk\times Q\rightarrow Q$ the canonical submersion.
Then, we can define the map
$$
\begin{array}{cccc}
\bar\Psi\colon& \Lambda_2^k\T^*(  \Rk\times Q) & \longrightarrow &   \Rk\times\R\times \oplus^k\Tan^*Q \\
 & \xi_{(t,q)} & \longmapsto & (t,p,\xi^1_q, \dots ,\xi^k_q)
\end{array},
$$
where, for $X\in\vf(Q)$,
\begin{align*}
\xi^\alpha_q(X) &=
\xi_{(t,q)}\left(\parder{}{x^1}\Big\vert_{(t,q)},\dots,
\parder{}{x^{\alpha-1}}\Big\vert_{(t,q)},(i_t)_*X_q,
 \parder{}{x^{\alpha+1}}\Big\vert_{(t,q)},\ldots,
\parder{}{x^k}\Big\vert_{(t,q)}\right)\, , \\
p&= \xi_{(t,q)}\left(\parder{}{x^1}\Big\vert_{(t,q)},\dots,\parder{}{x^k}
\Big\vert_{(t,q)}\right) \ ,
\end{align*}
(observe that $x^\alpha$ and $p$ are now global coordinates in the corresponding fibres and, then, 
the global coordinate $p$
can be identified with the natural projection
 $p\colon \Rk\times\R\times \oplus^k\Tan^*Q\to\R$). The inverse of $\bar\Psi$, at a point $(t,p,\xi^\alpha_q)\in \Rk\times\R\times\oplus^k\Tan^*Q$ is given by
$$
\xi_{(t,q)}= p \, \d^k
x\vert_{(t,q)}+(\rho_{_Q})_{(t,q)}^*\xi^\alpha_q\wedge\d^{k-1}x_\alpha\vert_{(t,q)}
\ ,
$$
Thus, $\bar\Psi$ is a diffeomorphism (locally $\bar\Psi$ is written as the identity).
\item
Bearing in mind the identification of $\Lambda^k(\Tan^*M)$ with $\R$,
this is a straightforward consequence of the above item since
$$
J^1\pi^*=\mathcal{M}\pi/\pi^*\Lambda^k\T^*M\simeq
(\Rk\times\R\times\oplus^k\Tan^*Q)/\R\simeq\Rk\times\oplus^k\Tan^*Q \,.
$$
\end{enumerate}
\end{proof}

Notice that $\Lambda^{k-1}(\Tan^*M)$ can be identified with $\R^k$, and therefore,
taking into account the above proposition, we can write
\begin{align}
    \widetilde{\cal P}&={\cal M}\pi\times_M\Lambda^{k-1}(\Tan^* M)\simeq\Rk\times \R\times\oplus^k\Tan^*Q\times\Rk\,,
    \\ 
    {\cal P}^*&=J^1\pi^*\times_M\Lambda^{k-1}(\Tan^* M)  \simeq\Rk\times \oplus^k\Tan^*Q\times\Rk  =\mathbf{P}^*\,.
\end{align}
The following diagram contains the projections and embeddings that we will use next.
\begin{align}
\xymatrix{
& \widetilde{\cal P}\simeq\Rk\times \R\times\oplus^k\Tan^*Q\times\Rk 
\ar@/^1pc/[dd]^{{\mu}} \ar@/_2pc/[dl]_{\widetilde{\nu}} 
\\
{P}^*=\oplus^k\Tan^*Q\times\Rk  \ar@{^{(}->}[dr]^\jmath \ar@{^{(}->}[ur]_{\widetilde\jmath} & 
\\
& {\cal P}^* \simeq\Rk\times \oplus^k\Tan^*Q\times\Rk  =\mathbf{P}^*
\ar@{^{(}->}[uu]^{\widetilde\iota}
\ar@/^2pc/[lu]^{\nu} 
}
\end{align}

The embeddings $\widetilde\iota,\jmath$ and $\widetilde\jmath$ are given by the zero-sections:
\begin{align}
\widetilde\iota(x^\alpha,y^a,p^\alpha_a,s^\alpha)&=(x^\alpha,p=0,y^a,p^\alpha_a,s^\alpha)
\\
 \jmath\,(y^a,p^\alpha_a,s^\alpha)&=(x^\alpha=0,y^a,p^\alpha_a,s^\alpha)
\\
    \widetilde\jmath\,(y^a,p^\alpha_a,s^\alpha)&=(x^\alpha=0,p=0,y^a,p^\alpha_a,s^\alpha)
\end{align}

\subsubsection{Relationship between the canonical structures of
\texorpdfstring{$\widetilde{\mathcal{P}}={\cal M}\pi\times_M\Lambda^{k-1}(\Tan^*M)$}{} and \texorpdfstring{$P^*=\oplus^k\Tan^*Q\times\R^k$}{}}

If the manifold $\widetilde{\cal P}={\cal M}\pi\times_M\Lambda^{k-1}(\Tan^*M)$ is diffeomorphic to $\R^k\times \R\times\oplus^k\Tan^*Q\times\Rk$, then it is a trivial bundle over $\Rk$. Therefore, the canonical vector fields $\displaystyle\parder{}{x^\alpha}\in\vf(\Rk)$
can be extended to vector fields in $\widetilde{\mathcal{P}}$,
which have the same coordinate expressions.

Then, following the same pattern as in the proof of the item 1 of Proposition \ref{lemaaux} and
starting from the canonical special multicontact form $\widetilde\Theta\in\df^k(\widetilde{\mathcal{P}})$
given in \eqref{canoncontact},
we can define the forms $\eta^\alpha\in\df^1 (\oplus^k\Tan^*Q\times\Rk)$ by
\begin{align}
\eta^\alpha&=(-1)^{\alpha-1}\widetilde\jmath^{\,*}\Big(\inn{\big(\parder{}{x^k}\big)}\ldots\inn{\big(\parder{}{x^{\alpha+1}}\big)}\inn{\big(\parder{}{x^{\alpha-1}}\big)}\ldots
 \inn{\big(\parder{}{x^1}\big)}\widetilde\Theta\Big)
 \\
 &=
-\widetilde\jmath^*\Big(\inn{\big(\parder{}{x^k}\big)}\ldots
 \inn{\big(\parder{}{x^1}\big)}(\widetilde\Theta\wedge\d x^\alpha)\Big) \ .
 \label{relcan11}
\end{align}
In coordinates, we have,
\begin{align}
\widetilde\Theta\wedge\d x^\alpha=(-1)^{k-1}\big(-p_a^\alpha\d y^a\wedge\d^{k}x+d s^\alpha\wedge\d^{k}x\big)\,.
\end{align}
Then
\begin{align*}
-\widetilde\jmath^*\Big(\inn{\big(\parder{}{x^k}\big)}\ldots
 \inn{\big(\parder{}{x^1}\big)}(\widetilde\Theta\wedge\d x^\alpha)\Big)&=(-1)^k\widetilde\jmath^*\Big(\inn{\big(\parder{}{x^k}\big)}\ldots
 \inn{\big(\parder{}{x^1}\big)} \big(-p_a^\alpha\d y^a\wedge\d^{k}x+d s^\alpha\wedge\d^{k}x\big)\Big) 
 \\
 &=-p_a^\alpha\d y^a+\d s^\alpha \ .
\end{align*}
Therefore, $\{\eta^\alpha\}$
is the canonical $k$-contact structure in $P^*=\oplus^k\Tan^*Q\times\Rk$.

Conversely, starting from the $k$-contact structure $\{\eta^\alpha\}$ in $P^*=\oplus^k\Tan^*Q\times\Rk$
we can obtain the canonical special multicontact form in $\widetilde{\mathcal{P}}\simeq(\R^k\times\R\times \oplus^k\Tan^*Q)\times\Rk$,
by doing
\begin{equation}
 \label{relcan12}
\widetilde\Theta= 
p\,\omega+\widetilde\nu^*\eta^\alpha\wedge\inn{\big(\parder{}{x^\alpha}\big)}\omega=
p \d^k x+\widetilde\nu^*\eta^\alpha\wedge\d^{k-1}x_\alpha \,.
\end{equation}
The Reeb vector fields are the same for both structures and are $\displaystyle\Big\{\parder{}{s^\alpha}\Big\}$.

In summary, we have proved the following.

\begin{theorem}
The canonical special multicontact form on
$\widetilde{\mathcal{P}}\simeq\R^k\times\R\times \oplus^k\Tan^*Q\times\Rk$ and
the contact forms of the canonical $k$-contact structure on $P^*=\oplus^k\Tan^*Q\times\Rk$
are related by \eqref{relcan11} and \eqref{relcan12}.
\end{theorem}

\subsubsection{Relationship between the structures of
\texorpdfstring{$\mathcal{P}^*=J^{1*}\pi\times_M\Lambda^{k-1}(\Tan^*M)$}{} and \texorpdfstring{$P^*=\oplus^k\Tan^*Q\times\Rk$}{}}
\label{eqsect1}

If $\mu\colon\widetilde{\mathcal{P}}\simeq\Rk\times \R\times\oplus^k\Tan^*Q\times\Rk\to\mathcal{P}^*\simeq\Rk\times\oplus^k\Tan^*Q\times\Rk$
is a trivial bundle, we can take a global Hamiltonian section
${\bf h}\colon\Rk\times\oplus^k\Tan^*Q\times\Rk\to\Rk\times \R\times\oplus^k\Tan^*Q\times\Rk$,
specified by a (global) Hamiltonian function $H\in\Cinfty(\mathcal{P}^*)$,
and then define the (non-canonical) special multicontact form $\Theta_{\mathcal{H}}\in\df^k(\mathcal{P}^*)$ given in \eqref{thetacoor2}.
Therefore, following the same pattern as in the above section,
we can obtain the forms $\eta^\alpha\in\df^1 (\oplus^k\Tan^*Q\times\Rk)$ given as follows,
\begin{align}
\label{relcan21}
\eta^\alpha&=(-1)^{\alpha-1}\jmath^*\Big(\inn{\big(\parder{}{x^k}\big)}\ldots\inn{\big(\parder{}{x^{\alpha+1}}\big)}\inn{\big(\parder{}{x^{\alpha-1}}\big)}\ldots
 \inn{\big(\parder{}{x^1}\big)}\Theta_{\mathcal{H}}\Big)
 \\&=
 -\jmath^*\Big(\inn{\big(\parder{}{x^k}\big)}\ldots
 \inn{\big(\parder{}{x^1}\big)}(\Theta_{\mathcal{H}}\wedge\d x^\alpha)\Big)\ ,
\end{align}
that define the canonical $k$-contact structure in $P^*=\oplus^k\Tan^*Q\times\Rk$.

Conversely, starting from the $k$-contact structure $\{\eta^\alpha\}$ in $P^*=\oplus^k\Tan^*Q\times\Rk$
we can obtain the canonical special multicontact form in $\widetilde{\mathcal{P}}\simeq(\R^k\times\R\times \oplus^k\Tan^*Q)\times\Rk$,
by doing
\begin{equation}
 \label{relcan22}
\Theta_{\mathcal{P}}= 
-H\,\omega+\nu^*\eta^\alpha\wedge\inn{\big(\parder{}{x^\alpha}\big)}\omega=
-H\,\d^k x+\nu^*\eta^\alpha\wedge\d^{k-1}x_\alpha \,.
\end{equation}
The Reeb vector fields are the same for both structures and are $\displaystyle\Big\{\parder{}{s^\alpha}\Big\}$.

Thus, we have proved the following result.

\begin{theorem}
The special multicontact form on
$\mathcal{P}^*\simeq\R^k\times \oplus^k\Tan^*Q\times\Rk$ and
the contact forms of the canonical $k$-contact structure on $P^*=\oplus^k\Tan^*Q\times\Rk$
are related by \eqref{relcan21} and \eqref{relcan22}.
\end{theorem}

\subsubsection{Relationship between the canonical structures of
\texorpdfstring{$\widetilde{\mathcal{P}}={\cal M}\pi\times_M\Lambda^{k-1}(\Tan^*M)$}{} and \texorpdfstring{$\mathbf{P}^*=\Rk\times\oplus^k\Tan^*Q\times\R^k$}{}}

From the canonical special multicontact form $\widetilde\Theta\in\df^1(\widetilde{\mathcal{P}})$ we define the forms $\eta^\alpha\in\df^1(\mathbf{P}^*)$ as in \eqref{relcan11}; that is,
\begin{equation}
\label{relcan31}
\eta^\alpha=(-1)^{\alpha-1}{\widetilde{\imath}}^{\,*}\Big(\inn{\big(\parder{}{x^k}\big)}\ldots\inn{\big(\parder{}{x^{\alpha+1}}\big)}\inn{\big(\parder{}{x^{\alpha-1}}\big)}\ldots
 \inn{\big(\parder{}{x^1}\big)}\widetilde\Theta\Big) \ .
\end{equation}
whose coordinate expressions are $\eta^\alpha=\d s^\alpha-p^\alpha_a\d y^a$.
In addition,
we also define the forms,
\begin{equation}
\label{taus}
\tau^\alpha=(-1)^{k-\alpha}\,{\widetilde{\imath}}^{\,*}\Big(\inn{\big(\parder{}{x^k}\big)}\ldots\inn{\big(\parder{}{x^{\alpha+1}}\big)}\inn{\big(\parder{}{x^{\alpha-1}}\big)}\ldots
 \inn{\big(\parder{}{x^1}\big)}\inn{\big(\parder{}{p}\big)}\d\widetilde\Theta\Big) \ .
\end{equation}
(recall that $p$ denotes the global canonical coordinate in $\R$).
Observe also that, since $p$ denotes the global canonical coordinate in $\R$, the $1$-forms $\tau^\alpha=\d x^\alpha$
are canonically defined on $\Rk\times\oplus^k\Tan^*Q\times\Rk$.

Conversely, from the forms $\{\eta^\alpha\}$ of the $k$-cocontact structure in $\mathbf{P}^*=\Rk\times\oplus^k\Tan^*Q\times\Rk$
we can obtain the canonical special multicontact form in $\widetilde{\mathcal{P}}\simeq(\R^k\times\R\times \oplus^k\Tan^*Q)\times\Rk$
similarly as in \eqref{relcan12}.

In this way, we have proved the following,

\begin{theorem}
The canonical special multicontact form on
$\widetilde{\mathcal{P}}\simeq\R^k\times\R\times \oplus^k\Tan^*Q\times\Rk$ and
the forms of the canonical $k$-cocontact structure on $\mathbf{P}^*=\R^k\oplus^k\Tan^*Q\times\Rk$
are related by \eqref{relcan12}, \eqref{relcan31}, and \eqref{taus}.
\end{theorem}

The (contact) Reeb vector fields are the same for both structures and are $\displaystyle\Big\{\parder{}{s^\alpha}\Big\}$,
and the space-time Reeb vector fields in $\mathbf{P}^*$ are $\displaystyle\Big\{\parder{}{x^\alpha}\Big\}$.

\subsubsection{Relationship between the (non-canonical) structures of
\texorpdfstring{$\mathcal{P}^*=J^{1*}\pi\times_M\Lambda^{k-1}(\Tan^*M)$}{} and \texorpdfstring{$\mathbf{P}^*=\Rk\times\oplus^k\Tan^*Q\times\Rk$}{}}
\label{eqsect2}

First, observe that, in this particular situation,
the manifolds $\mathcal{P}^*$ and $\mathbf{P}^*$ are canonically identified with $\Rk\times\oplus^k\Tan^*Q\times\Rk$.
Then, as in the above two sections,
we obtain that, starting from the (non-canonical) special multicontact form $\Theta_{\mathcal{H}}\in\df^k(\mathcal{P}^*)$ given in \eqref{thetacoor2}, we get the forms
\begin{equation}
\label{relcan41}
\eta^\alpha=-\inn{\big(\parder{}{x^k}\big)}\ldots
\inn{\big(\parder{}{x^1}\big)}(\Theta_{\mathcal{H}}\wedge\d x^\alpha) \ ,
\end{equation}
and the $1$-forms $\tau^\alpha=\d x^\alpha$ are canonically defined.

Conversely, the special multicontact form in $\mathcal{P}^*\simeq\R^k\times\oplus^k\Tan^*Q\times\Rk$ is obtained from the forms $\{\eta^\alpha\}$ of the $k$-cocontact structure in $\mathbf{P}^*=\Rk\times\oplus^k\Tan^*Q\times\Rk$
as in \eqref{relcan12}, without the pullback of $\widetilde\nu$.

So, we have:

\begin{theorem}
The special multicontact form on
$\mathcal{P}^*\simeq\R^k\times\oplus^k\Tan^*Q\times\Rk$ and
the forms of the canonical $k$-cocontact structure on $\mathbf{P}^*=\R^k\oplus^k\Tan^*Q\times\Rk$
are related by \eqref{relcan12} and \eqref{relcan41},
and the $1$-forms $\tau^\alpha$ are canonically defined.
\end{theorem}

\subsection{The Lagrangian case}

As in the Hamiltonian case, now we have the canonical identifications
$$
\mathcal{P}=J\pi\times_M\Lambda^{k-1}(\Tan^*M)\simeq\R^k\times \oplus^k\Tan Q\times\Rk\simeq\mathbf{P}\,,
$$
and the natural embedding
\begin{eqnarray*}
\imath\colon P=\oplus^k\Tan Q\times\Rk\longhookrightarrow\Rk\times\oplus^k\Tan Q\times\Rk \,.
\end{eqnarray*}

Therefore, following the same patterns as in Sections
\ref{eqsect1} and \ref{eqsect2} we obtain:

\begin{theorem}
Let $\L$ be an ``autonomous'' Lagrangian function;
namely, $\L=\L(y^a,y^a_\alpha,s^\alpha)$
and $E_\L$ its associated Lagrangian energy.
The Lagrangian multicontact form on
$\mathcal{P}\simeq\R^k\times \oplus^k\Tan Q\times\Rk$ and
the contact forms of the Lagrangian $k$-contact structure on $P=\oplus^k\Tan Q\times\Rk$
are related as follows,
\begin{align*}
\eta_{\mathcal{L}}^\alpha &= \imath^*\Big(\inn{\big(\parder{}{x^k}\big)}\ldots\inn{\big(\parder{}{x^{\alpha+1}}\big)}\inn{\big(\parder{}{x^{\alpha-1}}\big)}\ldots
 \inn{\big(\parder{}{x^1}\big)}\Theta_{\mathcal{L}}\Big)=
 -\imath^*\Big(\inn{\big(\parder{}{x^k}\big)}\ldots
 \inn{\big(\parder{}{x^1}\big)}(\Theta_{\mathcal{L}}\wedge\d x^\alpha)\Big)\,, \\
\Theta_{\mathcal{L}} &=
-E_{\mathcal{L}}\,\omega+\kappa_2^*\theta^\alpha\wedge\inn{\big(\parder{}{x^\alpha}\big)}\omega=
-E_{\mathcal{L}}\,\d^k t+\kappa_2^*\theta^\alpha\wedge\d^{k-1}x_\alpha \ ,
\end{align*}
where $\kappa_2\colon\R^k\times\oplus^k\Tan Q\times\Rk\to\oplus^k\Tan Q\times\Rk$
is the canonical submersion.
\end{theorem}

\begin{theorem}
Let $\L$ be a ``non-autonomous'' Lagrangian function;
that is, $\L=\L(x^\alpha,y^a,y^a_\alpha,s^\alpha)$
and $E_\L$ its associated Lagrangian energy.
The Lagrangian multicontact form on
$\mathcal{P}\simeq\R^k\times \oplus^k\Tan Q\times\Rk$ and
the contact forms of the Lagrangian $k$-cocontact structure on $\mathbf{P}=\Rk\times\oplus^k\Tan Q\times\Rk$
are related as follows,
\begin{align*}
\eta_{\mathcal{L}}^\alpha&= -\inn{\big(\parder{}{x^k}\big)}\ldots
\inn{\big(\parder{}{x^1}\big)}(\Theta_{\mathcal{L}}\wedge\d x^\alpha)\ , \\
\Theta_{\mathcal{L}}&= 
-E_{\mathcal{L}}\,\omega+\kappa_2^*\theta^\alpha\wedge\inn{\big(\parder{}{x^\alpha}\big)}\omega=
-E_{\mathcal{L}}\,\d^k t+\kappa_2^*\theta^\alpha\wedge\d^{k-1}x_\alpha \ ,
\end{align*}
and the $1$-forms $\tau^\alpha$ are canonically defined.
\end{theorem}

\section{Relation with other kinds of multicontact structures}
\label{6}

For completeness, we briefly comment on the relations between the structures studied in this paper and other similar structures present in the literature.

A  higher-dimensional version of contact distributions was proposed in \cite{Vit_15}. There, a distribution $D\subset \T\mathcal{M}$ is called multicontact if it is maximally non-integrable. Namely, the only vector field $X\in\Gamma(D)$ with the property that $[X,Y]\in \Gamma (D)$ for all $Y\in\Gamma(D)$ is $X=0$.

The relation between $k$-contact geometry and maximally non-integrable distributions has been studied in \cite{LRS_24}. Given a $k$-contact  structure on $\mathcal{M}$, consider de distribution $\mathcal{D}^{\mathrm{C}}$ defined in \ref{def:kcontact}. Then, Theorem 3.6 in \cite{LRS_24} states that $\mathcal{D}^{\mathrm{C}}$ is maximally non-integrable. The converse is not true (see Theorem 3.13 in \cite{LRS_24} for a counterexample). The hypothesis needed for a maximally non-integrable distribution to be associated with a $k$-contact structure are shown in Theorem 3.14 in \cite{LRS_24}.

We proceed to study the relation between $k$-cocontact and multicontact structures with maximally non-integrable distributions.

\begin{proposition} Let $(\mathcal{M},\eta^\alpha,\tau^\alpha)$ be a $k$-cocontact manifold. Then, $\mathcal{D}^{\mathrm{S}}\cap\mathcal{D}^{\mathrm{C}}$ is maximally non-integrable.    
\end{proposition}
\begin{proof}
    Consider $X\in\Gamma(\mathcal{D}^{\mathrm{S}}\cap\mathcal{D}^\mathrm{C})$ such that $[X,Y]\in\Gamma(\mathcal{D}^{\mathrm{S}}\cap\mathcal{D}^\mathrm{C})$ for all $Y\in\Gamma (\mathcal{D}^{\mathrm{S}}\cap\mathcal{D}^\mathrm{C})$. Then, 
    \begin{align}
               0=\inn{[X,Y]}\eta^\alpha=-\inn{Y}\inn{X}\d\eta^\alpha\,.
    \end{align}
   By a dimensional-counting argument,
    \begin{align}\label{eq:kcocoMnI}
        (\mathcal{D}^\mathrm{S}\cap \mathcal{D}^\mathrm{C})\oplus \mathcal{D}^\mathrm{R}=\T\mathcal{M}\,.
    \end{align}
    Since $\inn{R}\inn{X}\d \eta^\alpha=0$, for every $R\in\Gamma(\mathcal{D}^\mathrm{R})$,  we conclude that $\inn{X}\d\eta^\alpha$ vanishes by the action of all tangent vector fields.  
    Therefore, $\inn{X}\d\eta^\alpha=0$, for every $\alpha$ and $X\in\Gamma (\mathcal{D}^\mathrm{R})$. In light of \eqref{eq:kcocoMnI}, it must be $X=0$.
\end{proof}


\begin{proposition} Let $(\mathcal{M},\Theta,\omega)$ be a variational multicontact manifold. Then, $\ker\Theta$ is maximally non-integrable.    
\end{proposition}
\begin{proof}
 Consider $X\in\Gamma(\ker\Theta)$ such that $[X,Y]\in\Gamma(\ker\Theta)$ for all $Y\in\Gamma( \ker\Theta)$. Then, 
    \begin{align}
               0=\inn{[X,Y]}\Theta=-\inn{Y}\inn{X}\d\Theta\,.
    \end{align}
  Moreover, if $R\in \Gamma(\mathcal{D}^\mathfrak{R})$, then
$\inn{R}\inn{X}\d\Theta=0$ because $X\in\ker \Theta\subset\ker\omega$. Therefore, $\inn{X}\d \Theta$ vanishes by the action of the elements of $\ker \Theta$ and $\mathcal{D}^\mathfrak{R}$, which, by Lemma \ref{lem:multicoMit}, are all the elements of $\ker \omega$. Consequently, $X\in\Gamma(\mathcal{D}^{\mathfrak{R}})$. Since we also have $X\in\Gamma(\ker\Theta)$, then $X=0$ by the third condition of special multicontact structures \ref{multicontactdef} and lemma \ref{lem:multicoMit}.
\end{proof}

A non-variational special multicontact structure does not lead, in general, to a maximally non-integrable distribution. For instance, consider $\mathcal{M}=\mathbb{R}^7$ with the globally defined coordinates $(x,y,q,p^x,p^y,s^x,s^y)$, and the special multicontact structure
\begin{align}
    \omega=\d x\wedge\d y\,,\qquad \Theta=(\d s^x-p^x\d q)\wedge\d y-(\d s^y-p^y\d q)\wedge \d x+\d p^x\wedge \d q\,.
\end{align}
In this case, $\ds\ker \Theta=\ker \omega\,\cap\,\ker \Theta=\Big\langle\frac{\partial}{\partial p^y}\Big\rangle$, that is not maximally non-integrable.

\section{Conclusions and outlook}
\label{concl}

The covariant geometrical description of conservative classical field theories is well-established, relying on finite-dimensional structures such as $k$-symplectic, $k$-cosymplectic, and multisymplectic geometries, among others, which extend traditional symplectic geometry. For action-dependent (non-conservative) field theories, analogous frameworks, inspired by the aforementioned structures, are provided by $k$-contact, $k$-cocontact, and multicontact structures, extending contact geometry.

This work begins by reviewing the definitions and core features of these new geometric structures, including $k$-contact, $k$-cocontact, and multicontact (distinguishing different types of the latter), and Hamiltonian systems on each case. This foundation enables a finite-dimensional covariant geometrical description of action-dependent regular classical field theories in both Lagrangian and Hamiltonian formalisms.

The main paper's novel contribution has been the study of the relationships among these types of structure, specifically when the phase bundles of the field theories are trivial. We have also compared these structures with existing alternative definitions of multicontact structures.

This research opens avenues for future extension to singular (almost-regular) action-dependent classical field theories.

\appendix
\section{Multivector fields}
\label{append}

(See \cite{BCGGG_91,CIL_96,EMR_98,KMS_93,Mar_97} for more details).

Let $\mathcal{M}$ be an $N$-dimensional differentiable manifold.
The {\sl \textbf{$k$-multivector fields}} in $\mathcal{M}$ ($k\leq N$)
are the sections of the $k$-multitangent bundle
$\displaystyle\bigwedge^k\Tan\mathcal{M}:=\overbrace{\Tan\mathcal{M}\wedge\dotsb\wedge\Tan\mathcal{M}}^k$; that is,
the skew-symmetric contravariant tensor fields of order $k$ in $\mathcal{M}$.
The set $\vf^k(\mathcal{M})$ consists of all these multivector fields.
In particular, $\mathbf{X}\in\vf^k(\mathcal{M})$ is a
{\sl \textbf{locally decomposable multivector field}} if
there exist $X_1,\ldots ,X_k\in\vf (U)$ such that $\mathbf{X}\vert_U=X_1\wedge\dotsb\wedge X_k$.

Locally decomposable $k$-multivector fields are locally associated with $k$-dimensional
distributions $D\subset\Tan \mathcal{M}$, and this splits
$\vf^k(\mathcal{M})$ into {\sl equivalence classes} $\{ {\bf X}\}\subset\vf^k(\mathcal{M})$ 
which are made of the locally decomposable multivector fields associated with the same distribution.
If ${\bf X},{\bf X}'\in\{ {\bf X}\}$ then, for $U\subset \mathcal{M}$,
there exists a non-vanishing function $f\in \Cinfty(U)$ such that 
${\bf X}'=f{\bf X}$ on $U$.
In particular, an {\sl\textbf{integrable multivector field}}
is a locally decomposable multivector field whose associated distribution is integrable; that is, involutive. 

If $\Omega\in\df^k(\mathcal{M})$ and $\mathbf{X}\in\mathfrak{X}^k(\mathcal{M})$,
the {\sl \textbf{contraction}} between ${\bf X}$ and $\Omega$ is
the natural contraction between tensor fields; in particular, it gives zero when $m<k$ and, if $m\geq k$, and for locally decomposable multivector fields is
$$
 \inn{\bf X}\Omega\mid_{U}:= \inn{(X_1\wedge\dotsb\wedge X_k)}\Omega 
=\inn{X_k}\ldots\inn{X_1}\Omega \,.
$$

Now, let $\varrho\colon\mathcal{M}\to M$ be a fiber bundle,
where $M$ is an oriented manifold with volume form $\eta\in\df^k(M)$.
A multivector field $\mathbf{X}\in\mathfrak{X}^k(\mathcal{M})$ 
is {\sl $\varrho$-transverse} if,
for every $\beta\in\Omega^k(M)$ such that
$\beta_{\varrho({\rm p})}\not= 0$,
at every point ${\rm p}\in \mathcal{M}$,
we have that
$(\inn{\mathbf{X}}(\varrho^*\beta))_{{\rm p}}\not= 0$.
Then, if $\mathbf{X}\in\mathfrak{X}^k(\mathcal{M})$ is
integrable and $\varrho$-transverse, 
its integral manifolds are local sections of the projection
$\varrho\colon \mathcal{M}\to M$.

Therefore, the {\sl \textbf{canonical prolongation}} of a section  $\bm\psi\colon U\subset M\to \mathcal{M}$ to 
$\Lambda^k\Tan{\cal M}$ is the section $\bm\psi^{(k)}\colon U\subset M\to\Lambda^k\Tan\mathcal{M}$ 
defined as $\bm\psi^{(k)}:=\Lambda^k\Tan\psi\circ{\bf Y}_\eta$; where 
$\Lambda^m\Tan\psi\colon\Lambda^k\Tan M\to\Lambda^k\Tan\mathcal{M}$ is the natural extension of $\bm\psi$ 
to the corresponding multitangent bundles,
and ${\bf Y}_\eta\in\vf^k(M)$ is the unique $k$-multivector field on $M$
such that $\inn{{\bf Y}_\eta}\omega=1$. Then,
$\bm\psi$ is an {\sl\textbf{integral section}} of ${\bf X}\in\vf^m({\cal M})$ if, and only if, ${\bf X}\circ\bm\psi=\bm\psi^{(m)}$.

Thus, if $(U;z^i)$ is a chart of coordinates in ${\cal M}$,
and $x^\alpha$ are the coordinates in $\R^k$;
given a decomposable multivector field $\bfX = X_1\wedge\dotsb\wedge X_k$, with local expression
$\displaystyle X_\alpha = X_\alpha^i\parder{}{z^i}$,
a section $\bm\psi\colon M\to {\cal M}$, with $\bm\psi=(\psi^\alpha)$, is an {\sl\textbf{integral map}} of $\bfX$ if it satisfies the set of partial differential equations
\begin{equation}\label{eq:pde-multi}
\parder{\psi^i}{x^\alpha} = X_\alpha^i\circ\psi\,.
\end{equation}
In addition, the stronger integrability condition $[X_\alpha, X_\beta] = 0$, for every $\alpha,\beta = 1,\dotsc,k$, must be imposed;
and this is precisely the integrability condition of the PDE
\eqref{eq:pde-multi} (see \cite{Lee_12}).

\addcontentsline{toc}{section}{Acknowledgements}
\section*{Acknowledgements}

We acknowledge the financial support of the 
{\sl Ministerio de Ciencia, Innovaci\'on y Universidades} (Spain), projects PID2021-125515NB-C21, and RED2022-134301-T of AEI,
and Ministry of Research and Universities of
the Catalan Government, project 2021 SGR 00603 \textsl{Geometry of Manifolds and Applications, GEOMVAP}.

This work is dedicated to the memory of Prof. Miguel C. Muñoz-Lecanda, who actively contributed to the creation and development of the geometric structures presented herein.


\bibliographystyle{abbrv}

{\small 
\bibliography{references.bib}

\begin{thebibliography}{10}

\bibitem{AM_78}
R.~Abraham and J.~E. Marsden.
\newblock {\em {Foundations of mechanics}}, volume 364 of {\em AMS Chelsea
  publishing}.
\newblock Benjamin/Cummings Pub. Co., New York, 2nd edition, 1978.
\newblock \href{https://doi.org/10.1090/chel/364}{10.1090/chel/364}.

\bibitem{AA_80}
V.~Aldaya and J.~A. de~Azcárraga.
\newblock {Geometric formulation of classical mechanics and field theory}.
\newblock {\em Riv. Nuovo Cim.}, {\bf 3}(1):1--66, 1980.
\newblock \href{https://doi.org/10.1007/BF02906204}{10.1007/BF02906204}.

\bibitem{Arn_89}
V.~I. Arnold.
\newblock {\em {Mathematical Methods of Classical Mechanics}}, volume~60 of
  {\em Graduate Texts in Mathematics}.
\newblock Springer, New York, 2nd edition, 1989.
\newblock
  \href{https://doi.org/10.1007/978-1-4757-1693-1}{10.1007/978-1-4757-1693-1}.

\bibitem{Awa_92}
A.~Awane.
\newblock {$k$-symplectic structures}.
\newblock {\em J. Math. Phys.}, {\bf 33}(12):4046, 1992.
\newblock \href{https://doi.org/10.1063/1.529855}{10.1063/1.529855}.

\bibitem{Awa_94}
A.~Awane.
\newblock {G-espaces K-symplectiques homogènes}.
\newblock {\em J. Geom. Phys.}, {\bf 13}(2):139--157, 1994.
\newblock
  \href{https://doi.org/10.1016/0393-0440(94)90024-8}{10.1016/0393-0440(94)90024-8}.

\bibitem{AG_00}
A.~Awane and M.~Goze.
\newblock {\em {Pfaffian systems, $k$-symplectic systems}}.
\newblock Springer, Dordrecht, 1st edition, 2000.
\newblock
  \href{https://doi.org/10.1007/978-94-015-9526-1}{10.1007/978-94-015-9526-1}.

\bibitem{BH_16}
A.~Banyaga and D.~F. Houenou.
\newblock {\em {A brief introduction to symplectic and contact manifolds}},
  volume~15 of {\em Nankai Tracts in Mathematics}.
\newblock World Scientific Publishing Co. Pte. Ltd., Singapore, 2016.
\newblock \href{https://doi.org/10.1142/9667}{10.1142/9667}.

\bibitem{Bo_96}
P.~Bolle.
\newblock {Une condition de contact pour les sous-vari\'et\'es coısotropes
  d’une vari\'et\'e symplectique}.
\newblock {\em C. R. Math. Acad. Sci. S\'er. 1}, {\bf 1}:83--96, 1996.

\bibitem{Bra_17}
A.~Bravetti.
\newblock {Contact Hamiltonian dynamics: The concept and its use}.
\newblock {\em Entropy}, {\bf 10}(19):535, 2017.
\newblock \href{https://doi.org/10.3390/e19100535}{10.3390/e19100535}.

\bibitem{Bra_18}
A.~Bravetti.
\newblock {Contact geometry and thermodynamics}.
\newblock {\em Int. J. Geom. Methods Mod. Phys.}, {\bf 16}(supp01):1940003,
  2018.
\newblock
  \href{https://doi.org/10.1142/S0219887819400036}{10.1142/S0219887819400036}.

\bibitem{BCT_17}
A.~Bravetti, H.~Cruz, and D.~Tapias.
\newblock {Contact Hamiltonian mechanics}.
\newblock {\em Ann. Phys.}, {\bf 376}:17--39, 2017.
\newblock
  \href{https://doi.org/10.1016/j.aop.2016.11.003}{10.1016/j.aop.2016.11.003}.

\bibitem{BLMP_20}
A.~Bravetti, M.~de~León, J.~C. Marrero, and E.~Padrón.
\newblock {Invariant measures for contact Hamiltonian systems: symplectic
  sandwiches with contact bread}.
\newblock {\em J. Phys. A: Math. Theor.}, {\bf 53}:455205, 2020.
\newblock
  \href{https://doi.org/10.1088/1751-8121/abbaaa}{10.1088/1751-8121/abbaaa}.

\bibitem{BCGGG_91}
R.~L. Bryant, S.~S. Chern, R.~B. Gardner, H.~L. Goldshmidt, and P.~A. Griffith.
\newblock {\em {Exterior Differential Systems}}, volume~18 of {\em Mathematical
  Sciences Research Institute Publications}.
\newblock Springer-Verlag, New York, 1st edition, 1991.
\newblock
  \href{https://doi.org/10.1007/978-1-4613-9714-4}{10.1007/978-1-4613-9714-4}.

\bibitem{CIL_96}
F.~Cantrijn, A.~Ibort, and M.~de~León.
\newblock {Hamiltonian structures on multisymplectic manifolds}.
\newblock {\em Rend. Sem. Mat. Univ. Pol. Torino}, {\bf 54}(3):225--236, 1996.
\newblock
  \url{http://www.seminariomatematico.polito.it/rendiconti/cartaceo/54-3/225.pdf}.

\bibitem{CCI_91}
J.~F. Cariñena, M.~Crampin, and L.~A. Ibort.
\newblock {On the multisymplectic formalism for first order field theories}.
\newblock {\em Diff. Geom. Appl.}, {\bf 1}(4):345--374, 1991.
\newblock
  \href{https://doi.org/10.1016/0926-2245(91)90013-Y}{10.1016/0926-2245(91)90013-Y}.

\bibitem{CG_19}
J.~F. Cariñena and P.~Guha.
\newblock {Nonstandard Hamiltonian structures of the Liénard equation and
  contact geometry}.
\newblock {\em Int. J. Geom. Methods Mod. Phys.}, {\bf 16}(supp01):1940001,
  2019.
\newblock
  \href{https://doi.org/10.1142/S0219887819400012}{10.1142/S0219887819400012}.

\bibitem{CCM_18}
F.~M. Ciaglia, H.~Cruz, and G.~Marmo.
\newblock {Contact manifolds and dissipation, classical and quantum}.
\newblock {\em Ann. Phys.}, {\bf 398}:159--179, 2018.
\newblock
  \href{https://doi.org/10.1016/j.aop.2018.09.012}{10.1016/j.aop.2018.09.012}.

\bibitem{LMNRS-2002}
M.~de~Le\'on, M.~McLean, L.~Norris, A.~Rey-Roca, and M.~Salgado.
\newblock {Geometric structures in field theory}.
\newblock \href{https://arxiv.org/abs/math-ph/0208036}{arXiv:math-ph/0208036},
  2002.

\bibitem{LGGMR_23}
M.~de~León, J.~Gaset, X.~Gràcia, M.~C. Muñoz-Lecanda, and X.~Rivas.
\newblock {Time-dependent contact mechanics}.
\newblock {\em Monatsh. Math.}, {\bf 201}:1149--1183, 2023.
\newblock
  \href{https://doi.org/10.1007/s00605-022-01767-1}{10.1007/s00605-022-01767-1}.

\bibitem{LGMRR_23}
M.~de~León, J.~Gaset, M.~C. Muñoz-Lecanda, X.~Rivas, and N.~Román-Roy.
\newblock {Multicontact formulation for non-conservative field theories}.
\newblock {\em J. Phys. A: Math. Theor.}, {\bf 56}(2):025201, 2023.
\newblock
  \href{https://doi.org/10.1088/1751-8121/acb575}{10.1088/1751-8121/acb575}.

\bibitem{LGMRR_25}
M.~de~León, J.~Gaset, M.~C. Muñoz-Lecanda, X.~Rivas, and N.~Román-Roy.
\newblock {Practical introduction to action-dependent field theories}.
\newblock {\em Fortschr. Phys.}, {\bf 7}(5), 2025.
\newblock \href{https://doi.org/10.1002/prop.70000}{10.1002/prop.70000}.

\bibitem{LIR_25}
M.~de~León, R.~Izquierdo-López, and X.~Rivas.
\newblock {Brackets in multicontact geometry and multisymplectization}.
\newblock \href{https://arxiv.org/abs/2505.13224}{2505.13224}, 2025.

\bibitem{LL_19a}
M.~de~León and M.~Lainz-Valcázar.
\newblock {Singular Lagrangians and precontact Hamiltonian systems}.
\newblock {\em Int. J. Geom. Methods Mod. Phys.}, {\bf 16}(10):1950158, 2019.
\newblock
  \href{https://doi.org/10.1142/S0219887819501585}{10.1142/S0219887819501585}.

\bibitem{LL_19}
M.~de~León and M.~Laínz-Valcázar.
\newblock {Contact Hamiltonian systems}.
\newblock {\em J. Math. Phys.}, {\bf 60}(10):102902, 2019.
\newblock \href{https://doi.org/10.1063/1.5096475}{10.1063/1.5096475}.

\bibitem{LMS_03}
M.~de~León, D.~{Martín de Diego}, and A.~Santamaría-Merino.
\newblock {Tulczyjew triples and Lagrangian submanifolds in classical field
  theories}.
\newblock In {\em Applied Differential Geometry and Mechanic}, pages 21--47.
  Univ. of Gent, Academia Press, Gent, 2003.
\newblock \href{https://arxiv.org/abs/math-ph/0302026}{arxiv:math-ph/0302026}.

\bibitem{LMM_96a}
M.~de~León, J.~Marín-Solano, and J.~C. Marrero.
\newblock {\em {A geometrical approach to Classical Field Theories: a
  constraint algorithm for singular theories}}, pages 291--312.
\newblock Springer Netherlands, Dordrecht, 1996.
\newblock
  \href{https://doi.org/10.1007/978-94-009-0149-0_22}{10.1007/978-94-009-0149-0\_22}.

\bibitem{LMORS_98}
M.~de~León, E.~Merino, J.~A. Oubiña, P.~R. Rodrigues, and M.~Salgado.
\newblock {Hamiltonian systems on $k$-cosymplectic manifolds}.
\newblock {\em J. Math. Phys.}, {\bf 39}(2):876, 1998.
\newblock \href{https://doi.org/10.1063/1.532358}{10.1063/1.532358}.

\bibitem{LMOS_97}
M.~de~León, E.~Merino, J.~A. Oubiña, and M.~Salgado.
\newblock {Stable almost cotangent structures}.
\newblock {\em Bolletino Unione Mat. Ital. B (7)}, {\bf 11}(3):509--529, 1997.

\bibitem{LR_89}
M.~de~León and P.~R. Rodrigues.
\newblock {\em {Methods of differential geometry in analytical mechanics}},
  volume 158 of {\em Mathematics Studies}.
\newblock North-Holland, Amsterdam, 1989.
\newblock
  \href{https://doi.org/10.1016/s0304-0208(08)x7115-4}{10.1016/s0304-0208(08)x7115-4}.

\bibitem{LSV_15}
M.~de~León, M.~Salgado, and S.~Vilariño.
\newblock {\em {Methods of Differential Geometry in Classical Field Theories}}.
\newblock World Scientific, 2015.
\newblock \href{https://doi.org/10.1142/9693}{10.1142/9693}.

\bibitem{LRS_24}
J.~de~Lucas, X.~Rivas, and T.~Sobczak.
\newblock {Foundations on $k$-contact geometry}.
\newblock \href{https://arxiv.org/abs/2409.11001}{2409.11001}, 2024.

\bibitem{EMR_96}
A.~Echeverría-Enríquez, M.~C. Muñoz-Lecanda, and N.~Román-Roy.
\newblock {Geometry of Lagrangian first-order classical field theories}.
\newblock {\em Fortschr. Phys.}, {\bf 44}(3):235--280, 1996.
\newblock
  \href{https://doi.org/10.1002/prop.2190440304}{10.1002/prop.2190440304}.

\bibitem{EMR_98}
A.~Echeverría-Enríquez, M.~C. Muñoz-Lecanda, and N.~Román-Roy.
\newblock {Multivector fields and connections: Setting Lagrangian equations in
  field theories}.
\newblock {\em J. Math. Phys.}, {\bf 39}(9):4578--4603, 1998.
\newblock \href{https://doi.org/10.1063/1.532525}{10.1063/1.532525}.

\bibitem{EMR_2000}
A.~Echeverría-Enríquez, M.~C. Muñoz-Lecanda, and N.~Román-Roy.
\newblock {Geometry of multisymplectic Hamiltonian first-order field theories}.
\newblock {\em J. Math. Phys.}, {\bf 41}(11):7402--7444, 2000.
\newblock \href{https://doi.org/10.1063/1.1308075}{10.1063/1.1308075}.

\bibitem{Fi_24}
D.~Finamore.
\newblock {Contact foliations and generalised Weinstein conjectures}.
\newblock {\em Ann. Glob. Anal. Geom.}, {\bf 65}(27), 2024.
\newblock
  \href{https://doi.org/10.1007/s10455-024-09957-w}{10.1007/s10455-024-09957-w}.

\bibitem{FPR_05}
M.~Forger, C.~Paufler, and H.~Römer.
\newblock {Hamiltonian multivector fields and Poisson forms in multisymplectic
  field theory}.
\newblock {\em J. Math, Phys.}, {\bf 46}(11):112903, 2005.
\newblock \href{https://doi.org/10.1063/1.2116320}{10.1063/1.2116320}.

\bibitem{Gar_74}
P.~L. García.
\newblock {The Poincaré--Cartan invariant in the calculus of variations}.
\newblock In {\em Symposia Mathematiea}, volume~14, pages 219--246, London,
  1974. (Convegno di Geometria Simplettica e Fisica Matematica, INDAM, Rome,
  1973), Academic Press.

\bibitem{GGMRR_20}
J.~Gaset, X.~Gràcia, M.~C. Muñoz-Lecanda, X.~Rivas, and N.~Román-Roy.
\newblock {A contact geometry framework for field theories with dissipation}.
\newblock {\em Ann. Phys.}, {\bf 414}:168092, 2020.
\newblock
  \href{https://doi.org/10.1016/j.aop.2020.168092}{10.1016/j.aop.2020.168092}.

\bibitem{GGMRR_20a}
J.~Gaset, X.~Gràcia, M.~C. Muñoz-Lecanda, X.~Rivas, and N.~Román-Roy.
\newblock {New contributions to the Hamiltonian and Lagrangian contact
  formalisms for dissipative mechanical systems and their symmetries}.
\newblock {\em Int. J. Geom. Methods Mod. Phys.}, {\bf 17}(6):2050090, 2020.
\newblock
  \href{https://doi.org/10.1142/S0219887820500905}{10.1142/S0219887820500905}.

\bibitem{GGMRR_21}
J.~Gaset, X.~Gràcia, M.~C. Muñoz-Lecanda, X.~Rivas, and N.~Román-Roy.
\newblock {A $k$-contact Lagrangian formulation for nonconservative field
  theories}.
\newblock {\em Rep. Math. Phys.}, {\bf 87}(3):347--368, 2021.
\newblock
  \href{https://doi.org/10.1016/S0034-4877(21)00041-0}{10.1016/S0034-4877(21)00041-0}.

\bibitem{GLMR_24}
J.~Gaset, M.~Lainz, A.~Mas, and X.~Rivas.
\newblock {The Herglotz variational principle for dissipative field theories}.
\newblock {\em Geom. Mech.}, {\bf 1}(2):153--178, 2024.
\newblock
  \href{https://doi.org/10.1142/S2972458924500060}{10.1142/S2972458924500060}.

\bibitem{GLR_23}
J.~Gaset, A.~López-Gordón, and X.~Rivas.
\newblock {Symmetries, conservation and dissipation in time-dependent contact
  systems}.
\newblock {\em Fortschr. Phys.}, {\bf 71}(8--9):2300048, 2023.
\newblock
  \href{https://doi.org/10.1002/prop.202300048}{10.1002/prop.202300048}.

\bibitem{GM_23a}
J.~Gaset and A.~Mas.
\newblock {A variational derivation of the field equations of an
  action-dependent Einstein--Hilbert Lagrangian}.
\newblock {\em J. Geom. Mech.}, {\bf 15}(1):357--374, 2023.
\newblock \href{https://doi.org/10.3934/jgm.2023014}{10.3934/jgm.2023014}.

\bibitem{Gei_08}
H.~Geiges.
\newblock {\em {An introduction to contact topology}}, volume 109 of {\em
  Cambridge Studies in Advanced Mathematics}.
\newblock Cambridge University Press, New York, NY, 2008.
\newblock
  \href{https://doi.org/10.1017/CBO9780511611438}{10.1017/CBO9780511611438}.

\bibitem{GGB_03}
B.~Georgieva, R.~Gunther, and T.~Bodurov.
\newblock {Generalized variational principle of Herglotz for several
  independent variables. First Noether-type theorem}.
\newblock {\em J. Math. Phys.}, {\bf 44}(9):3911, 2003.
\newblock \href{https://doi:10.1063/1.1597419}{10.1063/1.1597419}.

\bibitem{GMS_97}
G.~Giachetta, L.~Mangiarotti, and G.~A. Sardanashvily.
\newblock {\em {New Lagrangian and Hamiltonian Methods in Field Theory}}.
\newblock World Scientific, River Edge, 1997.
\newblock \href{https://doi.org/10.1142/2199}{10.1142/2199}.

\bibitem{GS_73}
H.~Goldschmidt and S.~Sternberg.
\newblock {The Hamilton--Cartan formalism in the calculus of variations}.
\newblock {\em Ann. Inst. Fourier}, {\bf 23}(1):203--267, 1973.
\newblock \href{https://doi.org/10.5802/aif.451}{10.5802/aif.451}.

\bibitem{GIMM_98}
M.~J. Gotay, J.~Isenberg, J.~E. Marsden, and R.~Montgomery.
\newblock {Momentum Maps and Classical Relativistic Fields. Part I: Covariant
  Field Theory}.
\newblock \href{https://arxiv.org/abs/physics/9801019}{arxiv:physics/9801019},
  1998.

\bibitem{GG_22}
K.~Grabowska and J.~Grabowski.
\newblock {A geometric approach to contact Hamiltonians and contact
  Hamilton--Jacobi theory}.
\newblock {\em J. Phys. A: Math. Theor.}, {\bf 55}(43):435204, 2022.
\newblock
  \href{https://doi.org/10.1088/1751-8121/ac9adb}{10.1088/1751-8121/ac9adb}.

\bibitem{GRR_22}
X.~Gràcia, X.~Rivas, and N.~Román-Roy.
\newblock {Skinner--Rusk formalism for $k$-contact systems}.
\newblock {\em J. Geom. Phys.}, {\bf 172}:104429, 2022.
\newblock
  \href{https://doi.org/10.1016/j.geomphys.2021.104429}{10.1016/j.geomphys.2021.104429}.

\bibitem{HK_02}
F.~Helein and J.~Kouneiher.
\newblock {Finite dimensional Hamiltonian formalism for gauge and quantum field
  theories}.
\newblock {\em J. Math. Phys.}, {\bf 43}(5):2306--2347, 2002.
\newblock \href{https://doi.org/10.1063/1.1467710}{10.1063/1.1467710}.

\bibitem{Her_30}
G.~Herglotz.
\newblock {\em {Berührungstransformationen}}.
\newblock Lectures at the University of Göttingen, 1930.

\bibitem{Her_85}
G.~Herglotz.
\newblock {\em {Vorlesungen \"uber die Mechanik der Kontinua}}, volume~3 of
  {\em Teubner-Archiv zur Mathematik}.
\newblock Teubner, Leipzig, 1985.
\newblock
  \href{https://doi.org/10.1007/978-3-7091-9510-9}{10.1007/978-3-7091-9510-9}.

\bibitem{Holm_11}
D.~D. Holm.
\newblock {\em {Geometric Mechanics. Part I}}.
\newblock Imperial College Press, World Scientific, London, 2011.
\newblock \href{https://doi.org/10.1142/p801}{10.1142/p801}.

\bibitem{Kan_98}
I.~V. Kanatchikov.
\newblock {Canonical structure of classical field theory in the polymomentum
  phase space}.
\newblock {\em Rep. Math. Phys.}, {\bf 41}(1):49--90, 1998.
\newblock
  \href{https://doi.org/10.1016/S0034-4877(98)80182-1}{10.1016/S0034-4877(98)80182-1}.

\bibitem{Kho_13}
A.~L. Kholodenko.
\newblock {\em {Applications of contact geometry and topology in physics}}.
\newblock World Scientific, Singapore, 2013.
\newblock \href{https://doi.org/10.1142/8514}{10.1142/8514}.

\bibitem{Kij_73}
J.~Kijowski.
\newblock {A finite-dimensional canonical formalism in the classical field
  theory}.
\newblock {\em Comm. Math. Phys.}, {\bf 30}(2):99--128, 1973.
\newblock \href{https://doi.org/10.1007/BF01645975}{10.1007/BF01645975}.

\bibitem{KT_79}
J.~Kijowski and W.~M. Tulczyjew.
\newblock {\em {A symplectic framework for field theories}}, volume 107 of {\em
  Lecture Notes in Physics}.
\newblock Springer-Verlag, Berlin Heidelberg, 1st edition, 1979.
\newblock \href{https://doi.org/10.1007/3-540-09538-1}{10.1007/3-540-09538-1}.

\bibitem{KMS_93}
I.~Kolář, P.~W. Michor, and J.~Slovák.
\newblock {\em {Natural operations in differential geometry}}.
\newblock Springer Berlin, Heidelberg, 1993.
\newblock
  \href{https://doi.org/10.1007/978-3-662-02950-3}{10.1007/978-3-662-02950-3}.

\bibitem{LPAF_17}
M.~Lazo, J.~Paiva, J.~Amaral, and G.~Frederico.
\newblock {Action principle for action-dependent Lagrangians toward
  nonconservative gravity: Accelerating universe without dark energy}.
\newblock {\em Phys. Rev. D}, {\bf 95}(10):101501, 2017.
\newblock
  \href{https://doi.org/10.1103/PhysRevD.95.101501}{10.1103/PhysRevD.95.101501}.

\bibitem{Lee_12}
J.~M. Lee.
\newblock {\em {Introduction to Smooth Manifolds}}, volume 218 of {\em Graduate
  Texts in Mathematics}.
\newblock Springer New York Heidelberg Dordrecht London, 2nd edition, 2012.
\newblock
  \href{http://doi.org/10.1007/978-1-4419-9982-5}{10.1007/978-1-4419-9982-5}.

\bibitem{LM_87}
P.~Libermann and C.-M. Marle.
\newblock {\em {Symplectic geometry and analytical mechanics}}, volume~35 of
  {\em Mathematics and Its Applications}.
\newblock Springer Dordrecht, 1987.
\newblock
  \href{http://doi.org/10.1007/978-94-009-3807-6}{10.1007/978-94-009-3807-6}.

\bibitem{Mar_97}
C.-M. Marle.
\newblock {The Schouten--Nijenhuis bracket and interior products}.
\newblock {\em J. Geom. Phys.}, {\bf 23}(3--4):350--359, 1997.
\newblock
  \href{https://doi.org/10.1016/S0393-0440(97)80009-5}{10.1016/S0393-0440(97)80009-5}.

\bibitem{Mar_88}
G.~Martin.
\newblock {A Darboux theorem for multi-symplectic manifolds}.
\newblock {\em Lett. Math. Phys.}, {\bf 16}:133--138, 1988.
\newblock \href{https://doi.org/10.1007/BF00402020}{10.1007/BF00402020}.

\bibitem{Mo_08}
B.~Montano.
\newblock {Integral submanifolds of $r$-contact manifolds}.
\newblock {\em Demonstr. Math.}, {\bf 41}(1):189--202, 2015.
\newblock
  \href{https://doi.org/10.1515/dema-2013-0054}{10.1515/dema-2013-0054}.

\bibitem{MR_25}
M.~C. Muñoz-Lecanda and N.~Rom\'an-Roy.
\newblock {\em {Geometry of Mechanics}}.
\newblock World Scientific, London, 2025.
\newblock \href{https://doi.org/10.1142/q0490}{10.1142/q0490}.

\bibitem{RRSV_11}
A.~M. Rey, N.~Román-Roy, M.~Salgado, and S.~Vilariño.
\newblock {On the $k$-symplectic, $k$-cosymplectic and multisymplectic
  formalisms of classical field theories}.
\newblock {\em J. Geom. Mech.}, {\bf 3}(1):113--137, 2011.
\newblock
  \href{https://doi.org/10.3934/jgm.2011.3.113}{10.3934/jgm.2011.3.113}.

\bibitem{Riv_23}
X.~Rivas.
\newblock {Nonautonomous $k$-contact field theories}.
\newblock {\em J. Math. Phys.}, {\bf 64}(3):033507, 2023.
\newblock \href{https://doi.org/10.1063/5.0131110}{10.1063/5.0131110}.

\bibitem{RRZ_25}
X.~Rivas, N.~Román-Roy, and B.~M. Zawora.
\newblock {Symmetries and Noether's theorem for multicontact field theories}.
\newblock \href{https://arxiv.org/abs/2505.13224}{2505.13224}, 2025.

\bibitem{RT_23}
X.~Rivas and D.~Torres.
\newblock {Lagrangian--Hamiltonian formalism for cocontact systems}.
\newblock {\em J. Geom. Mech.}, {\bf 15}(1):1--26, 2023.
\newblock \href{https://doi.org/10.3934/jgm.2023001}{10.3934/jgm.2023001}.

\bibitem{Rom_09}
N.~Román-Roy.
\newblock {Multisymplectic Lagrangian and Hamiltonian formalisms of classical
  field theories}.
\newblock {\em Symmetry Integr. Geom.: Methods Appl. (SIGMA)}, {\bf 5}(100),
  2009.
\newblock
  \href{https://doi.org/10.3842/SIGMA.2009.100}{10.3842/SIGMA.2009.100}.

\bibitem{RW_19}
L.~Ryvkin and T.~Wurzbacher.
\newblock {An invitation to multisymplectic geometry}.
\newblock {\em J. Geom. Phys.}, {\bf 142}:9--36, 2019.
\newblock
  \href{https://doi.org/10.1016/j.geomphys.2019.03.006}{10.1016/j.geomphys.2019.03.006}.

\bibitem{Sau_89}
D.~J. Saunders.
\newblock {\em {The geometry of jet bundles}}, volume 142 of {\em London
  Mathematical Society Lecture Note Series}.
\newblock Cambridge University Press, 1989.
\newblock
  \href{https://doi.org/10.1017/CBO9780511526411}{10.1017/CBO9780511526411}.

\bibitem{TV_08}
A.~Tomassini and L.~Vezzoni.
\newblock {Contact Calabi--Yau manifolds and special Legendrian submanifolds}.
\newblock {\em Osaka J. Math.}, {\bf 45}(1):127--147, 2008.

\bibitem{Vit_15}
L.~Vitagliano.
\newblock {$L_\infty$-algebras from multicontact geometry}.
\newblock {\em Diff. Geom. Appl.}, {\bf 39}:147--165, 2015.
\newblock
  \href{https://doi.org/10.1016/j.difgeo.2015.01.006}{10.1016/j.difgeo.2015.01.006}.

\end{thebibliography}
}

\end{document}